\documentclass[12pt]{article}

\usepackage{graphicx}
\usepackage{pdfpages}
\usepackage{subfig}
\usepackage{adjustbox}
\usepackage{amsmath,amsfonts,amssymb,amsthm}
\usepackage{caption}
\usepackage[ruled,vlined]{algorithm2e}
\usepackage{bbm}
\usepackage{breakcites}
\usepackage{tabularx}
\usepackage{geometry}
\usepackage{enumitem}
\usepackage{authblk}
\usepackage{graphicx}
\usepackage{natbib}

\usepackage{dirtytalk} % For quotations
\usepackage[normalem]{ulem} % For strikeouts

\usepackage{hyperref}
\usepackage{float}
\usepackage{aliascnt}
\newaliascnt{eqfloat}{equation}
\newfloat{eqfloat}{h}{eqflts}
\floatname{eqfloat}{Equation}

\newcommand*{\ORGeqfloat}{}
\let\ORGeqfloat\eqfloat
\def\eqfloat{%
  \let\ORIGINALcaption\caption
  \def\caption{%
    \addtocounter{equation}{-1}%
    \ORIGINALcaption
  }%
  \ORGeqfloat
}

\setlist[itemize]{align=parleft,left=0pt..1em}
\newcolumntype{L}{>{\raggedright\arraybackslash}X}

\usepackage{stix}
\usepackage{color}
\usepackage{tikz}
\usetikzlibrary{shapes,decorations,arrows,calc,arrows.meta,fit,positioning}
\tikzset{
    -Latex,auto,node distance =1 cm and 1 cm,semithick,
    state/.style ={ellipse, draw, minimum width = 0.7 cm},
    point/.style = {circle, draw, inner sep=0.04cm,fill,node contents={}},
    bidirected/.style={Latex-Latex,dashed},
    el/.style = {inner sep=2pt, align=left, sloped}
}

\theoremstyle{definition}
\newtheorem{definition}{Definition}[section]
\newtheorem{assumption}{Assumption}
\newtheorem{theorem}{Theorem}[section]

\newtheorem{lemma}[theorem]{Lemma}
\newtheorem{remark}{\textbf{Remark}}

\providecommand{\keywords}[1]
{
  \small	
  \textbf{{Keywords:}} #1
}

\newcommand{\independent}{\perp \!\!\! \perp}
\newcommand{\indep}{\perp \!\!\! \perp}
\newcommand{\nulls}{$Y \independent X \mid Z$ }
\newcommand{\nullns}{$Y \independent X \mid Z$}

\newcommand{\data}{$(\mathbf{X}, \mathbf{Z}, \mathbf{Y})$}
\newcommand{\datas}{$(\mathbf{X}, \mathbf{Z}, \mathbf{Y})$ }

\newcommand{\iid}{\textit{iid }}

    \title{Hypothesis Testing in Adaptively Sampled Data: ART to Maximize Power Beyond \iid Sampling\footnote{Both authors contributed equally and are listed in alphabetical order. We thank Lucas Janson and Iavor Bojinov for advice and feedback.}}
    
    \author[1]{Dae Woong Ham\thanks{Email: \href{mailto:daewoongham@g.harvard.edu}{\tt daewoongham@g.harvard.edu}}}
    \author[1]{Jiaze Qiu\thanks{Email: \href{mailto:jiazeqiu@g.harvard.edu}{\tt jiazeqiu@g.harvard.edu}}}
    \affil[1]{\textit{Department of Statistics, Harvard University}}

    \date{\today}
    
\begin{document}
    \maketitle

\begin{abstract}
   Testing whether a variable of interest affects the outcome is one of the most fundamental problem in statistics and is often the main scientific question of interest. To tackle this problem, the conditional randomization test (CRT) is widely used to test the independence of variable(s) of interest ($X$) with an outcome ($Y$) holding other variable(s) ($Z$) fixed. The CRT uses randomization or \textit{design-based} inference that relies solely on the \iid sampling of $(X,Z)$ to produce exact finite-sample $p$-values that are constructed using any test statistic. We propose a new method, the \textit{adaptive randomization test} (ART), that tackles the independence problem while allowing the data to be adaptively sampled. Like the CRT, the ART relies solely on knowing the (adaptive) sampling distribution of $(X,Z)$. Although the ART allows practitioners to flexibly design and analyze adaptive experiments, the method itself does not guarantee a powerful adaptive sampling procedure. For this reason, we show substantial power gains by adaptively sampling compared to the typical \iid sampling procedure in two illustrative settings in the second half of this paper. We first showcase the ART in a particular multi-arm bandit problem known as the normal-mean model. Under this setting, we theoretically characterize the powers of both the \iid sampling procedure and the adaptive sampling procedure and empirically find that the ART can uniformly outperform the CRT that pulls all arms independently with equal probability. We also surprisingly find that the ART can be more powerful than even the CRT that uses an \say{oracle} \iid sampling procedure when the signal is relatively strong. We believe that the proposed adaptive procedure is successful because it takes arms that may initially look like ``fake'' signals due to random chance and stabilizes them closer to ``null'' signals. We additionally showcase the ART to a popular factorial survey design setting known as conjoint analysis. We find similar results through simulations and a recent application concerning the role of gender discrimination in political candidate evaluation. 
\end{abstract}


\tableofcontents
\newpage

\section{Introduction}
Independence testing is ubiquitous in statistics and often the main task of interest in variable selection problems. For example, it is used in causal inference for testing the absence of any treatment effect for various applications \citep{Bates, mypaper, CRT}. More specifically, social scientists may wonder if a political candidate's gender may affect voting behavior while controlling for all other gender related stereotypes to isolate the true effect of gender \citep{gender, genderassociation, genderassociation2}. Biologists may also be interested in the effect of a specific gene on a characteristic after holding all other genes constant \citep{bio1}.

In the independence testing problem, the main objective is to test whether a response $Y$ is statistically affected by a variable of interest $X$ while holding other variable(s) $Z$ fixed. Informally speaking, we aim to test \nullns, where $Z$ can be the empty set for an unconditional test. For the aforementioned gender example, $Y$ is voting responses, $X$ is the political candidate's gender, and $Z$ are the candidate's personality, party affiliation, etc. One way to approach this problem is the \textit{model-based} approach that uses parametric or semi-parametric methods such as regression while assuming some knowledge of $Y \mid (X,Z)$. Recently, the \textit{design-based} approach has been increasingly gaining popularity \citep{mypaper, Bates, conditional_permu} to tackle the independence testing problem. In an influential paper \citep{CRT}, the authors introduce the conditional randomization test (CRT), which uses a \textit{design-based} or the ``Model-X'' approach to perform randomization based inference. This approach assumes nothing about the $Y \mid (X,Z)$ relationship but shifts the burden on requiring knowledge of the $X \mid Z$ distribution (hence named ``Model-X''). In exchange, the CRT has exact type-1 error control while allowing the user to propose any test statistics, including those from complicated machine learning models, to increase power. We remark that if the data was collected from an experiment, then the distribution of the experimental variables $(X,Z)$ is immediately available and the CRT can be classified as a non-parametric test.  

The CRT, however, does require that $(X,Z)$ is collected independently and identically (\textit{iid}) from some distribution, which may not be always appropriate or desired. For example, large tech companies, such as Uber or Doordash, have rich experimental data that are sequentially and adaptively collected, i.e., the next treatment is sampled as a function of all of its previous history \citep{uber, industry1}. Despite this non-\iid experimental setup, the companies are interested in performing hypothesis tests on whether a certain treatment or features of their products affects the response in any way. Additionally, many practitioners may prefer an adaptive sampling procedure as it can be more effective to detect an effect since obtaining a large number of samples is often difficult and costly.

\subsection{Our Contributions and Overview}
Given this motivation, a natural direction is to weaken the \iid assumption in the ``Model-X'' randomization inference approach and allow testing adaptively collected data. Therefore, the main contribution of our paper is we allow the same ``Model-X'' randomization inference procedure under adaptively collected data, i.e., we allow the data $(X_t, Z_t)$ to be sequentially collected at time $t$ as a function of the historical values of $X_{1:(t-1)}, Z_{1:(t-1)}, Y_{1:(t-1)}$, where $X_{1:(t-1)}$ denotes the vector of $(X_1, \dots, X_{t-1})$ and $Z_{1:(t-1)}, Y_{1:(t-1)}$ is defined similarly. To the best of our knowledge, there does not exist a general randomization inference procedure that enjoys all the same benefits as that of the CRT while allowing for adaptively sampled data (see Section~\ref{subsection:related_works} for more details).  

Our contribution is useful in both the experimental stage (the focus of this paper), i.e., allowing experimenters to construct powerful adaptive sampling procedures, and the analysis stage, i.e., after the data was adaptively collected as long as the analyst knows how the data was adaptively sampled. We name our method the ART (Adaptive Randomization Test) and we remark that the validity of the ART, like the CRT, does not require any knowledge of $Y \mid (X,Z)$ and leverages the distribution of $(X,Z)$. Therefore, in an experimental setting, the ART can also be viewed as a non-parametric test. 

In Section~\ref{section:method} we formally introduce the proposed method, ART, and prove how the ART leverages the known distribution of $(X,Z)$ to produce exact finite-sample valid $p$-values for any test statistic. Although this formally allows practitioners to adaptively sample data to potentially increase power, it does not give any guidance on how to choose a reasonable adaptive procedure. Consequently, we first showcase the ART in the normal-means model setting (Section~\ref{sec:normal_means_model}), a special case of the ``multi-arm'' bandit setting, through simulations and a theoretical asymptotic power analysis. Secondly, we also explore the ART's potential in a factorial survey setting in Section~\ref{section:conjoint_studies} through simulations and a recent conjoint application concerning the role of gender discrimination in political candidate evaluation \citep{gender}. For both examples, we find that the ART can be uniformly more powerful than the CRT with a typical \iid sampling scheme. 

To give a preview of this power gain, we show in Figure~\ref{fig:preview} the power of the ART with a specific adaptive sampling procedure compared to that of the CRT with a typical \iid uniform sampling procedure. Section~\ref{subsec:normal_means_model_takeways} and Section~\ref{subsection:sim_results} contains the full details of these plots. Figure~\ref{fig:preview} previews how the ART can be more powerful by up to 15 percentage points than the CRT in both the normal-means and conjoint settings when using a na\"ive adaptive procedure. 

\begin{figure}[t]
\begin{center}
\includegraphics[width=\textwidth]{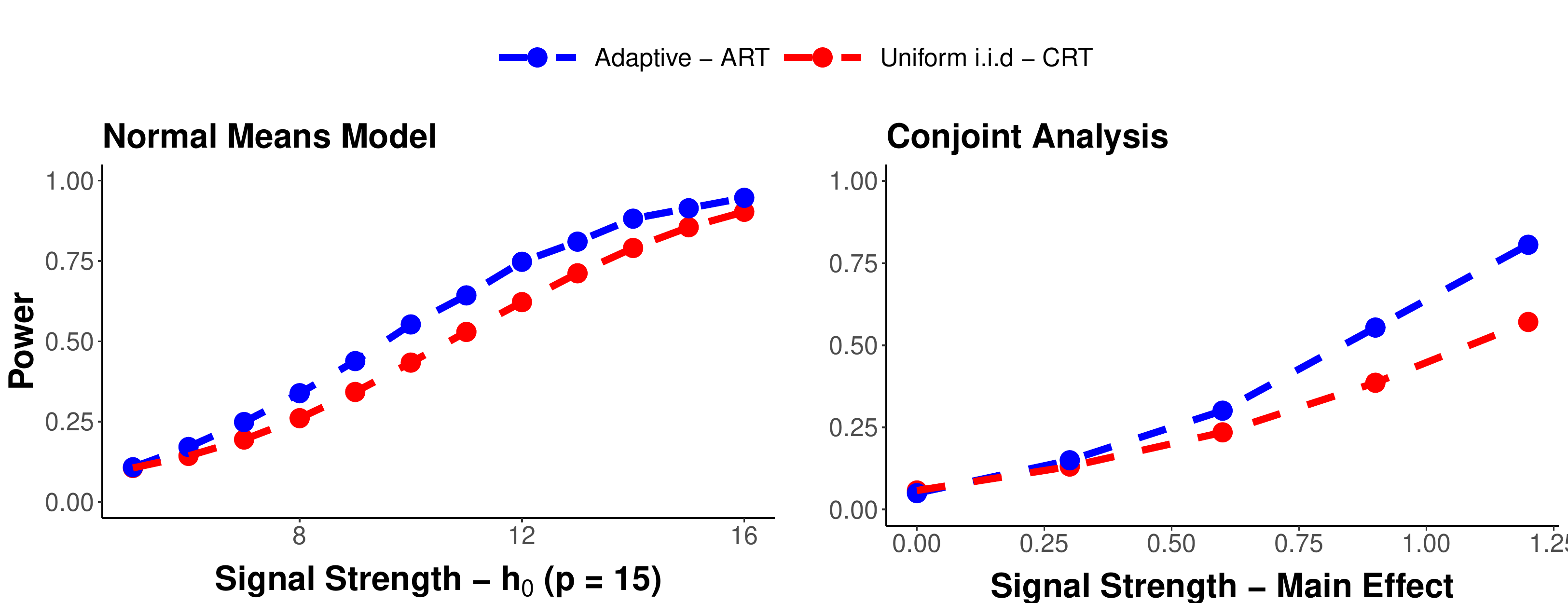}
\caption{The figure is a preview of how the power of the ART (based on an adaptive sampling procedure) and that of the CRT (based on an \iid uniform sampling procedure) varies as the signal increases for the normal-means bandit setting (left plot) and the conjoint setting (right plot). We defer the details of the left and right plots to Section~\ref{subsec:normal_means_model_takeways} and Section~\ref{subsection:sim_results}, respectively. }
\label{fig:preview}
\end{center}
\end{figure}

We postulate that adapting can substantially increase power compared to a typical uniform \iid sampling procedure because adaptive procedure sample more from arms that are not only the true signal but also ``fake'' null arms that initially look like true signals by chance. This reduces the variance of detecting the signal by bringing the ``fake'' noisy arms closer to null arms. Additionally, it is likely that the ART will also down-weight arms which (with high probability) contain no signal, thus allocating more sampling budget on exploring other relevant arms. We also find a stronger conclusion in the normal-means model setting, namely that an adaptive sampling procedure can be more powerful than even the oracle \iid sampling procedure when the signal is relatively strong (see Section \ref{subsec:normal_means_model_takeways} for details). Section~\ref{section:discussion} concludes with a discussion and remarks about future work. 

\subsection{Related Works and Setting}
\label{subsection:related_works}
In this section, we put our proposed method in the context of the current literature. The ART methodology is in the intersection of reinforcement learning and ``Model-X'' randomization inference procedures. As far as we know, our paper is the first to weaken the \iid assumption and allow adaptive testing in the context of randomization inference when specifically tackling the independence testing problem. We remark that \citep{boji:shep:19} considers \textit{unconditional} randomization testing in sequentially adaptively sampled treatment assignments. However, this work does not cover the more general case of conditional randomization testing and assumes a causal inference framework under the finite-population view, i.e., conditioning on the potential outcomes  \citep{rubin:imbens}. Our work differs in that we allow for both super-population and finite-population view and additionally generalize to the \textit{conditional} independence testing problem for general sequentially adaptive procedures (see Section~\ref{subsection:testing} for more details). We also acknowledge that \citep{Adap_RT_ref1} (and references within) contain mentions of randomization inference in adaptive settings but serves primarily as a literature summary of randomization inference and provides no formal testing for general adaptive procedures. 

There is also a large literature on sequential testing, where the primary goal is to produce any-time valid $p$-values, i.e., testing the null hypothesis sequentially at every time point $t$ while controlling type-1 error \citep{sequential1, sequential2}. In this sequential setting, there is a stochastic stopping rule that determines when to stop collecting data (typically when there is enough evidence to reject the null hypothesis), thus the sample size is random. We remark that although we use the word ``sequential'' sampling throughout the paper, our work is \textit{not} related to this sequential testing framework. In other words, we assume we have a fixed sample size and use an adaptive sampling procedure that sequentially updates the sampling probabilities at every time to increase the statistical power of rejecting the null hypothesis. 

As hinted above, many ideas from the reinforcement learning literature can also be useful starting points to construct a sensible adaptive procedure. For example, we find ideas from the multi-arm bandit literature, including the Thompson sampling \citep{thompson}, epsilon-greedy algorithms \citep{epsilon_greedy}, and the UCB algorithm \citep{UCBRL} to be useful when constructing the adaptive sampling procedure. Although ideas from reinforcement learning can be utilized when performing the ART, the objective of independence testing is different than that of a typical reinforcement learning problem. This difference is illustrated and further emphasized in the theoretical analysis of the normal means bandit problem in Section~\ref{subsec:normal_means_model_takeways} and Section~\ref{subsec:adapting_intuition}.

\subsection{The Conditional Randomization Test (CRT)}
\label{subsection:CRT}
We begin by introducing the CRT that requires an \iid sampling procedure. The CRT assumes that the data $(X_t, Z_t, Y_t) \overset{\iid}{\sim} f_{XZY}$ for $t = 1, 2, \dots, n$, where $f_{XZY}$ denotes the joint probability density function (pdf) or probability mass function (pmf) of $(X,Z, Y)$ and $n$ is the total sample size. For brevity, we refer to both probability density function and probability mass function as pdf\footnote{Neither the CRT nor our paper needs to assume the existence of the pdf. However, for clarity and ease of exposition, we present the data generating distribution with respect to a pdf.}. The CRT aims to test whether the variable of interest $X$ affects the distribution of $Y$ conditional on $Z$, i.e., \nullns. If $Z$ is the empty set, the CRT reduces to the (unconditional) randomization test. The CRT tests \nulls by creating ``fake'' resamples $\tilde{X}_{t}^{b}$ for $t = 1, 2, \dots, n$ from the conditional distribution $X \mid Z$ induced by $f_{XZ}$, the joint pdf of $(X,Z)$, for $b = 1, 2, \dots, B$, where $B$ is the Monte-Carlo parameter of choice. More formally, the fake resamples $\tilde{X}_{t}^{b}$ are sampled in the following way,
\begin{equation}
\tilde{X}_t^b \sim \frac{f_{XZ}(\tilde{x}_t^b, Z_t)}{ \int_{z} f_{XZ}(\tilde{x}_t^b, z) dz} \text{ for } t = 1, 2, \dots, n,    
\label{eq:CRT_cond_dist}
\end{equation}
where the right hand side is the pdf of the conditional distribution $X \mid Z$ induced by the joint pdf $f_{XZ}$, lower case $\tilde{x}_t^b$ represents the realization of random variable $\tilde{X}_t^b$, and each $\tilde{X}_t^b$ is sampled \iid for $b = 1, 2, \dots, B$ independently of $X$ and $Y$. Since each sample $X_t$ only depends on the current $Z_t$, the right hand side of Equation~\eqref{eq:CRT_cond_dist} is a conditional distribution that is a function of only its current $Z_t$. Under the conditional independence null, \nullns, \citeauthor{CRT} show that  $(\mathbf{\tilde{X}^1}, \mathbf{Z}, \mathbf{Y})$, $(\mathbf{\tilde{X}^2}, \mathbf{Z}, \mathbf{Y})$, \dots, $(\mathbf{\tilde{X}^B}, \mathbf{Z}, \mathbf{Y})$, and $(\mathbf{X}, \mathbf{Z}, \mathbf{Y})$ are exchangeable, where $\mathbf{X}$ denotes the complete collection of $(X_1, X_2, \dots, X_n)$. $\mathbf{\tilde{X}^b}$, $\mathbf{Z}$, and $\mathbf{Y}$ are defined similarly. This implies that any test statistic $T(\mathbf{X}, \mathbf{Z}, \mathbf{Y})$ is also exchangeable with $T(\mathbf{\tilde{X}^b}, \mathbf{Z}, \mathbf{Y})$ under the null. This key exchangeability property allows practitioners to use any test statistic $T$ when calculating the final $p$-value. More formally, the CRT proposes to obtain a $p$-value in the following way, 
\begin{equation}
    p_{\text{CRT}} = \frac{1}{B+1} \left[1 + \sum_{b=1}^{B} \mathbbm{1}_{\{T(\mathbf{\tilde{X}^b}, \mathbf{Z}, \mathbf{Y}) \geq T(\mathbf{X}, \mathbf{Z}, \mathbf{Y})\}}\right],
\label{eq:p-value}
\end{equation}
\noindent where the addition of 1 is included so that the null $p$-values are stochastically dominated by the uniform distribution. Due to the exchangeability of the test statistics, the $p$-value in Equation~\eqref{eq:p-value} is guaranteed to have exact type-1 error control, i.e., $\mathbb{P}(p_{\text{CRT}}\leq \alpha) \leq \alpha$ for all $\alpha \in [0,1] $ (under the null) despite the choice of $T$ and any $Y \mid (X,Z)$ relationship. This also allows the practitioner to ideally choose a test statistic to powerfully distinguish the observed test statistic with the resampled fake test statistic such as the sum of the absolute value of the main effects of $X$ from a penalized Lasso regression \citep{lasso}. 

\section{Methodology}
\label{section:method}
\subsection{Sequential Adaptive Sampling Procedure}
The ART, like the CRT, is tied to a specific sampling procedure. Although it generalizes the \iid sampling procedure, it still relies on a specific sequentially adaptive sampling procedure. Therefore, we also refer to the sequentially adaptive sampling procedure as the ART sampling procedure and now formally present the definition of this procedure. 

\begin{definition}[Sequential Adaptive Sampling Procedure - The ART sampling procedure]
We say the sample $(\mathbf{X}, \mathbf{Z}, \mathbf{Y})$ follows a sequential adaptive sampling procedure $A$ if the sample obeys the following sequential data generating process.
\begin{align*}
(X_1, Z_1) &\sim f^A_1(x_1, z_1), \hspace{0.2cm} Y_1 \sim f_{Q}(x_1, z_1)  \\
(X_2, Z_2) &\sim f^A_2(x_2, z_2 \mid x_1, z_1, y_1), \hspace{0.2cm} Y_2 \sim f_{Q}(x_2, z_2)  \\
 &  \hspace{0.5cm} \vdots \\
(X_t, Z_t) &\sim f^A_t(x_t, z_t \mid x_1, z_1, y_1, \dots, x_{t-1}, z_{t-1}, y_{t-1}), \hspace{0.2cm} Y_t \sim f_{Q}(x_t, z_t) ,
\end{align*}
where lower case $(x_t, z_t, y_t)$ denotes the realization of the random variables $(X_t, Z_t, Y_t)$ at time $t$, respectively, $f^A_t$ denotes the joint pdf of $(X_t, Z_t)$ given the past realizations, and $f_Q$ denotes the pdf of the response $Y_t$ as a function of only the current $(X_t, Z_t)$.
\label{def:adaptive}
\end{definition}
Definition~\ref{def:adaptive} captures a general sequential adaptive experimental setting, where an experimenter adaptively samples the next values of $(X_t, Z_t)$ according to an adaptive sampling procedure $f_t^A$ that may be dependent on all the history (including the outcome) while ``nature'' $f_Q$ determines the next outcome. We emphasize that $f_Q$ is generally unknown and in most cases hard to model exactly. We also remark that practitioners need not implement a fully adaptive scheme, e.g., $f_t^A$ can remain identical and even independent of the history for many $t$ if the researcher wishes to only adapt at some time points (see Section~\ref{sec:normal_means_model} for an adaptive sampling scheme that only adapts once).

Figure~\ref{fig:dag_natural_sampling_procedure} visually summarizes the sequential adaptive procedure, where we allow the next sample to depend on all the history (including the response). Although Definition~\ref{def:adaptive} makes no assumption about the adaptive procedure $f_t^A$ (even allowing the adaptive procedure to change across time), it does implicitly assume that the response $Y$ has no carryover effects, i.e., $f_Q$ is only a function of its current realizations $(x_t, z_t)$ as there are no arrows in Figure~\ref{fig:dag_natural_sampling_procedure} from previous $(X_{t-1}, Z_{t-1})$ into current $Y_t$. It also assumes that $f_Q$ is stationary and does not change across time. Both of these assumptions are typically invoked in the sequential reinforcement learning literature \citep{RLcausal1, RLbook}. 

% Our methodology naturally extends to non-stationary models, where $f_Q$ also depends on $t$. However, for presentational clarity, we present our method in the common stationary scenario. 

\label{subsection:adaptive_sampling}
\begin{figure}[H]
    \centering
    \includegraphics[width=3in]{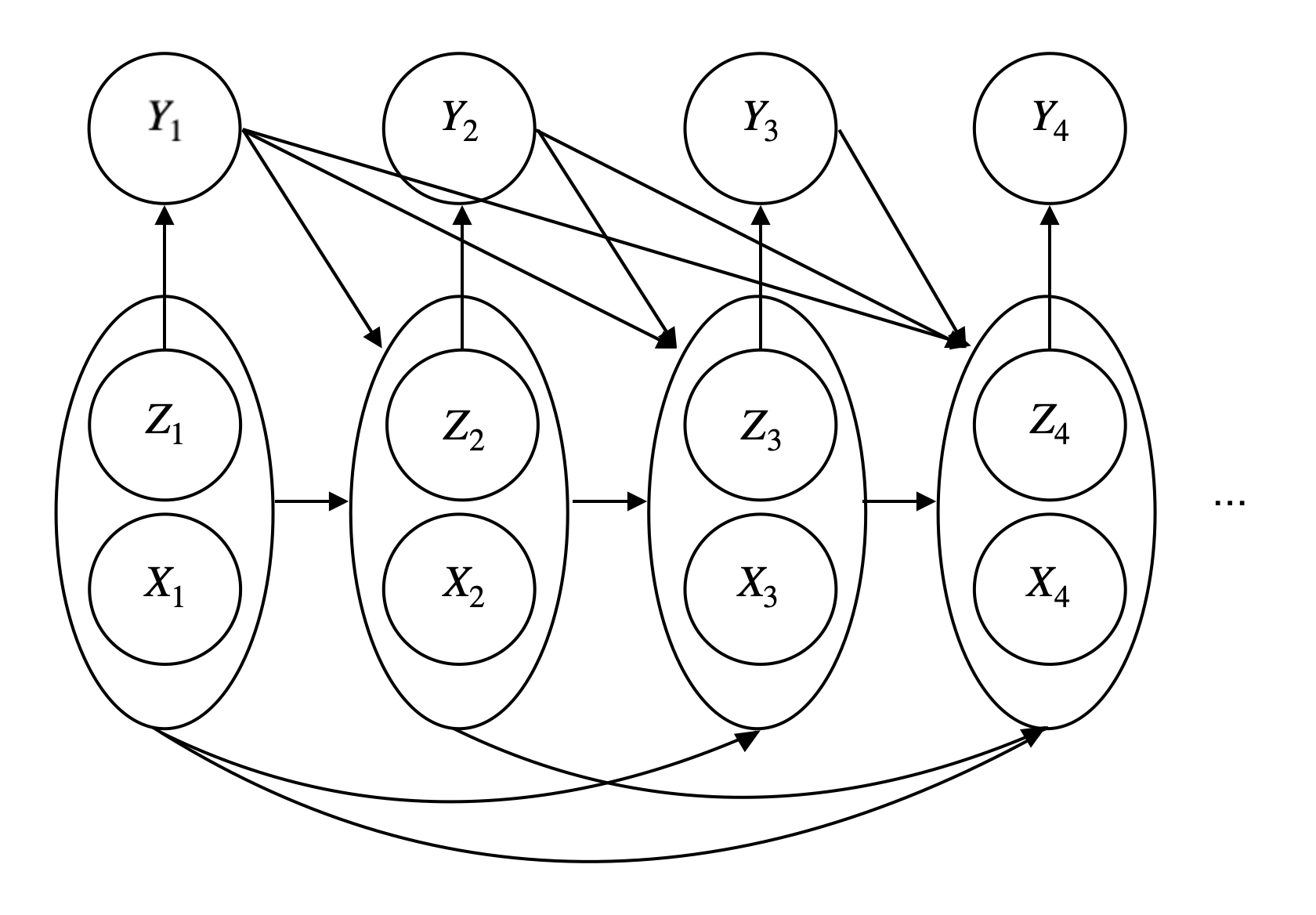}
    \centering
    \caption{Schematic diagram of the ART sampling procedure in Definition~\ref{def:adaptive}. The directed arrows denote the order in how the random variable(s) may affect the corresponding random variable(s).}
    \label{fig:dag_natural_sampling_procedure}
\end{figure}

\subsection{Hypothesis Test}
Given the sampling procedure defined in Definition~\ref{def:adaptive}, the main objective is to determine whether the variable of interest $X$ affects $Y$ after controlling for $Z$. Because the sampling scheme is no longer \textit{iid}, testing $Y \independent X \mid Z$ requires further notation and formalization. In the CRT, the null hypothesis of interest is formally $Y_t \independent X_t \mid Z_t$ for all $t = 1, 2, \dots, n$. Since the data is sampled \textit{iid}, $Y_t \independent X_t \mid Z_t$ reduces to testing $\mathbf{Y} \independent \mathbf{X} \mid \mathbf{Z}$ using the whole data since the subscript $t$ is irrelevant. However, for an adaptive collected data, $\mathbf{Y} \independent \mathbf{X} \mid \mathbf{Z}$ is trivially false for any non-degenerate adaptive procedure $A$ because $\mathbf{X}$ depends on $\mathbf{Y}$ through $f^A_{t}$. Just like the CRT, the practitioners are interested in whether $X$ affects $Y$ for each sample $t$. We now formalize this by testing the following null hypothesis $H_0$ against $H_1$, 
\begin{equation}
    \begin{aligned}
        H_0: f_Q(x, z) &= f_Q(x', z) \text{ for all } x, x' \in \mathcal{X}, z \in \mathcal{Z} \\
        H_1: f_Q(x, z) &\ne f_Q(x', z) \text{ for some } x, x' \in \mathcal{X}, z \in \mathcal{Z}, 
    \end{aligned}
\label{eq:adap_null}
\end{equation}
where $\mathcal{X}$ denotes the entire domain of $X$ that captures all possible values of $X$ regardless of the distribution of $X$ induced by the adaptive procedure. For example, if $X$ is a univariate discrete variable that can take any integer values, then $\mathcal{X} = \mathbb{Z}$ even if the adaptive procedure $A$ only has a finite support with positive probability only on values $(-1, 0, 1)$. In such a case, testing $H_0$ using the aforementioned adaptive procedure $A$ will only be powerful up to the restricted support induced by $A$. $\mathcal{Z}$ is defined similarly as the entire domain for $Z$.

We finish this subsection by connecting $H_0$ to the causal inference literature. First, $H_0$ captures the same notion as the CRT null of $Y \independent X \mid Z$ because if $X$ makes any distributional impact on $Y$ given $Z$, then $H_0$ is false. On the other hand, if $H_0$ is false, then the CRT null is trivially false. Recently, \citeauthor{mypaper} show that the CRT null is equivalent to testing the following causal hypothesis
\begin{equation*}
H_0^{\text{Causal}}: Y_t(x,z) \overset{d}{=} Y_t(x', z) \text{ for all } x, x' \in \mathcal{X}, z \in \mathcal{Z},
\end{equation*}
where $Y_t(x,z)$ is the potential outcome for individual $t$ at values $X = x, Z = z$ and we have implicitly assumed the SUTVA assumption \citep{rubin:imbens}. The proposed $H_0$ already captures the causal hypothesis $H_0^{\text{Causal}}$ because $f_{Q}(x,z)$ characterizes the causal relationship between $(X,Z)$ and $Y$. To formally establish this in the potential outcome framework, we define $Y_t(x,z) \overset{i.i.d}{\sim} f_Q(x, z)$ from a super-population framework, i.e., the potential outcomes are viewed as random variables. Then $H_0$ is equivalent to the causal hypothesis $H_0^{\text{Causal}}$. Additionally, if the researcher wishes to think in terms of the finite-population framework, i.e., conditioning on the potential outcomes and units in the sample, then only a simple modification of Definition~\ref{def:adaptive} is needed. We first replace obtaining the response $Y_t \sim f_Q$ in Definition~\ref{def:adaptive} from a stochastic $f_Q$ to a fixed potential outcome $Y_t = Y_t(x_t, z_t)$ at every time point $t$, where $Y_t(x_t, z_t)$ is the deterministic (non-random) potential outcome of individual $t$ with values $X_t = x_t$ and $Z_t = z_t$. Then $H_0$ reduces to the sharp Fisher null that states $Y_t(x, z) = Y_t(x', z)$ for all $x, x' \in \mathcal{X}, z \in \mathcal{Z}$ and all individuals $t$ in our finite population. This finite-population testing framework is the one proposed in \citep{boji:shep:19}, where the authors perform the unconditional randomization test in a sequential adaptive setting like ours.

\subsection{Adaptive Randomization Test (ART)}
\label{subsection:testing}
Since $(X_t, Z_t, Y_t)$ are no longer sampled \iid from some joint distribution, the main challenge is to construct $\mathbf{\tilde{X}^b}$ such that $(\mathbf{\tilde{X}^b}, \mathbf{Z}, \mathbf{Y})$ and $(\mathbf{X}, \mathbf{Z}, \mathbf{Y})$ are still exchangeable to ensure the validity of the $p$-value in Equation~\eqref{eq:p-value}. A necessary condition for the joint distributions of $(\mathbf{\tilde{X}^b}, \mathbf{Z}, \mathbf{Y})$ and $(\mathbf{X}, \mathbf{Z}, \mathbf{Y})$ to be exchangeable is that they are equal in distribution. For our sequential adaptive sampling procedure, $X_t$ depends on all the history including the response and it is unclear how to construct our resamples.

To solve this, we propose a natural resampling procedure that respects our sequential adaptive setting in Definition~\ref{def:adaptive}. Before formally presenting the resampling procedure, we provide intuition on how to construct valid resamples $\mathbf{\tilde{X}^b}$. Similar to the CRT, the key is to create fake copies of $X$ by replicating the original sampling procedure of $X$ conditional on $\mathbf{Z}, \mathbf{Y}$. For the CRT sampling procedure, this reduces to sampling $X_t$ \iid from the conditional distribution of $X_t \mid Z_t$ for all $t = 1, 2, \dots, n$. In our sequential adaptive sampling procedure, this reduces to sampling $X_t$ conditional on the history as done exactly in the original adaptive sampling procedure since $X_t$ does not depend on the future values of $Z$ and $Y$. We now formalize this in the following definition.
\begin{definition}[Natural Adaptive Resampling Procedure]
Given data \data, $\mathbf{\tilde{X}}^{b}$ follows the natural adaptive resampling procedure if $\mathbf{\tilde{X}}^{b}$ satisfies the following data generating process,
$$
\tilde{X}_1^b \sim \frac{f^A_1(\tilde{x}_1^b, z_1)}{\int_{z} f^A_1(\tilde{x}_1^b, z) dz}, \tilde{X}_2^b \sim \frac{f^A_2(\tilde{x}_2^b, z_2 \mid \tilde{x}_1^b, z_1, y_1)}{\int_{z} f^A_1(\tilde{x}_1^b, z \mid \tilde{x}_1^b, z_1, y_1 ) dz}, \dots, \tilde{X}_n^b \sim \frac{f^A_t(\tilde{x}_n^b, z_n \mid \tilde{x}_1^b, z_1, y_1, \dots, \tilde{x}_{n-1}^b, z_{n-1}, y_{n-1})}{\int_{z} f^A_1(\tilde{x}_n^b, z \mid \tilde{x}_1^b, z_1, y_1, \dots, \tilde{x}_{n-1}^b, z_{n-1}, y_{n-1}) dz} \text{,}
$$
for $b = 1, 2, \dots, B$ independently conditional on $(\mathbf{Z}, \mathbf{Y})$, where $\tilde{x}_t^b$ are dummy variables representing $\tilde{X}_t^b$. 
\label{def:naturalresampling}
\end{definition}
Similar to Equation~\eqref{eq:CRT_cond_dist}, Definition~\ref{def:naturalresampling} formalizes how each $\tilde{X}_t$ is sequentially sampled from the conditional distribution of $X_t \mid (X_{1:(t-1)}, Z_{1:t}, Y_{1:(t-1)})$. We call this the natural adaptive resampling procedure (NARP) because at each time $t$ the fake resamples $\tilde{X}_t^b$ are sampled from the original sequential adaptive distribution of $X_t$ conditional on $Z_{1:t}$ and $Y_{1:(t-1)}$. Just like the CRT, Definition~\ref{def:naturalresampling} requires one to sample from a conditional distribution. For this practically important consideration, we propose a more practical alternative where the experimenter, at each time $t$, samples $Z_t$ first and then samples the variable of interest $X_t$ from $X_t \mid Z_{1:t}, Y_{1:(t-1)}$ at every time step (as opposed to simultaneously sampling $(X_t, Z_t)$ from a joint distribution). This alternative procedure loses very little generality but allows the NARP in Definition~\ref{def:naturalresampling} to directly sample from the already available conditional distribution. We refer to this as the convenient adaptive sampling procedure. 

Unfortunately resampling from the NARP does not immediately gaurantee a valid $p$-value. Recall that we require our resampled $\mathbf{\tilde{X}}^b$ to be exchangeable with $\mathbf{X}$ conditional on $(\mathbf{Z}, \mathbf{Y})$. A necessary condition of exchangeability requires the joint distribution of $(\tilde{\mathbf{X}},\mathbf{Y},\mathbf{Z})$ be the same as that of $(\mathbf{X},\mathbf{Y},\mathbf{Z})$. In particular, the following distributional relationship is always true for any $t$ when assuming the NARP,
\begin{equation}
    \tilde{X}_{1:(t-1)} \indep Z_t \mid \left ( Y_{1:(t-1)}, Z_{1:(t-1)}  \right ) \text{,}
    \label{eq:intuition_assumption}
\end{equation}
because $\tilde{X}_{1:(t-1)}$ is a random function of only $\left ( Y_{1:(t-1)}, Z_{1:(t-1)}  \right ) $ and not the future $Z_t$. Equation~\eqref{eq:intuition_assumption} directly shows that $Z$ can not depend on previous $X$ because we require $(\tilde{\mathbf{X}},\mathbf{Y},\mathbf{Z})$ and $(\mathbf{X},\mathbf{Y},\mathbf{Z})$ to be exchangeable. This constraint turns out to be both sufficient and necessary to ensure validity of using the ART with the NARP to test $H_0$ as formally stated in Theorem \ref{thm:main} and Theorem \ref{thm:necessity}.

\begin{assumption}[$Z$ can not adapt to previous $X$]
For each $t = 1, 2, \dots, n$ we have by basic rules of probability $f_t^A(x_t, z_t \mid x_{1:(t-1)},z_{1:(t-1)} ,y_{1:(t-1)} ) = g_t^A(x_t \mid x_{1:(t-1)},z_{1:(t-1)} ,y_{1:(t-1)} , z_t )h_t^A(z_t \mid x_{1:(t-1)},z_{1:(t-1)} ,y_{1:(t-1)})$, where $g_t^A, h_t^A$ denotes the conditional and marginal density functions induced by the joint pdf of $f_t^A$ respectively. We say an adaptive procedure $A$ satisfies Assumption~\ref{assumption:main} if $h_t^A(z_t \mid x_{1:(t-1)},z_{1:(t-1)}, \allowbreak y_{1:(t-1)})$ does not depend on $x_{1:(t-1)}$, for $t = 2, 3, \dots, n$.
\label{assumption:main}
\end{assumption}
\noindent Assumption \ref{assumption:main} states that the sequential adaptive procedure $A$ does not allow $Z_{t}$ to depend on $X_{1:(t-1)}$. For the gender example above, Assumption~\ref{assumption:main} does not allow other factors, e.g., party affiliation, candidate personality, etc., to depend on the previous values of gender. However, Assumption~\ref{assumption:main} still allows the practitioner to sample the next values of gender based on all the historical data, even sampling more of male or female based on a strong interaction with other factors. Although Assumption~\ref{assumption:main} does restrict our adaptive procedure, it is crucial that each $X_t$ and $Z_t$ are still allowed to adapt by looking at its own previous values and the previous responses. 

\begin{figure}[H]
    \centering
    \includegraphics[width=3in]{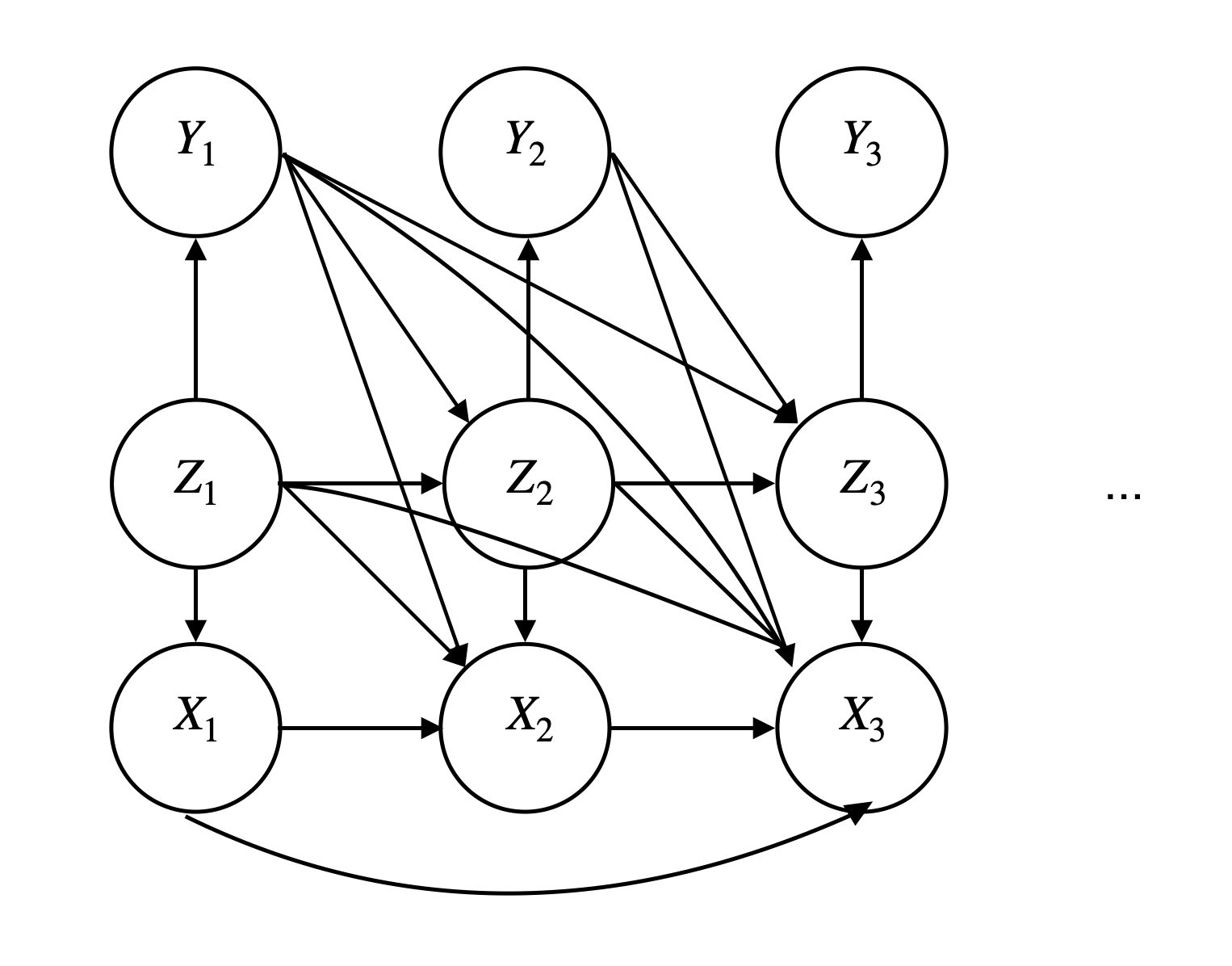}
    \centering
    \caption{Schematic diagram of the convenient adaptive sampling procedure that satisfies Assumption~\ref{assumption:main}. As before, the directed arrows denote the order in how the random variable(s) may affect the corresponding random variable(s).}
    \label{fig:dag_fast}
\end{figure}

We visually summarize Assumption~\ref{assumption:main} and a more convenient, but not necessary, way to conduct a restricted adaptive sampling procedure in Figure~\ref{fig:dag_fast}. Figure~\ref{fig:dag_fast} shows a set of arrows from $Z_t$ into $X_t$ as opposed to them being simultaneously generated as in Figure~\ref{fig:dag_natural_sampling_procedure} to allow the proposed NARP in Definition~\ref{def:naturalresampling} to conveniently sample directly from the already available conditional distribution. Assumption~\ref{assumption:main} is also satisfied in Figure~\ref{fig:dag_fast} as there exist no arrows from any $X_{t'}$ into $Z_t$ for $t' < t$. Before stating our main theorem, we summarize the ART procedure in Algorithm~\ref{algo:ART_procedure}. We note that although the $p$-value $p_{\text{ART}}$ in Equation~\eqref{eq:p-valueART} is similar to $p_{\text{CRT}}$ in Equation~\eqref{eq:p-value}, the resamples $\tilde{X}^b$ are different in the two procedures. We now state the main theorem that shows the finite-sample validness of using the ART for testing $H_0$.

\begin{algorithm}[t]
  \textbf{Input:} Adaptive procedure $A$, test statistic $T$, total number of resamples $B$\;
  Given an adaptive procedure $A$, obtain $n$ samples of $(\mathbf{X}, \mathbf{Z}, \mathbf{Y})$ according to the sequential adaptive procedure in Definition~\ref{def:adaptive}. \\
  \For{$b = 1, 2, \dots, B$}{
  Sample $\tilde{X}^{(b)}$ according to the NARP in Definition~\ref{def:naturalresampling};
 } 
  \textbf{Output:}
  \begin{equation}
     p_{\text{ART}} := \frac{1}{B+1} \left[1 + \sum_{b=1}^{B} \mathbbm{1}_{\{T(\mathbf{\tilde{X}^b}, \mathbf{Z}, \mathbf{Y}) \geq T(\mathbf{X}, \mathbf{Z}, \mathbf{Y})\}}\right]
     \label{eq:p-valueART}
  \end{equation}
  \caption{ART $p$-value}  \label{algo:ART_procedure}
\end{algorithm}

\begin{theorem}[Valid $p$-values under the ART]
\label{thm:main}
Suppose the adaptive procedure $A$ follows the adaptive procedure in Definition~\ref{def:adaptive} and satisfies Assumption~\ref{assumption:main}. Further suppose that the resampled $\tilde{X}^{b}$ follows the NARP in Definition~\ref{def:naturalresampling} for $b = 1, 2, \dots B$. Then the p-value $p_{\text{ART}}$ in Algorithm~\ref{algo:ART_procedure} for testing $H_0$ is a valid $p$-value. Equivalently, $\mathbb{P}(p_{\text{ART}} \leq \alpha) \leq \alpha$ for any $\alpha \in [0,1]$. 

\begin{remark}
We note that $p_{\text{ART}}$ is also a valid $p$-value conditional on $\mathbf{Y}$ and $\mathbf{Z}$.
\end{remark}

\end{theorem}

The proof of Theorem \ref{thm:main} is in Appendix \ref{subsection:proofART}. This theorem is the main result of this paper, which allows testing $H_0$ for sequentially adaptive sampling procedures through randomization inference. Before concluding this section, as alluded before, we state formally in Theorem~\ref{thm:necessity} that our assumption is indeed necessary to establish the exchangeability of $(\mathbf{X}, \mathbf{Z}, \mathbf{Y})$ and $(\mathbf{\tilde{X}^b}, \mathbf{Z}, \mathbf{Y})$ if we follow the natural adaptive procedure in Definition~\ref{def:naturalresampling}.
\begin{theorem}[Necessity of Assumption]
\label{thm:necessity}
For an adaptive procedure $A$, if the resampled $\mathbf{\tilde{X}}^{b}$ follows the natural adaptive resampling procedure in Definition~\ref{def:naturalresampling} and $(\mathbf{X}, \mathbf{Z}, \mathbf{Y})$ and $(\mathbf{\tilde{X}^b}, \mathbf{Z}, \mathbf{Y})$ are exchangeable, then Assumption \ref{assumption:main} must hold, i.e., Assumption~\ref{assumption:main} is necessary. 
\end{theorem}
\noindent The proof is in Appendix \ref{subsection:proofART}.

\subsection{Multiple Testing}
\label{subsection:multipletesting}
So far we have introduced our proposed method to test $H_0$ for a single variable of interest $X$ conditional on other experimental variables $Z$. However, the practitioner may be interested in testing multiple $H_0$ for multiple variables of interest (including variables from $Z$). 

To formalize this, denote $X = (X^1, X^2, \dots, X^{p})$ to contain $p$ variables of interest, each of which can also be multidimensional. Informally speaking, our objective is to  perform $p$ tests of $Y \indep X^{j} \mid X^{-j}$ for $j = 1, 2, \dots p$, where $X^{-j}$ denotes all variables in $X$ except $X^j$. Given a fixed $j$, our proposed methodology in Section~\ref{subsection:adaptive_sampling}-\ref{subsection:testing} can be used to test any single one of these hypothesis. The main issue with directly extending our proposed methodology for testing all $j = 1, 2, \dots p$ variables is that Assumption~\ref{assumption:main} does not allow $X^{-j}$ to depend on previous $X^{j}$ but $X^{j}$ may depend on previous $X^{-j}$ when testing a single hypothesis $Y \indep X^{j} \mid X^{-j}$. This asymmetry may cause this assumption to hold when testing for $X^j$ but simultaneously not hold when testing for $X^{j'}$ for $j\neq j'$. Thus, in order to satisfy Assumption~\ref{assumption:main} for all variables of interest simultaneously, we modify our procedure such that each $X_t^j$ is independent of $X_{t'}^{j'}$ for all $j, j'$ and $t' \leq t$. In other words, we force each $X_t^j$ to be sampled according to its \textit{own} history $X_{1:(t-1)}^j$ and the history of the response but not the history and current values of $X^{j'}$ for $j \neq j'$ and for every $j$. We formalize this in following assumption.  
\begin{assumption}[Each $X^j$ does not adapt to other $X^{j'}$]
For each $t = 1, 2, \dots, n$ suppose each $X_t = (X_t^1, X_t^2, \dots, X_t^p)$ are sampled according to a sequential adaptive sampling procedure $A$: $X_t \sim f_{t}^A(x_t^1, x_t^2, \dots, x_t^p \mid x^{-j}_{1:(t-1)}, x_{1:(t-1)}^j, \allowbreak y_{1:(t-1)})$. We say an adaptive procedure $A$ satisfies Assumption~\ref{assumption:multiple} if $ f_{t}^A $ can be written into following factorized form, for $t = 2, 3, \dots, n$,
$$
f_{t}^A(x_t^1, x_t^2, \dots, x_t^p \mid x^{-j}_{1:(t-1)}, x_{1:(t-1)}^j, y_{1:(t-1)}) = \prod_{j = 1}^p f_{t, j}^A (x_t^j \mid x_{1:(t-1)}^j, y_{1:(t-1)}^j)
$$
with every $f_{t,j}^A(\cdot | x^t_{1:(t-1)}, y^t_{1:(t-1)})$ being a valid probability measure for all possible values of $ (x^t_{1:(t-1)}, y^t_{1:(t-1)}) $.
\label{assumption:multiple}
\end{assumption}
Assumption \ref{assumption:multiple} states that $X^{j}$ can not adapt based on the history of any other $X^{j'}$ for all $j \neq j'$. This assumption is sufficient to satisfy Assumption~\ref{assumption:main} when testing $H_0$ for any $X^j$ for any $j = 1, 2, \dots, p$, thus leading to a valid $p$-value for every $X^j$ simultaneously when using the proposed ART procedure in Algorithm~\ref{algo:ART_procedure}. Although our framework gives valid $p$-values for each of the multiple tests, we need to further account for multiple testing issues. For example, one na\"ive way to control the false discovery rate is to use the Benjamini Hochberg procedure \citep{benj:hoch:95}, but this is not the focus of our paper.

\subsection{Discussion of the Natural Adaptive Resampling Procedure}
\label{subsection:discussionnatural}
Keen readers may argue the NARP is merely a practical choice but an unnecessary one, thus no longer requiring Assumption~\ref{assumption:main}. Exchangeability requires \datas and $(\mathbf{\tilde{X}}^b, \mathbf{Z}, \mathbf{Y})$ to be equal in distribution. Consequently, if one could sample the entire data vector $\mathbf{\tilde{X}}$ from the conditional distribution of $\mathbf{X} \mid (\mathbf{Z}, \mathbf{Y})$, then this construction of $\mathbf{\tilde{X}}$ would satisfy the required distributional equality. In general, however, it is well known that it is difficult to sample from a complicated graphical model \citep{wainwright2008graphical}. To illustrate this, we show how constructing valid resamples $\mathbf{\tilde{X}}^b$ for even two time periods may be difficult without Assumption~\ref{assumption:main} with the following equations.
\begin{equation*}
    \begin{aligned}
        &  P(X_1 = x_1, X_2 = x_2 \mid Z_1 = z_1, Z_2 = z_2, Y_1 = y_1, Y_2 = y_2) \\
        &= \frac{ P(X_2 = x_2 \mid X_1 = x_1, Z_1 = z_1, Z_2 = z_2, Y_1 = y_1) P(Z_2 = z_2 \mid X_1 = x_1 , Y_1 = y_1, Z_1 = z_1) P(X_1 =x_1 \mid Z_1 = z_1)}{\int_{x} P(Z_2 = z_2 \mid X_1 =  x, Y_1 = y_1, Z_1 = z_1) \mathrm{d} P(X_1 = x\mid Z_1 = z_1)} \\
        &\propto P(X_2 = x_2 \mid X_1 = x_1, Z_1 = z_1, Z_2 = z_2, Y_1 = y_1) \left [ P(Z_2 = z_2 \mid X_1 = x_1 ,Y_1 = y_1, Z_1 = z_1) P(X_1 = x_1 \mid Z_1 = z_1) \right ].
    \end{aligned}
\end{equation*}
This follows directly from elementary probability calculations. Since any valid construction of $\mathbf{\tilde{X}}^b$ must have that $P(\tilde{X}_1 = x_1, \tilde{X}_2 = x_2 \mid Z_1 = z_1, Z_2 = z_2, Y_1 = y_1, Y_2 = y_2) =  P(X_1 = x_1, X_2 = x_2 \mid Z_1 = z_1, Z_2 = z_2, Y_1 = y_1, Y_2 = y_2)$, the above equation shows that it is generally hard to construct valid resamples due to the normalizing constant in the denominator of the second line. We further note that Assumption~\ref{assumption:main} bypasses this problem because $P(Z_2 = z_2| X_1=x_1, Y_1 = y_1, Z_1 = z_1)$ is now independent of the condition $X_1 = x_1$. Therefore, the denominator in the second line is always $ P(Z_2 = z_2 \mid X_1 = x_1 , Y_1 = y_1, Z_1 = z_1)$, cancelling out with the numerator. 

Although sampling from a distribution that is known up to a proportional constant has been extensively studied in the Markov Chain Monte Carlo (MCMC) literature \citep{liu2001monte}, many MCMC methods introduce extra computational burden to an already computationally expensive algorithm that requires $B + 1$ resamples and computation of test statistic $T$. Moreover, it is unclear how ``approximate'' draws from the desired distribution in a MCMC algorithm may impact the exact validness of the $p$-values. This problem may be exacerbated when the sample size $n$ is large because the errors for each resamples could exponentially accumulate across time. Therefore, we choose to use the NARP along with Assumption~\ref{assumption:main} as the proposed method because it avoids these complications.

\section{ART in Normal Means Model}
\label{sec:normal_means_model}
In this section, we explore the ART under the well-known normal-means setting \cite{stein_normal_means}. We first introduce the normal-means setting, the sampling procedures we consider, and the test statistic in Section \ref{subsec:normal_means_model_definition}. We then present two main theorems, Theorem \ref{thm:power_iid} and Theorem \ref{thm:power_adptive}, that characterize the asymptotic power of both the \iid procedure and a na\"ive, but still insightful, two stage adaptive sampling procedure under local alternatives of $O(n^{-1/2})$ distance in Section \ref{subsec:theorems_for_normal_means_model}. Finally, we numerically evaluate Theorem \ref{thm:power_iid} and Theorem \ref{thm:power_adptive} to illustrate when the adaptive sampling procedure leads to an increase of power in Section~\ref{subsec:normal_means_model_takeways}. Lastly, we postulate the main reasons for why an adaptive sampling procedure is more powerful than an \iid sampling procedure in Section~\ref{subsec:adapting_intuition}.

\subsection{Normal Means Model}
\label{subsec:normal_means_model_definition}
    Formally, the normal-means model is characterized by the following model. 
    \begin{equation*}
       f_Q =  Y  \mid (X = j) \sim \mathcal{N}(\theta_j , 1) , \quad \text{ for }  j \in \mathcal{X} \coloneq \{1,2,\dots,p \} \text{,}
    \end{equation*}
    where $j$ refers to the $p$ different possible integer values of $X$. We refer to the different values of $X$ as different arms. For this setting there are no other experimental variables $Z$. Our task is to characterize power under the alternative, i.e., when at least one arm of $X$ has a different mean than that of the other arms. For simplicity, we consider an alternative where only one arm has a positive non-zero mean while the remaining $p-1$ arms have zero mean. This leads to the following one-sided alternative. 
    \begin{equation*}
        \text{H}_1^{\text{NMM}}: \text{there exists only one } j^{\star} \text{ such that } \theta_{j^{\star}} = h > 0 \text{ and } \theta_{j} = 0, \forall j \ne j^{\star} \text{.}
    \end{equation*}
    As usual, our null assumes that $X$ does not affect $Y$ in any way,
    \begin{equation*}
        H_0^{\text{NMM}}: \theta_{j} = 0, \forall j \in \{ 1,2,\dots, p\} \text{.}
    \end{equation*}
    Given a budget of $n$ samples, our task is to come up with a reasonable adaptive sampling procedure that leads to a higher power than that of the typical uniform \iid sampling procedure. Because we do not use a fully adaptive procedure for this setting but a simplified two step adaptive procedure, we use subscript $i$ instead of $t$ to denote the sample index for this section. We now formally state the general \iid sampling procedure.
    \begin{definition}[Normal Means Model: \textit{iid} Sampling procedure with Weight Vector $q$]
    \label{def:normal_means_model_iid_sampling}
        We call a sampling procedure \textit{\textit{iid} with weight vector $q = (q_1,q_2,\cdots,q_p)$} if each sample of $ \mathbf{X} = \left ( X_1, X_2, \dots, X_n \right )$ is sampled independently and 
    \begin{equation}
    \label{eq:definition_of_iid_sampling_procedure_with_weight_vetor_q}
        \mathbb{P} (X_i = j) = q_j\text{,}  \text{ for any } i \in {1,2,\dots,n} \text{ and } j  \in \mathcal{X} \text{.}
        % \in \{ 1,2,\dots, p\} \text{.}
    \end{equation}
    \end{definition}
    \noindent We note that this definition is more general than the uniform \text{iid} sampling procedure that pulls each arm with equal probability, i.e., $q = (1/p, 1/p, \dots, 1/p)$. We further denote $ \mathbf{X} \sim \mathcal{M}(q)$ to compactly describe the \textit{iid} sampling procedure for $\mathbf{X}$. With a slight abuse of notation, we also use $X_i \sim \mathcal{M}(q)$ to denote the above distribution of $X_i$.
    
    Despite the simplicity of the normal-means setting, analyzing the power of a fully adaptive procedure is generally theoretically infeasible. Therefore, we consider a na\"ive \say{two stage} adaptive procedure. The first stage is an exploration stage that follows the typical \iid sampling procedure while the second stage is again another \textit{different} \iid sampling procedure that adapts once based on the first stage's data. More specifically, the second stage will adapt by reweighting the probability of pulling each arm by a function of the sample mean. Under the alternative, we expect the arm with the true signal will on average have a higher sample mean, thus we can exploit this arm more in the second stage. Furthermore, the adaptive procedure will also detect arms that, by chance, lead to higher sample means. In such a case, we can additionally identify these ``fake'' signal arms and sample more to ``de-noise'' and reduce the variance from these arms. We note that this two-stage adaptive procedure does not utilize the full potential of an adaptive sampling procedure, but we show that even a simple two stage adaptive procedure can lead to insightful gains and conclusions. We formally summarize the adaptive procedure in Definition~\ref{def:normal_means_model_two_stage_adptive_sampling}. 
    
    \begin{definition}[Normal Means Model: Two Stage Adaptive Sampling procedure]
    \label{def:normal_means_model_two_stage_adptive_sampling}
        An adaptive sampling procedure is called a \textit{two stage adaptive sampling procedure} with \textit{exploration parameter} $\epsilon$, \textit{reweighting function} $f$ and \textit{scaling parameter} $t$ if $(\mathbf{X}, \mathbf{Y})$ are sampled by the following procedure. First, for $1 \le i \le [n\epsilon]$,
        \begin{equation*}
            \begin{aligned}
                 X_i \overset{\textit{iid}}{\sim} \mathcal{M}(q); \quad
                 Y_i \overset{\textit{iid}}{\sim} f_Q(x_i).
            \end{aligned}
        \end{equation*}
        Second, for each $j \in \mathcal{X}$, we compute the sample mean for each arm using the $[n \epsilon]$ samples from the first stage,
        \begin{equation*}
            \bar{Y}_j^{\text{F}} \coloneq \frac{\sum_{i = 1}^{[n\epsilon]} Y_i \mathbb{1}_{X_i = j}}{\sum_{i = 1}^{[n\epsilon]} \mathbb{1}_{X_i = j} }\text{,}
        \end{equation*}
        in which the superscript \say{F} stands for the first stage. Third, we calculate a reweighting vector $Q \in \mathbb{R}^{p}$ as a function of $\bar{Y}_i^{\text{F}}$'s that captures the main adaptive step,
        \begin{equation}
        \label{eq:def_Q_j}
            Q_j = \frac{ f( t \sqrt{n} \cdot \bar{Y}_j^{\text{F}})  }{  \sum_{k=1}^p f( t \sqrt{n} \cdot  \bar{Y}_k^{\text{F}}) }.
        \end{equation}
        Finally, we sample the second batch of samples using the new weighting vector, namely, for $[n\epsilon] + 1 \le i \le n$
        \begin{equation*}
            \begin{aligned}
                X_i \overset{\textit{iid}}{\sim} \mathcal{M}(Q); \quad
                Y_i \overset{\textit{iid}}{\sim} f_Q(x_i). 
            \end{aligned}
        \end{equation*}
    \end{definition}
We comment that \noindent $f(\cdot)$ denotes the adaptive re-weighting function. For example if $f(x) = e^x$, then this reweighs the probability by an exponential function, where $t$ is a hyper-parameter of choice and a larger value of $t$ will lead to a more disproportional sampling of different arms for the second stage. We also scale the reweighting function by $\sqrt{n}$ because the signal decreases with rate $1/\sqrt{n}$ as we describe now in the following section. 

\subsection{Theoretical Power Analysis Through Local Asymptotics}
\label{subsec:theorems_for_normal_means_model}
\subsubsection{Setting}
Although practically one could simulate the power for both the \iid sampling procedure and the adaptive sampling procedure, we theoretically characterize the power for deeper insights and exploration across an entire grid of different signal strengths and number of arms of $X$. To characterize the asymptotic power of both the uniform \iid sampling procedure and the two stage adaptive sampling procedure, we use key ideas from the classical local asymptotic theory \cite{le1956asymptotic}. We remark that for our setting we apply local asymptotic theory to characterize the power of different sampling procedures as opposed to characterizing the distribution of different test statistics of the data from a fixed sampling procedure. 

In our asymptotic setting, we keep $p$ fixed and let $n \to \infty$. To avoid the power from approaching one, we scale our signal strength $h$ proportional to the standard parametric rate $n^{-1/2}$, i.e., 
    \begin{equation}
    \label{eq:h_and_h_0}
        h  =  \frac{h_0}{\sqrt{n}} > 0 \text{,}
    \end{equation}
where $h_0$ is a positive constant.
    
As introduced in Definition \ref{def:normal_means_model_iid_sampling}, we first analyze the power under an \textit{iid} sampling procedure with arbitrary weight vector $q = (q_1,q_2,\cdots,q_p)$ such that $q_i$'s are all positive and $\sum_{i=1}^{p} q_i = 1 $. Without loss of generality, we assume under $H_1^{\text{NMM}}$ the signal is in the first arm, i.e., $j^{\star} = 1$. Consequently, we have under $\text{H}_1^{\text{NMM}}$,

\begin{equation*}
    \begin{aligned}
        \mathbf{X} \sim & \mathcal{M}(q), \\
        %X_i & \overset{\text{i.i.d}}{\sim} \operatorname{Multinomial}(1,q) \\
        Y_i | X_i = 1 & \overset{\text{i.i.d}}{\sim} \mathcal{N} \left (\frac{h_0}{\sqrt{n}}, 1 \right ), \\
        Y_i | X_i = j & \overset{\text{i.i.d}}{\sim} \mathcal{N}(0,1), \text{ for } j \ne 1 \text{.}
    \end{aligned}
\end{equation*}
% Equivalently,
% \begin{equation*}
%     Y_i = S_i \left ( W_i + \frac{h_0}{\sqrt{n}} \right ) + (1 - S_i) G_i,
% \end{equation*}
% where $W_i$ and $G_i$ are standard normal random variables, $S_i \coloneq \mathbb{1}_{X_i = 1} \sim \operatorname{Bernoulli(q_1)}$ and all $W_i$'s, $G_i$'s and $S_i$'s are independent. \dwh{Two suggestions. 1) Can we move the first three equations into one line, I think we can also get rid of the first equation where $\mathbf{X} \sim & \mathcal{M}(q)$ since you basically state it again right after this. 2) Do we also need the equation that comes out after equivalently? I feel like it does not really add anything. I don't think we ever even refer back to W and G is redefined in your theorem right?}

Following the CRT procedure in Section~\ref{subsection:CRT}, since there is no $Z$ to condition on, the fake resample copies, $\{ \tilde{\mathbf{X}}^{b} \}_{b=1}^{B}$, are generated independently from the same distribution as $\mathbf{X}$, namely $\tilde{X}_i^b \overset{\text{i.i.d}}{\sim} \mathcal{M}(q)$.
% \begin{equation*}
%     \tilde{X}_i^b \overset{\text{i.i.d}}{\sim} \mathcal{M}(q)\text{.}
% \end{equation*}
% \dwh{Do we need to put this in an equation enviornment? I think we can just put it in text?} 

To finally compute the $p$-value as done in Equation~\eqref{eq:p-valueART}, we need a reasonable test statistic. 
% The literature on the normal-means setting hints utilizinng the maximum or mean of the sample means of each arm \textbf{David: citationn needed}. 
Therefore, we use maximum of all sample means for each arm as the main proposed test statistic,
\begin{equation}
\label{eq:def_of_T}
    T(\mathbf{X},\mathbf{Y}) = \max_{j \in {1,2,\dots,p}} \bar{Y}_j \coloneq \max_{j \in {1,2,\dots,p}}\frac{\sum_{i=1}^n Y_i \mathbb{1}_{X_i = j}}{ \sum_{i=1}^n \mathbb{1}_{X_i =j} } \text{.}
\end{equation}
We remark that another natural test statistic, $\bar{Y}$ (the sample mean), is degenerate in our testing framework since it does not depend on $\mathbf{X}$ or $\tilde{\mathbf{X}}$. For the sake of notation simplicity, we define the following resampled test statistic
\begin{equation*}
    \tilde{T} (\tilde{\mathbf{X}}, \mathbf{Y}) = \max_{j \in {1,2,\dots,p}} \tilde{\bar{Y}}_j \coloneq \max_{j \in {1,2,\dots,p}}\frac{\sum_{i=1}^n Y_i \mathbb{1}_{\tilde{X}_i = j}}{ \sum_{i=1}^n \mathbb{1}_{\tilde{X} =j} } \text{,}
\end{equation*}
in which, formally speaking, $\tilde{\mathbf{X}} = (\tilde{{X}}_1^1,\dots,\tilde{{X}}_n^1) \coloneq \mathbf{\tilde{{X}}^1}$ and readers should comprehend $\tilde{\mathbf{X}}$ as a generic copy of $ \tilde{\mathbf{X}}^{b}$. Lastly, to deal with the Monte-Carlo parameter $B$, we show in Appendix~\ref{subsection:proof_NMM} that as $B \to \infty$ the power of testing $\text{H}_1$ against $\text{H}_0$ is equal to
\begin{equation}
\label{eq:power_normal_means_B_to_infty}
    \mathbb{P} \left ( \mathbb{P}\left [ T(\mathbf{X},\mathbf{Y}) > z_{1-\alpha}\left( \tilde{T} (\tilde{\mathbf{X}}, \mathbf{Y})\right)  \mid \mathbf{Y} \right ] \right ) \text{.}
\end{equation}
where $z_{1-\alpha}\left( \tilde{T} (\tilde{\mathbf{X}}, \mathbf{Y})\right)$ is the $1 - \alpha$ quantile of the distribution of $\tilde{T}(\tilde{\mathbf{X}}, \mathbf{Y})$ conditioning on $\mathbf{Y}$.

With the above setting, one can explicitly derive the joint asymptotic distributions of $\bar{Y}_j$'s, $\tilde{\bar{Y}}_j$'s and $\bar{Y}$ under the alternative $\text{H}_1$. Consequently, we state the first main theorem of this section which characterizes the asymptotic power of the \textit{iid} sampling procedures with test statistic $T$ as defined in Equation \ref{eq:def_of_T}.

\subsubsection{Asymptotic Results}
All proofs presented in this section are in Appendix~\ref{subsection:proof_NMM}.
\begin{theorem}[Normal Means Model: Power of RT under \textit{iid} sampling procedures]
\label{thm:power_iid}
    Upon taking $B \to \infty$, the asymptotic power of the \textit{iid} sampling procedure with probability weight vector $q = (q_1,q_2,\cdots,q_p)$, as defined in Definition \ref{def:normal_means_model_iid_sampling}, with respect to the RT with the \say{maximum} test statistic,
    % \begin{equation}
    %     T = \max_{j \in {1,2,\dots,p}} \bar{Y}_j = \max_{j \in {1,2,\dots,p}} \frac{\sum_{i = 1}^n Y_i \mathbb{1}_{X_i = j}}{\sum_{i=1}^n \mathbb{1}_{X_i = j}}
    % \end{equation}
    is equal to
    \begin{equation*}
        \text{Power}_{\text{iid}}(q) = \mathbb{P} \left ( T_{\text{iid}} \ge z_{1-\alpha}\left ( \tilde{T}_{\text{iid}} \right ) \right ) \text{,}
    \end{equation*}
    where $z_{1 - \alpha}$ is the $1 - \alpha$ quantile of the distribution of $\tilde{T}_{\text{iid}}$. $T_{\text{iid}}$ and $ \tilde{T}_{\text{iid}}$ are defined/generated as a function of $G \coloneq (G_1, G_2, \dots, G_{p-1})$ and $H \coloneq (H_1, H_2, \dots, H_{p-1})$, both of which are independent and follow the same $(p-1)$ dimensional multivariate Gaussian distribution $\mathcal{N} \left  ( 0, \Sigma (q) \right ) $. 
    %\textbf{Jiaze: There is something wrong here, $R$ should not be independent of $G$ and $H$ - I am double-checking}
    % $T_{\text{iid}}$ and $\tilde{T}_{\text{iid}}$ are then defined as functions of $q$, $G$ and $H$,
    $T_{\text{iid}}$ and $\tilde{T}_{\text{iid}}$ are then defined as
    \begin{equation}
        \begin{aligned}
        \label{eq:T_iid_theorem} 
        % \textbf{David: you mean iid not ind}
            T_{\text{iid}} = T_{\text{iid}} \left ( q, G, H\right ) 
            %\coloneq & h_0 q_1 (1 - q_1) + \max \Bigg ( \left \{ H_1 +  h_0 - h_0 q_1 (1 - q_1)\right \} \cap \{ H_j ,j = 2,\dots,p-1\} \cap \left \{  - \frac{1}{q_p}  \sum_{i=1}^{p-1} q_j H_j   \right \} \Bigg )
            \coloneq  \max \Bigg ( \left \{ H_1 +  h_0 \right \} \cap \{ H_j ,j = 2,\dots,p-1\} \cap \left \{  - \frac{1}{q_p}  \sum_{i=1}^{p-1} q_j H_j   \right \} \Bigg )
        \end{aligned}
    \end{equation}
    and
    \begin{equation}
    \label{eq:T_tilde_iid_theorem}
        \begin{aligned}
            \tilde{T}_{\text{iid}} =\tilde{T}_{\text{iid}} \left (q, G, H \right ) \coloneq  h_0 q_1 + \max \left ( \left \{ G_j ,j = 1,\dots,p-1 \right \}  \cap \left \{ - \frac{1}{q_p} \sum_{j=1}^{p-1}q_j G_j \right \} \right ) \text{.}
        \end{aligned}
    \end{equation}
    Matrices $\Sigma_0$ and $D$ are defined as 
    % \textbf{David: Is v(q) ever defined before this}
    \begin{equation*}
        \Sigma_0(q)  \coloneq
        \begin{bmatrix}
            v(q_1) & -q_1 q_2 & -q_1 q_3 & \cdots & -q_1 q_{p-1}\\
            -q_1 q_2 & v(q_2) &  - q_2 q_3 & \cdots & -q_2 q_{p-1} \\
            -q_1 q_3 & - q_2 q_3  & v(q_3) & \cdots & -q_3 q_{p-1} \\
            \cdots & \cdots & \cdots & \cdots & \cdots \\
            -q_1 q_{p-1} & -q_2 q_{p-1} & -q_3 q_{p-1} & \cdots & v(q_{p-1})
        \end{bmatrix}
        \in \mathbb{R}^{(p-1)\times (p-1)} \text{,}
    \end{equation*}
    with $v(x) = x(1-x)$, and
    \begin{equation*}
        D(q)  \coloneq \operatorname{diag}(q_1,q_2,\dots,q_{p-1}) \in \mathbb{R}^{(p-1)\times (p-1)} \text{.}
    \end{equation*}
    Finally, 
    \begin{equation}
    \label{eq:def_Sigma}
        \Sigma(q)  \coloneq D(q)^{-1} \Sigma_0(q)  D(q)^{-1} \text{.}
    \end{equation}
\end{theorem}
Although Theorem~\ref{thm:power_iid} is stated for any general weight vector $q$, the default choice of weight vector $q$ should be $(1/p, 1/p, \dots, 1/p)$ since the practitioner typically has no prior information about which arm is more important. We refer to this choice of $q$ as the uniform \textit{iid} sampling procedure. We also note that if we assume $p$ to be \say{large} (in a generic sense) and our sampling probabilities $q_j = O(1/p)$ for all $j$, then the diagonal elements of $\Sigma(q)$ will be generally much larger than the off-diagonal elements. Consequently $G$ and $H$ in Theorem~\ref{thm:power_iid} will have approximately independent coordinates, thus both $T_{\text{ind}}, \tilde{T}_{\text{iid}}$ are characterized by nearly independent Gaussian distribution. Before stating the theorem that characterizes the power of the adaptive sampling procedure, we make a few remarks that hint at surprising results that we further explore in the subsequent sections.
\begin{remark}
\label{remark:q_star}
    Suppose an oracle that knows which arm is the signal. Then a na\"ive, but natural idea for the oracle would be to sample more from the arm with signal (large value of $q_1$) to maximize power. As shown in the next section, this is not necessarily the best strategy. In other words, the optimizer $\hat{q}_1 \coloneq \arg \max_{q_1} \text{Power}_{\text{iid}}(q)$ is not always larger than $1/p$, illustrating that it is actually better to sometimes sample less from the actual signal arm depending on the signal strength. This hints at the well known bias-variance trade-off between the mean difference of $T$ and $\Tilde{T}$ and their variances.
    % which can be hinted by comparing Equation \ref{eq:T_iid_theorem} and Equation \ref{eq:T_tilde_iid_theorem}. \dwh{Is this last part true? I don't know how any one can get this from Equation 11 and 12.}
\end{remark}
\begin{remark}
    Following the previous remark, another natural idea is to construct an adaptive procedure that up-weights or down-weights the signal arm according to the oracle weight. However, Section~\ref{subsec:normal_means_model_takeways} shows this na\"ive strategy is not always recommended as the adaptive procedure can do better than even the oracle \textit{iid} sampling procedure.
\end{remark}

By an argument similar to proof for Theorem~\ref{thm:power_iid}, we can also derive the asymptotic power for our two-stage adaptive sampling procedures.

% \textbf{Jiaze: Describe the \say{trade-off} suggested by this theorem in rather intuitive language.}

\begin{theorem}[Normal Means Model: Power of the ART under two-stage adaptive sampling procedures]
\label{thm:power_adptive}
    Upon taking $B \to \infty$, the asymptotic power of a two-stage adaptive sampling procedures with \textit{exploration parameter} $\epsilon$, \textit{reweighting function} $f$, \textit{scaling parameter} $t$ and test statistic $T$ as defined in Definition \ref{def:normal_means_model_two_stage_adptive_sampling}, with respect to the ART with the \say{maximum} test statistic, is equal to 
    %\textbf{Jiaze: Previously I thought we have to use something like the following $T^{\text{Unweighted Sum}}$, but then I realize we can do the most natural $T$ directly}
    % \begin{equation}
    %     \begin{aligned}
    %     T^{\text{Unweighted Sum}} &= \epsilon \max_{j \in \{1,2,\dots,p\}} \bar{Y}_j^{\text{F}} + (1 - \epsilon) \max_{j \in \{1,2,\dots,p\}} \bar{Y}_j^{\text{S}} \\
    %     &= \epsilon \max_{j \in \{1,2,\dots,p\}} \frac{\sum_{i = 1}^{ [\epsilon n] } Y_i \mathbb{1}_{X_i = j}}{\sum_{i=1}^{ [\epsilon n ]} \mathbb{1}_{X_i = j}} + (1 - \epsilon) \max_{j \in \{1,2,\dots,p\}} \frac{\sum_{i = [\epsilon n ]+1}^{ n } Y_i \mathbb{1}_{X_i = j}}{\sum_{i=[\epsilon n ] + 1}^{n} \mathbb{1}_{X_i = j}}
    %     \end{aligned}
    % \end{equation}
    % \begin{equation}
    %     T = \max_{j \in \{1,2,\dots,p\}} \bar{Y}_j = \max_{j \in \{1,2,\dots,p\}} \frac{\sum_{i = 1}^n Y_i \mathbb{1}_{X_i = j}}{\sum_{i=1}^n \mathbb{1}_{X_i = j}}
    % \end{equation}
    % is equal to 
    \begin{equation}
    \label{eq:asymptotic_power_formula_normal_means_model_adaptive}
        \text{Power}_{\text{adap}} \left (\epsilon, t, f \right) \coloneq \mathbb{P}_{ R^{\text{F}}, G^{\text{F}},R^{\text{S}}, H^{\text{F}} } \left ( \mathbb{P}  \left (  T_{\text{adap}} \ge z_{1-\alpha}(\tilde{T}_{\text{adap}} \mid  R^{\text{F}},R^{\text{S}},H^{\text{F}},H^{\text{S}} )    \mid R^{\text{F}},R^{\text{S}}, H^{\text{F}}, H^{\text{S}} \right  ) \right )
    \end{equation}
    where $z_{1-\alpha}(\tilde{T}_{\text{adap},j}\mid  R^{\text{F}},R^{\text{S}}, H^{\text{F}}, H^{\text{S}})$ denotes the $1-\alpha$ quantile of the conditional distribution of $\tilde{T}_{\text{adap}}$ given $ R^{\text{F}}$, $R^{\text{S}}$, $ G^{\text{F}}$ and $ G^{\text{S}}$.
    % $T_{\text{adap}}$, $\tilde{T}_{\text{adap}}$, $ R^{\text{F}}$, $R^{\text{S}}$, $ G^{\text{F}}$ and $ G^{\text{S}}$ are random variables defined through the following generating procedure.
    % All the random quantities mentioned above are defined as followed.
    \begin{equation*}
        T_{\text{adap}} = \max_{j \in \{ 1,2,\dots,p \}} T_{\text{adap},j}
    \end{equation*}
    \begin{equation*}
        \tilde{T}_{\text{adap}} = \max_{j \in \{ 1,2,\dots,p \}} \tilde{T}_{\text{adap},j}
    \end{equation*}
    \begin{equation*}
        % T_{\text{adap},j} = \frac{ q_j \sqrt{\epsilon} W_j  +  Q_j \sqrt{(1 - \epsilon)} \left [ H^{\text{S}}_j + R^{S} + \sqrt{1 -\epsilon} h_0 Q_1 (1 - Q_1) + \mathbb{1}_{j=1} \sqrt{1 -\epsilon} \left (h_0 - h_0 Q_1 (1 - Q_1) \right )  \right ] }{\epsilon q_j + (1 - \epsilon) Q_j}
        T_{\text{adap},j} = \frac{ q_j \sqrt{\epsilon} W_j  +  Q_j \sqrt{(1 - \epsilon)} \left [ H^{\text{S}}_j + R^{S} + \mathbb{1}_{j=1} \sqrt{1 -\epsilon} h_0  \right ] }{\epsilon q_j + (1 - \epsilon) Q_j}
        %f( H^{\text{F}}_j + R + h_0 q_1 (1- q_1) + \mathbb{1}_{j=1} h_0 q_1 ) }
    \end{equation*}
    \begin{equation*}
        \tilde{T}_{\text{adap},j} = \frac{ q_j  \sqrt{\epsilon} \tilde{W}_j +  \tilde{Q}_j \sqrt{(1 - \epsilon)} \left (G^{\text{S}}_j + R^{\text{S}} + \sqrt{1 -\epsilon} h_0 Q_1 \right ) }{\epsilon q_j + (1 - \epsilon) \tilde{Q}_j }
    \end{equation*}
    where $R^{\text{F}}$, $R^{\text{S}}$, $G^{\text{F}}$, $G^{\text{S}}$, $H^{\text{F}}$, $H^{\text{S}}$, $Q$, $\tilde{Q}$, $W$ and $\tilde{W}$ are random quantities generated from the following procedure. First, generate $R^{\text{F}} \sim \mathcal{N}(0,1)$, $G^{\text{F}} \sim \mathcal{N}\left (0, \Sigma(q) \right )$, and $H^{\text{F}} \sim \mathcal{N}\left (0, \Sigma(q) \right )$ independently, where $\Sigma(\cdot)$ is defined in Equation \ref{eq:def_Sigma}. Second, compute 
    % \textbf{David: Some of these variables have subscript j in the generation but not in the equation and vice-versa}
    \begin{equation*}
    \begin{aligned}
        % W_j &= H^{\text{F}}_j + R^{\text{F}} + \sqrt{\epsilon} h_0  q_1 (1-q_1) + \mathbb{1}_{j=1} \sqrt{\epsilon} (h_0  - h_0 q_1 (1-q_1)) \text{, for } j \in \{1,2,\dots, p-1 \} \text{,}\\
        W_j &= H^{\text{F}}_j + R^{\text{F}}  + \mathbb{1}_{j=1} \sqrt{\epsilon} h_0 \text{, for } j \in \{1,2,\dots, p-1 \} \text{,}\\
        \tilde{W}_j &= G_j^{\text{F}} + R^{\text{F}} + \sqrt{\epsilon} h_0 q_1 \text{, for } j \in \{1,2,\dots, p-1 \} \text{,}\\
        W_p &= - \frac{1}{q_p}  \sum_{j=1}^{p-1} q_j H^{\text{F}}_j + R^{\text{F}} +\sqrt{\epsilon} h_0 q_1(1 - q_1) \text{,} \\
        \tilde{W}_p &=- \frac{1}{q_p}  \sum_{j=1}^{p-1} q_j G^{\text{F}}_j  +  R^{\text{F}} + \sqrt{\epsilon} h_0 q_1 
        \text{.}
    \end{aligned}
    \end{equation*}
    Third, compute 
    \begin{equation*}
        \begin{aligned}
            Q_j &= \frac{f(W_j / \sqrt{\epsilon})}{\sum_{j=1}^{p} f(W_j / \sqrt{\epsilon})}  \text{,}\\
            \tilde{Q}_j &= \frac{f(\tilde{W}_j / \sqrt{\epsilon})}{\sum_{j=1}^{p} f(\tilde{W}_j / \sqrt{\epsilon})} \text{.}
        \end{aligned}
    \end{equation*} 
    % \textbf{David: Should be $R^S$ not $R^F$ right}
    We note that with a slight abuse of notation, the $Q$ defined here is the asymptotic distributional characterization of Equation \ref{eq:def_Q_j}. Lastly, generate $R^{\text{S}}\sim \mathcal{N}(0,1)$, $H^{\text{S}} \sim \mathcal{N} \left ( 0, \Sigma(Q)\right )$ and $G^{\text{S}} \sim \mathcal{N} \left ( 0,\Sigma \left (\tilde{Q} \right )  \right )$ independently. 
\end{theorem}

While Theorem~\ref{thm:power_adptive} formally characterizes the asymptotic power for two-stage adaptive procedures, the final result for the asymptotic power, i.e., Equation \ref{eq:asymptotic_power_formula_normal_means_model_adaptive}, is not immediately insightful due to the complicated nature of both the \say{maximum} test statistic and the adaptive sampling procedure. Though Theorem \ref{thm:power_iid} and Theorem \ref{thm:power_adptive} are not directly interpretable, the computational cost of evaluating it numerically is less than na\"ively simulating the adaptive procedure for a large value of $n$ by a factor of $O(n)$. Moreover, since the asymptotic power characterized in Theorem \ref{thm:power_iid} and Theorem \ref{thm:power_adptive} does not depend on $n$, the conclusion is naturally more consistent and unified when compared to the empirical power obtained from simulating with different large sample size. Apart from the computational advantages the theorem provides, it is also of theoretical interest by itself because our work leverages local asymptotic power analysis to characterize the distributions under different sampling strategies as opposed to characterizing the distributions under different test statistics. In addition, this theorem can also serve as a starting point and motivating example for theoretically analyzing the power of the ART for future works.

\subsection{Power Results}
\label{subsec:normal_means_model_takeways}
Given the asymptotic results presented in the previous section, we now attempt to understand how the ART using an adaptive sampling procedure may be more powerful than the CRT using an \iid sampling procedure. As alluded in Remark \ref{remark:q_star}, if a practitioner knows which arm contains the signal, then a na\"ive but natural adaptive strategy is to up-weight or down-weight the known signal arm according to the oracle. We formally define the oracle in the following way, where we assume, without loss of generality, $j^{\star} = 1 $,
\begin{equation*}
    q_1^{\star} \coloneq \arg \max_{ 0 \le q_1 \le 1} \text{Power}_{\text{iid}} (q(q_1)),
\end{equation*}
in which $q(q_1) \coloneq \left (q_1, (1-q_1)/(p-1),(1-q_1)/(p-1),\dots, (1-q_1)/(p-1) \right ) \in \mathbb{R}^{p}$ denotes the sampling probabilities of all $p$ arms, where the first signal arm has probability $q_1$ and the remaining arms (that have no signal) equally share the remaining sampling probability. Let $q^{\star} = q(q_1^{\star})$, i.e., the oracle \iid sampling procedure that samples the known treatment arm in an optimal way. We refer to the \iid sampling with weight vector $q^{\star}$ as the \say{oracle \iid sampling procedure}.\footnote{$q^{\star}$ is not formally the most optimal \iid sampling procedure for all possible \iid sampling procedure since we consider the maximum power when only varying $q_1$ while imposing the remaining arms to all have equal probabilities. However, we do not imagine any other reasonable \iid sampling procedure to have a stronger power than $q^{\star}$ since the remaining $p-1$ arms with no signals are not differentiable in any way, thus we lose no generality by setting them with equal probability.}  

Next, we use numerical evaluations of Theorem \ref{thm:power_iid} and Theorem \ref{thm:power_adptive} to compare the power of the (two-stage) adaptive sampling procedure, uniform \iid sampling, and the oracle \iid sampling procedure across a grid of possible signal strengths $h_0$ and number of arms $p$. For the adaptive sampling procedure described in Definition~\ref{def:normal_means_model_two_stage_adptive_sampling}, we choose the reweighting function $f$ to be the exponential function, i.e., $f(x) = \exp(x)$.

Figure~\ref{fig:power_difference_heat_map} shows how the ART's power with the proposed adaptive sampling procedure is greater than that of both the uniform \iid sampling procedure and even the oracle \iid sampling procedure. To produce this figure, we first fix an arbitrary, but reasonable, combination of hyper-parameters for the ART, i.e., we set exploration parameter $\epsilon = 0.5$ and reweighting parameters $t_0 = \log 2$ and $t = t_0/h_0$. As a reminder, exploration parameter $\epsilon = 0.5$ implies the adaptive procedure spends half of the sampling budget on exploration and only adapts once by reweighting (see Definition \ref{def:normal_means_model_two_stage_adptive_sampling}) after the first half of the \iid samples are collected. The choice of $t_0 = \log 2$ allows the first arm (containing the real signal) to get roughly twice more sampling weight than the remaining arms in the second stage in expectation. Appendix~\ref{Appendix:sim_NMM} shows additional simulations with different choices for the adaptive parameters ($\epsilon, t_0$), demonstrating that the results presented here are not sensitive to the initially chosen parameters. 

The left panel of Figure~\ref{fig:power_difference_heat_map} shows that the power of the ART from the adaptive sampling procedure is uniformly better than that of the CRT using the default uniform \iid sampling procedure. For example, in areas that have high number of arms and signal, the adaptive sampling procedure can have close to 10 percentage points higher power than the uniform \iid sampling procedure. We also note that the left panel of Figure~\ref{fig:preview} plots the left panel of Figure~\ref{fig:power_difference_heat_map} when $p = 15$ while varying $h_0$. The right panel of Figure~\ref{fig:power_difference_heat_map} surprisingly shows that the adaptive sampling procedure can be more powerful than even the oracle \iid sampling procedure when the signal strength is relatively high. This power difference can be as large as 10 percentage points when the signal and number of arms are high. However, we note that the adaptive sampling procedure's power can be lower than that of the oracle \iid sampling procedure when the signal is low. We postulate further in Section~\ref{subsec:adapting_intuition} how and why the ART may be helping in power. We note that for both panels in Figure~\ref{fig:power_difference_heat_map}, the top left corners of the heatmaps have zero difference between the two sampling procedures because this regime of strong signal and low $p$ results in a degenerate power close to one, allowing no significant differences.
\begin{figure}[t]
\begin{center}
\includegraphics[width=\textwidth]{"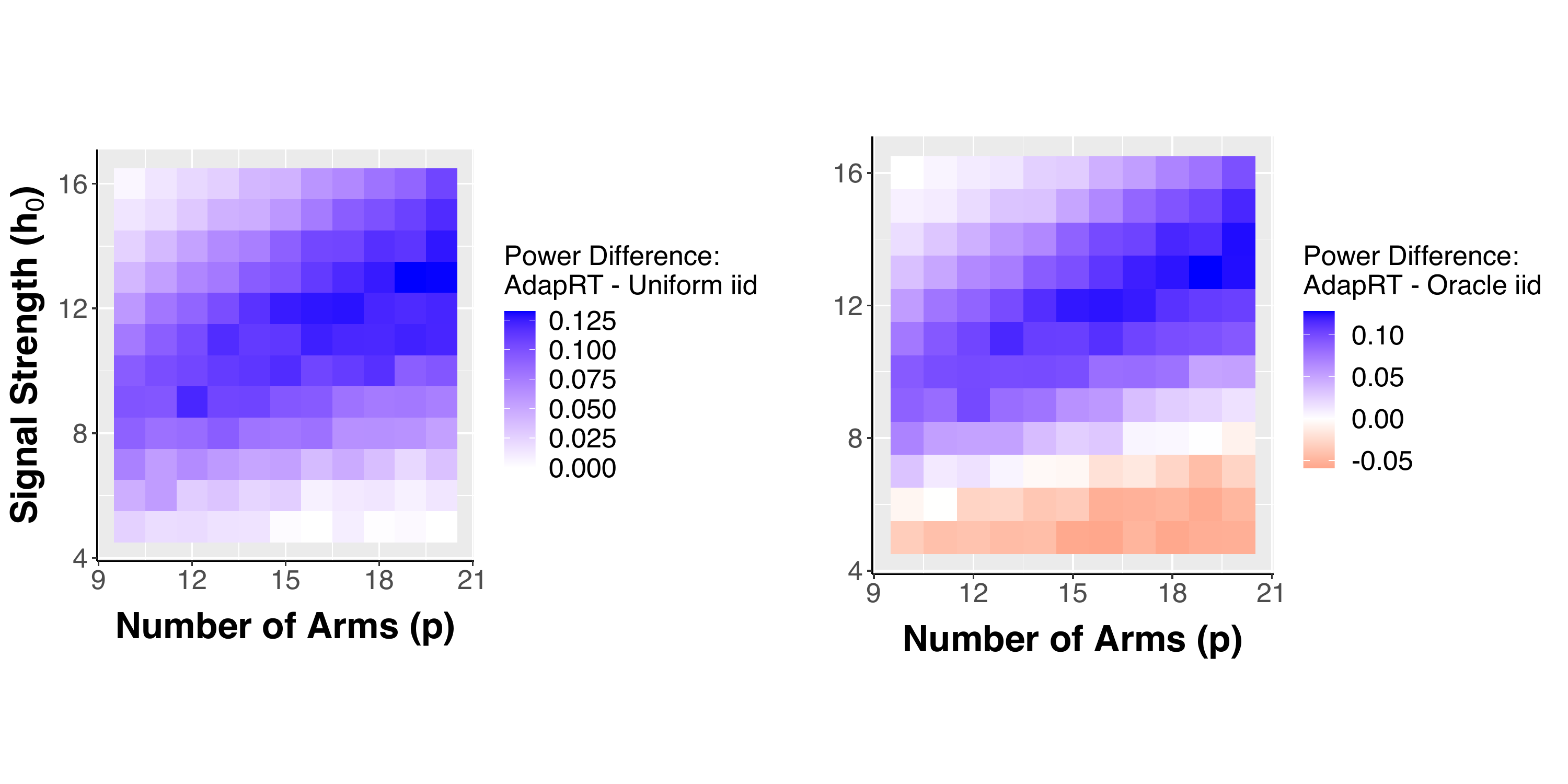"}
\caption{The figure shows the difference between the local asymptotic power of the ART using the adaptive sampling procedure in Definition~\ref{def:normal_means_model_two_stage_adptive_sampling} (with a fixed arbitrary choice of hyper-parameters $\epsilon = 0.5 $ and $ t = \log 2/h_0$) and the CRT using an \iid sampling procedure for different values of signal strength $h_0$ and number of arms $p$. All tests use the test statistic defined in Equation \ref{eq:def_of_T}. The left plot shows that the power of the adaptive sampling procedure is almost uniformly higher than that of the default uniform \iid sampling. The right plot shows that the power of the adaptive sampling procedure is higher than that of even the oracle \iid sampling procedure when the signal strength is relatively high. We note that values on the top left corners of both heatmaps are close to $0$ only because the power of all three sampling procedures is almost degenerately one. The significance level is $\alpha = 0.05$. These heat maps are generated based on Monte Carlo evaluations of Theorem \ref{thm:power_iid} and Theorem \ref{thm:power_adptive}. }  
\label{fig:power_difference_heat_map}
\end{center}
\end{figure}

\subsection{Understanding why Adapting Helps}
\label{subsec:adapting_intuition}
In this subsection, we summarize some of the insights we find from the above analysis of the normal means model. Our goal is to characterize key ideas of why adapting is helpful so practitioners can also build their own successful adaptive procedure. We acknowledge that all statements here are respect to the specific normal-means model setting, but we believe that the main ideas should generalize to different applications and scenarios as shown in Section~\ref{section:conjoint_studies} for instance. Unfortunately, it is difficult to theoretically verify many of the presented insights because the power of the ART and the CRT depends on the behavior of also the resampled test statistics. For example, even if we empirically verify that the adaptive procedure is sampling arms with zero signal with lower probability, it does not directly imply the power is greater because the resampled test statistic may exhibit the same behavior. This would make both the \textit{observed} and \textit{resampled} test statistic approximately indistinguishable, leading to an insignificant $p$-value. Therefore, Figure~\ref{fig:power_difference_heat_map} should serve as the main result that highlights how adapting can indeed help. Nevertheless, we attempt to show some empirical evidence of how adapting is helping.
\begin{figure}[t]
\minipage{0.53\textwidth}
    \includegraphics[width=\linewidth]{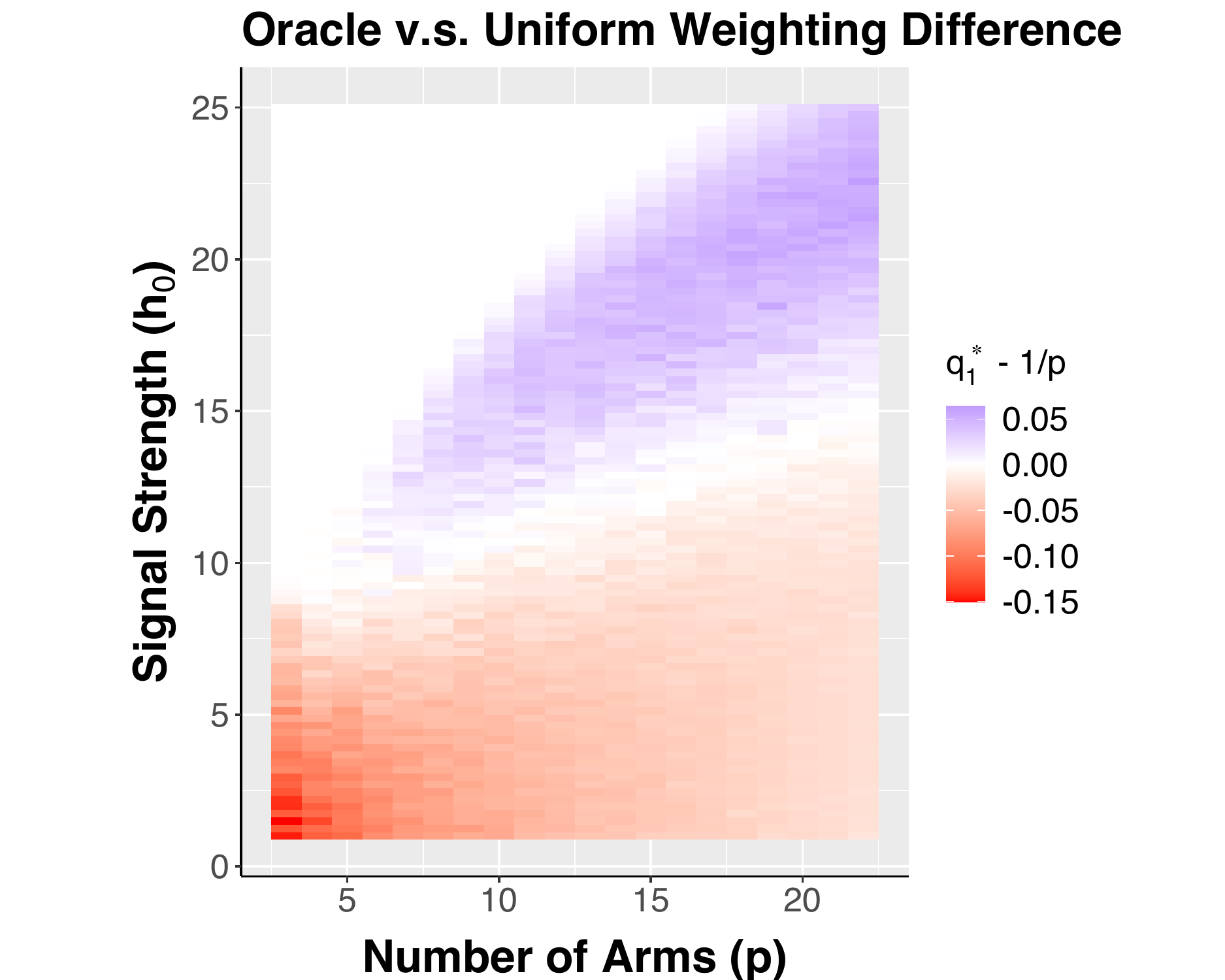}
\endminipage\hfill
\minipage{0.53\textwidth}
    \includegraphics[width=\linewidth]{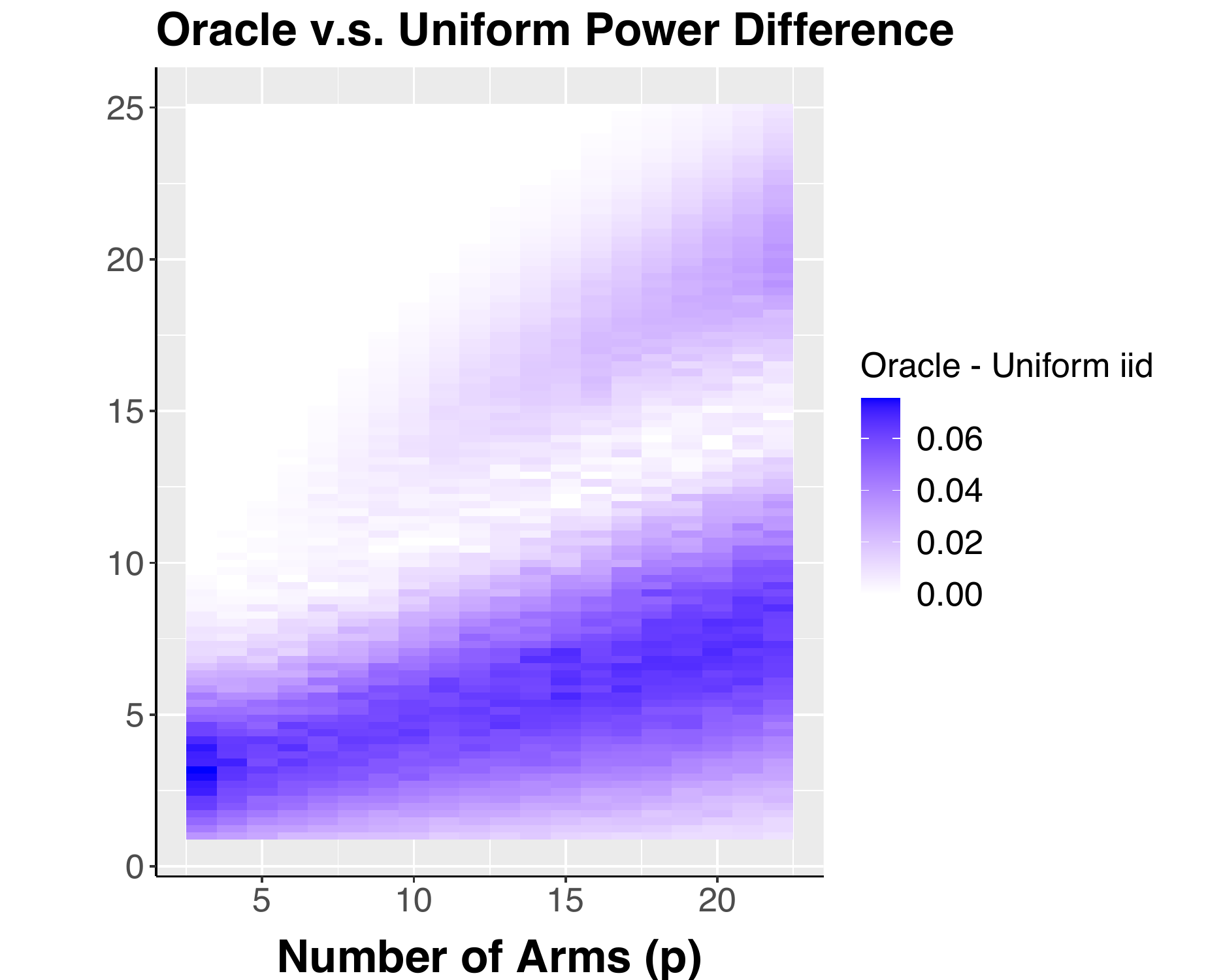}
\endminipage\hfill
    \caption{This figure compares the theoretical power of the CRT from an \iid oracle sampling procedure with the CRT from an uniform \iid sampling procedure. The first panel on the left compares whether oracle $q_1^{\star}$ should down-weight (less than $1/p$) or up-weight (more than $1/p$) the signal arm. The second panel compares the power difference between the oracle and uniform sampling procedures.}
    \label{fig:q_star}
\end{figure}

As pointed out at the beginning of Section \ref{subsec:normal_means_model_takeways}, a natural idea is to try to design adaptive strategies that mimic the oracle \iid procedure. However, the power gain shown in the left plot of Figure~\ref{fig:power_difference_heat_map} can not be attributed to only mimicking the oracle \iid sampling procedure because the right plot of Figure \ref{fig:power_difference_heat_map} shows the adaptive sampling procedure can be more powerful than even the oracle \iid sampling as long as the signal strength is not too low. Additionally, it is unclear if the oracle sampling procedure always samples the signal arm with higher probability as our adaptive sampling procedure does. Consequently, to understand the oracle sampling procedure's behavior further, we present Figure~\ref{fig:q_star} that compares the oracle sampling procedure's behavior with the \iid uniform sampling procedure. 

The left plot of Figure \ref{fig:q_star} shows that the oracle up-weights and also down-weights the signal arm depending on $h_0$ and $p$. For example, the red regions shows that the oracle actually down-weights the signal arm to spend more sampling budget on other arms. Therefore, if mimicking the oracle sampling procedure is the ideal solution, the adaptive procedure should down-weight the signal arm for the red regions in Figure~\ref{fig:q_star}. However, when comparing Figure \ref{fig:power_difference_heat_map} and the left plot in Figure \ref{fig:q_star}, we see that the up-weighting (since $t > 0$) adaptive procedure can actually beat not only the uniform \iid sampling procedure but also the oracle \iid sampling procedure. This shows that the adaptive procedure is doing more than just mimicking the oracle sampling procedure.

Instead, as alluded previously, we believe the main intuition behind the success of the ART is for the following three reasons. As expected, the first reason is that an adaptive sampling procedure can, to some extent, mimic the oracle \iid procedure and achieve closer-to-oracle sampling proportions on average (at least for the regimes that up-weight the signal arm). Additionally, and most importantly, when the adaptive sampling procedure samples more from the arms that look like signal it is not only sampling from the arms that is truly the real signal but also the arms that are ``fake'' signals due to random chance. This allows the adaptive procedure to de-noise these ``fake'' signal arms to a correctly null state. Thirdly, adapting also down-weights arms (with high probability) that contain no signal, allowing our remaining samples to focus on exploring the more relevant arms. 

\section{ART in Conjoint Studies}
\label{section:conjoint_studies}
In this section, we further demonstrate how the ART can help in a popular factorial design called conjoint analysis. Conjoint analysis, introduced more than half a century ago \citep{conjoint}, is a factorial survey-based experiment designed to measure preferences on a multidimensional scale. Conjoint analysis has been extensively used by marketing firms to determine desirable product characteristics \citep[e.g.,][]{market, market2} and among social scientists \citep{AMCE, socialscience2} interested in studying individual preferences concerning election and immigration \citep[e.g.,][]{gender, immigration}. Recently \citeauthor{mypaper} also introduced the CRT in the context of conjoint analysis to test whether a variable of interest $X$ matters at all for a response $Y$ given $Z$. 

Similar to Section~\ref{sec:normal_means_model}, we first show through simulations how the ART can be helpful in a conjoint setting. Unlike the analysis performed above in Section~\ref{sec:normal_means_model}, we do not theoretically characterize the asymptotic power and in exchange consider a fully adaptive procedure and a more complicated test statistic. We then apply our proposed methodology on a recent conjoint study concerning the role of gender discrimination in political candidate evaluation. We show how the proposed adaptive procedure is able to more powerfully detect the role of gender discrimination compared to the original \iid sampling procedure.

\subsection{Simulations and the Adaptive Procedure}
\label{subsection:simulations_conjoint}
In a typical conjoint design, respondents are forced to choose between two profiles presented to them - often known as a forced-choice conjoint design \citep{mypaper, AMCE, gender}. We refer to the two profiles as the ``left'' ($L$) and ``right'' ($R$) profiles\footnote{The profiles are not necessarily always presented side by side.}. In this forced-choice design, the response $Y$ is a binary variable that takes value 1 if the respondent chooses the left profile and zero otherwise. $X$ is our categorical factor(s) of interest (for example candidate's gender) and $Z$ are the remaining factors (for example candidate's political party, age, etc.). Since each respondent observes two profiles, we have that $X_t = (X_t^L, X_t^R)$ and $Z_t = (Z_t^L, Z_t^R)$ for every sample $t$, where the superscripts $L$ and $R$ denote the left and right profiles respectively. 

For simplicity, our simulation setting assumes $(X,Z)$ each contain one factor with four levels. Our response model, $\Pr(Y_{t} = 1 \mid X_t, Z_t)$, follows a logistic regression that includes one main effect for one level of $X$ and $Z$ and one interaction effect between $(X,Z)$. Our response model  assumes ``no profile order effect'', which is commonly invoked in conjoint studies \citep{AMCE, mypaper} and states that changing the order of profiles, i.e., left versus right, does not affect the actual profile chosen. Appendix~\ref{appendix:conjoint_details} contains further details of the simulation setup. 

Before presenting our adaptive procedure, we first build intuition on how an adaptive sampling procedure may help. Consider the typical uniform \iid sampling procedure, where all levels for each factors are sampled with equal probability. If the sample size $n$ is not sufficiently large enough and the signal is sparse and weak, the data may have insufficient samples for levels of $X$ that contain the true effect and by chance may have levels of $X$ that look like ``fake'' effects due to noise. On the other hand, an adaptive sampling procedure can mitigate such issues by ``screening out'' levels that do not look like signal, thus allocating the remaining samples to explore more noisy levels that may not be true signals. Therefore, we speculate the reasons presented in Section~\ref{sec:normal_means_model} for why adapting may be helpful also similarly applies for this setting.

We define $X_t \sim \text{Multinomial}(p_{t, 1}^X, p_{t, 2}^X, \dots, p_{t, K^2}^X)$, where $p_{t, j}^X$ represents the probability of sampling the $j$th arm (arm refers to each unique combination of left and right factor levels) out of $K^2$ possible arms and $K$ is the total levels of $X$. For example, in our simulation setup $K = 4$ and there are 16 possible arms, $(1,1), (1,2), $ etc., and $p_{t,j}^Z$ is defined similarly. The uniform \iid sampling procedure pulls each arm with equal probability, i.e., $p_{t,j}^X = \frac{1}{K^2}, p_{t,j}^Z = \frac{1}{L^2}$ for every $j$ and $L$ is the total number of factor levels for factor $Z$.\footnote{We also note that conjoint applications do indeed default to the uniform \iid sampling procedure (or a very minor variant from it) \citep{immigration, gender}.} Although we present our adaptive procedure when $Z$ contains only one other factor (typical conjoint analysis have 8-10 other factors), our adaptive procedure loses no generality in higher dimensions of $Z$. 

We now propose the following adaptive procedure that adapts the sampling weights of $p_{t, j}^X, p_{t, j}^Z$ at each time step $t$ in the following way,
\begin{equation}
p_{t,j}^X \propto |\bar{Y}_{j,t}^X - 0.5| + |N(0, 0.01^2)|, \qquad p_{t,j}^Z \propto |\bar{Y}_{j,t}^Z - 0.5| + |N(0, 0.01^2)|,
\label{eq:adaptive_conjoint}
\end{equation}
where $\bar{Y}_{j,t}^X$ denotes the sample mean of $Y_1, Y_2, \dots, Y_{t-1}$ for arm $j$ in variable $X$, $\bar{Y}_{j,t}^Z$ is defined similarly, and $N(0, 0.01^2)$ denotes a Gaussian random variable with mean zero and variance $0.01^2$ (the two Gaussians in Equation~\eqref{eq:adaptive_conjoint} are drawn independently). Such an adaptive sampling scheme matches our aforementioned intuition because Equation~\eqref{eq:adaptive_conjoint} will sample more from arms that look like signal (further away from 0.5). We add a slight perturbation in case $\bar{Y}_{j,t}^X$ is exactly equal to 0.5 at any time point $t$ to discourage an arm from having zero probability to be sampled. 

With this reweighting procedure, we build our adaptive procedure. Just like Definition~\ref{def:normal_means_model_two_stage_adptive_sampling}, we also have an $\epsilon$ adaptive parameter that denotes the beginning $[n\epsilon]$ samples that are used for ``exploration'' by using the typical uniform \iid sampling procedure. In the remaining samples, we adapt by changing the weights according to Equation~\eqref{eq:adaptive_conjoint}. We note that this adaptive sampling procedure immediately satisfies Assumption~\ref{assumption:main} and also Assumption~\ref{assumption:multiple} since each variable only looks at its own history and previous responses. Algorithm~\ref{algo:adap_conjoint} summarizes the adaptive procedure.
\begin{algorithm}[t]
 Given adaptive parameter $\epsilon$
  \For{$t = 1, 2, \dots, [n\epsilon]$}{
  Sample $X_t \sim \text{Multinomial}(p_{t, 1}^X, p_{t, 2}^X, \dots, p_{t, K^2}^X)$, where $p_{t,j}^X = \frac{1}{K}$ for all $j = 1, 2, \dots, K^2$ \\
  Sample $Z_t \sim \text{Multinomial}(p_{t, 1}^Z, p_{t, 2}^Z, \dots, p_{t, L^2}^Z)$, where $p_{t,j}^Z = \frac{1}{L}$ for all $j = 1, 2, \dots, L^2$ \\

 }
 
 \For{$t = [n\epsilon] + 1, \dots, n$}{
  Sample $X_t \sim \text{Multinomial}(p_{t, 1}^X, p_{t, 2}^X, \dots, p_{t, K^2}^X)$, where $p_{t,j}^X$ is given in Equation~\eqref{eq:adaptive_conjoint} \\
  Sample $Z_t \sim \text{Multinomial}(p_{t, 1}^Z, p_{t, 2}^Z, \dots, p_{t, L^2}^Z)$, where $p_{t,j}^Z$ is given in Equation~\eqref{eq:adaptive_conjoint} 
 }
 \caption{Adaptive Procedure for Conjoint Studies} \label{algo:adap_conjoint}
\end{algorithm}

Lastly, in order for us to compute the $p$-value in Equation~\eqref{eq:p-valueART}, we need a reasonable test statistic $T$. Although \citeauthor{mypaper} consider a complex Hierarchical Lasso model to capture all second-order interactions, we consider a simple cross-validated Lasso logistic test statistic that fits a Lasso logistic regression of $\mathbf{Y}$ with main effects of $\mathbf{X}$ and $\mathbf{Z}$ and their interactions due to the simplicity of this simulation setting. This leads to the following test statistic
\begin{equation}
T^{\text{lasso}}(\mathbf{X}, \mathbf{Z}, \mathbf{Y}) = \sum_{k = 1}^{K - 1} |\hat \beta_k | + \sum_{k = 1}^{K - 1} \sum_{l = 1}^{L-1}|\hat \gamma_{kl}| ,
\label{eq:teststat_conjoint}
\end{equation}
where $\hat\beta_k$ denotes the estimated main effects for level $k$ out of $K$ levels of $X$ (one is held as baseline) and $\hat\gamma_{kl}$ denotes the estimated interaction effects for level $k$ of $X$ with level $l$ of $L$ total levels of $Z$. This test statistic also imposes the ``no profile order effect'' constraints, i.e., we do not separately estimate coefficients for the left and right profiles to increase power (see \citep{mypaper} and Appendix~\ref{appendix:conjoint_details} for further details). Appendix~\ref{appendix:conjoint_differentTS} also contains additional robustness results, where we repeat our analysis using another test statistic based on the $F$-statistic.

\subsection{Simulation Results}
\label{subsection:sim_results}
We first compare the power of our adaptive procedure stated in Algorithm~\ref{algo:adap_conjoint} with the \iid setting where each arm for $X$ and $Z$ are drawn uniformly at random under the simulation setting described in Section~\ref{subsection:simulations_conjoint}. We empirically compute the power as the proportion of $1,000$ Monte-Carlo $p$-values less than $\alpha = 0.05$. 

For the left panel of Figure~\ref{fig:main_simulation}, we increase sample size when there exist both main effects and interaction effects of $X$. More specifically, we vary our sample size $n = (450, 600, 750, 1,000, 1,300)$ while fixing the main effects of $X$ and $Z$ at 0.6 and a stronger interaction effect at 0.9 (these refer to the coefficients of the logistic response model defined in Appendix~\ref{appendix:conjoint_details}). For the right panel of Figure~\ref{fig:main_simulation}, we increase the main effects of $X$ and $Z$ with no interaction effect and a fixed sample size at $n = 1,000$. We also vary the exploration parameter $\epsilon$ in Algorithm~\ref{algo:adap_conjoint} to $\epsilon = 0.25, 0.5, 0.75$. 

Both panels of Figure~\ref{fig:main_simulation} show that the power of the ART with the proposed adaptive sampling procedure is uniformly greater than that of the CRT with a typical uniform \iid sampling procedure (green). For example when $n = 1,000$ in the left panel, there is a difference in 8.5 percentage points (59\% versus 67.5\%) between the \iid sampling procedure and the adaptive sampling procedure with $\epsilon = 0.5$ (red). When the main effect is as strong as 1.2 in the right panel, there is a difference in 24 percentage points (57\% versus 81\%) between the \iid sampling procedure and the adaptive sampling procedure with $\epsilon = 0.5$. Additionally, when the main effect is 0 in the right panel, thus under $H_0$, the power of all methods, as expected, has type-1 error control as the power for all methods are near $\alpha = 0.05$ (dotted black horizontal line). We also remark that the right panel of Figure~\ref{fig:preview} plots the red (ART with $\epsilon = 0.5$) and green line (CRT) of the right panel of Figure~\ref{fig:main_simulation}.  Appendix~\ref{appendix:conjoint_differentTS} also shows the above conclusions are robust even under a different test statistic based on the $F$-statistic (see Figure~\ref{fig:additional_simulation} in Appendix~\ref{appendix:conjoint_differentTS} for further details). 

\begin{figure}[t]
\begin{center}
\includegraphics[width=\textwidth]{"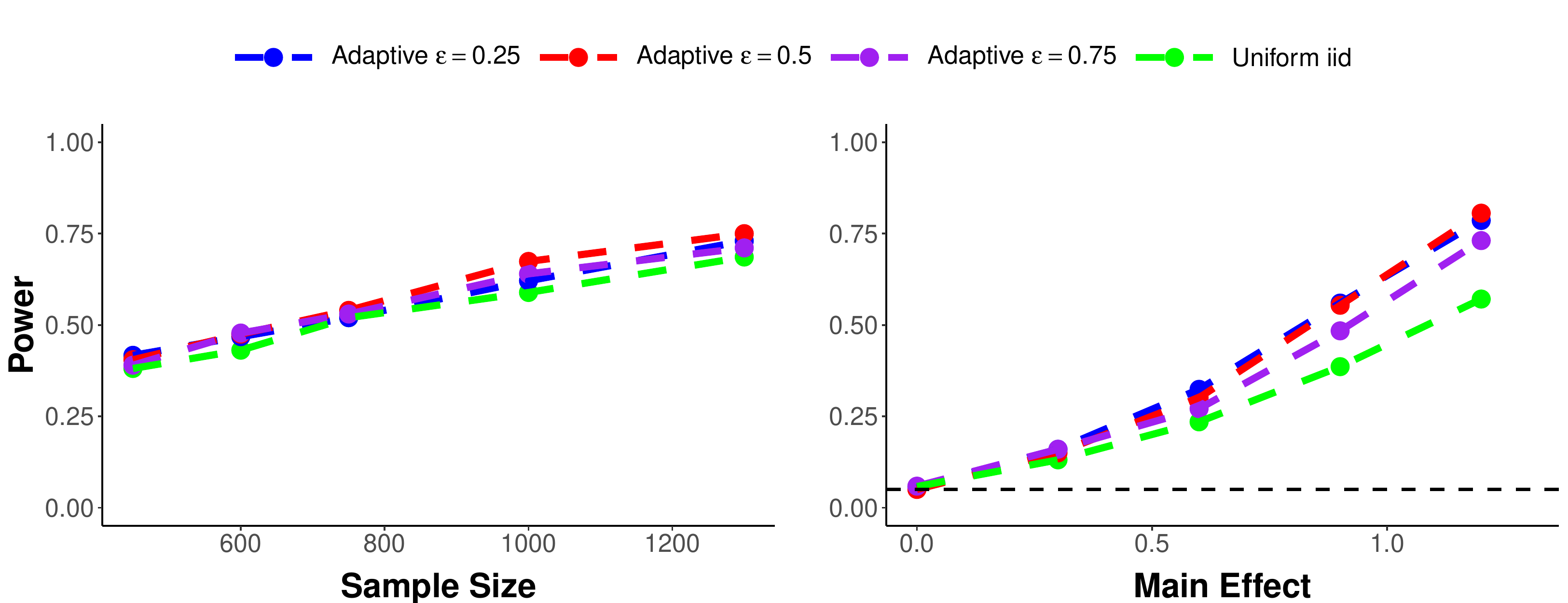"}
\caption{The figure shows how the power of the ART (based on adaptive sampling procedure in Algorithm~\ref{algo:adap_conjoint}) and the CRT (based on an \iid sampling procedure) varies as the sample size increases (left plot) or the main effect increases (right plot). All power curves are calculated from 1,000 Monte-Carlo calculated $p$-values using Equation~\eqref{eq:p-valueART} with $B = 300$ and test statistic given in Equation~\eqref{eq:teststat_conjoint} with their respective resampling procedures. The blue, red, and purple power curves denote the power of the ART using the adaptive procedure described in Algorithm~\ref{algo:adap_conjoint} and $\epsilon = 0.25, 0.50, 0.75$, respectively. The green power curve denotes the power of the uniform \iid sampling procedure. The black dotted line in the right panel shows the $\alpha = 0.05$ line. Finally, the standard errors are negligible with a maximum value of $0.016$.}
\label{fig:main_simulation}
\end{center}
\end{figure}

\subsection{Application: Role of Gender in Political Candidate Evaluation}
\label{subsection:application}
We now apply our proposed method to a recent conjoint study concerning the role of gender discrimination in political candidate evaluation \citep{gender}.  In this study, the authors conduct an experiment based on a sample of voting-eligible adults in the U.S. collected in March 2016, where each of the 1,583 respondents were given 10 pairs of political candidates with uniformly sampled levels of: gender, age, race, family, experience in public office, salient personal characteristics, party affiliation, policy area of expertise, position on national security, position on immigrants, position on abortion, position on government deficit, and favorability among the public (see original article for details). The respondents were then forced to choose one of the two pair of candidate profiles to vote into office, which is our main binary response $Y$. The study consists of a total of $7,915$ responses, where the primary objective was to test whether gender ($X$) matters in voting behavior ($Y$) while controlling for other variables such as age, race, etc. ($Z$).\footnote{The original study consists of $15,830$ responses half of which were about Presidential candidates and the remaining half for Congressional candidates. Because the original study found a statistically significant result for only the Presidential candidates, we focus on the responses for Presidential candidates}

\citeauthor{gender} were able to find a statistically significant effect of candidate's gender on voting behavior of Presidential candidates. We attempt to answer this important question of whether gender matters in voting behavior had the experimenter ran the same experiment for the first time but with a lower sample size or budget $n < 7,915$. To run this quasi-experiment, we assume the original data of size $7,915$ is the population and we draw samples (without replacement) from the original dataset according to our experiment. The original experimental design independently and uniformly sampled all factor levels with equal probability, which will be used as the baseline \iid sampling procedure for comparison. For example, the left and right profiles' gender was either ``Male'' or ``Female'' with equal probability. 

The quasi-experimental procedure is as follows. For simplicity, suppose $X$ is gender and $Z$ is only candidate party. Since each sample consists of a \textit{pair} of profiles, one potential sample may be $X_1 = (\text{Male}, \text{Female})$ and $Z_1 = (\text{Democrat}, \text{Democrat})$, indicating the left profile was a Democratic male candidate and the right profile was Democratic female candidate. Given such a sample, we obtain the subsequent response $Y$ from the original study of 7,915 samples from randomly drawing response $Y$ with corresponding pair of profiles with a Democratic male candidate and a Democratic female candidate. Once we draw this response $Y$, we do not put it back into the population. Since $Z$ in the original study contained 12 other factors, the probability of observing a unique sequence of a particular $(X,Z)$ is close to zero due to the curse of dimensionality. For example, if $Z$ contained only two more factors such as candidate age and experience in public office, then there may exist no samples in the original study that contain a specific profile that is a Democratic male with 20 years of experience in public office and 50 years of age. For this reason, we only run this quasi-experiment for up to one other $Z$, namely the candidate's party affiliation (Democratic or Republican). We choose this variable because \citeauthor{mypaper} suggest strong interactions of gender with the candidate's party affiliation. Since our aim is to show that using the ART with the specified adaptive procedure can help achieve a greater power than that of the CRT using an \iid procedure, it is sensible to try to use other factors $Z$ that may help in power as long as both sampling procedures use the same data for fair comparison.

Given a budget constraint $n < 7,915$, we obtain the power of the ART and the CRT using the uniform \iid sampling procedure that samples each level with equal probability by computing 1,000 $p$-values, where each $p$-value is computed from one quasi-experimentally obtained data of size $n$. Each $p$-value is computed using Equation~\eqref{eq:p-valueART} and the appropriate resamples for the corresponding procedure. The power is empirically computed as the proportion of the 1,000 $p$-values less than $\alpha = 0.1$. Since the applied setting is the same as that of the simulation setting in Section~\ref{subsection:simulations_conjoint}, we use the same adaptive procedure in Algorithm~\ref{algo:adap_conjoint} with $\epsilon = 0.5$ as suggested by Section~\ref{subsection:sim_results} and the same test statistic in Equation~\eqref{eq:teststat_conjoint}.

\begin{table}[!t]
\begin{center}
\begin{adjustbox}{max width=\textwidth,center}
\begin{tabular}{c|cc} 
 & \iid sampling procedure - CRT & Adaptive sampling procedure - ART  \\ 
 \hline
$n = 500$ & 0.13 & 0.14 \\
$n = 1,000$ & 0.14 & 0.17 \\
$n = 2,000$ & 0.24 & 0.30 \\
$n = 3,000$ & 0.31 & 0.40 \\
\end{tabular}
\end{adjustbox}
\caption{The two columns represent the power of the CRT with the uniform \iid sampling procedure and the ART with the adaptive procedure in Algorithm~\ref{algo:adap_conjoint}, respectively, for testing $H_0$, where $X$ is gender (Male or Female) in the gender political candidate study in \citep{gender} and $Z$ is the candidate's party affiliation (Democratic or Republican). Each row represents a different sample size $n$ that aims to replicate the original experiment had the researchers re-ran the experiment with the respective sampling procedures. The power is calculated from the proportion of 1,000 $p$-values less than $\alpha = 0.1$. Each $p$-value is calculated using Equation~\eqref{eq:p-valueART} using the appropriate resamples for the corresponding procedure with test statistic defined in Equation~\eqref{eq:teststat_conjoint}. The ART uses adaptive procedure in Algorithm~\ref{algo:adap_conjoint} with $\epsilon = 0.5$ and the uniform \iid sampling procedure pulls each arm with equal probability.}
\label{tab:conjoint_application}
\end{center}
\end{table}

Table~\ref{tab:conjoint_application} shows the power results using both the \iid sampling procedure and the proposed adaptive sampling procedure. Although the power difference is not as stark as that shown in the simulation in Figure~\ref{fig:main_simulation}, Table~\ref{tab:conjoint_application} still shows that the power of the adaptive sampling procedure is consistently and non-trivially higher than that of the \iid sampling procedure. For example, when $n = 3,000$ (approximately 37\% of the original sample size), we observe a power difference of 9 percentage points with the \iid sampling procedure only having 31\% power, approximately a 30\% increase of power. 

\section{Concluding Remarks}
\label{section:discussion}
In this paper, we introduce the Adaptive Randomization Test (ART) that allows the ``Model-X'' randomization inference approach for sequentially adaptively collected data. The ART, like the CRT, tackles the fundamental independence testing problem in statistics. We showcase the ART's potential through various simulations and empirical examples that show how an adaptive sampling procedure can lead to a more powerful test compared to the typical \iid sampling procedure. In particular, we demonstrate the ART's advantages in the normal-means model and conjoint settings. We believe that adaptively sampling can help for three main reasons. The first reason relates to how an adaptive sampling procedure mimics the oracle \iid procedure in terms of finding optimal sampling weight. Secondly, up-weighting arms that look like signal allows sampling more from arms that contain the true signal but also ``fake'' signal arms that may look like true signals by chance. This allows the adaptive procedure to de-noise and stabilize the fake signal arms.
Thirdly, adapting also down-weights arms (with high probability) that contain no signal, allowing our remaining samples to more efficiently exploring the relevant arms. 

Our work, however, is not comprehensive. While our work analyzes two common settings where the ART is clearly helpful, there exist many future research that can further explore how to build efficient adaptive procedures with theoretical and empirical guarantees under many different scenarios for the respective application. Secondly, as briefly discussed in Section \ref{subsection:multipletesting}, the ART can successfully give multiple valid $p$-values for each relevant hypothesis, but it is not clear if one could make theoretical or empirical guarantees about its properties in the context of multiple testing and variable selection such as controlling the false discovery rate. Thirdly, with the goal of extending our methodology beyond independence testing, an interesting direction is to combine adaptive sampling with other ideas from the ``Model-X'' framework. For instance, \citeauthor{zhang2020floodgate} recently proposed the Floodgate method that goes beyond independence testing by additionally characterizing the strength of the dependency. It would be interesting to extend our adaptive framework in this Floodgate setting. Lastly, the ART is crucially reliant on the natural adaptive resampling procedure (NARP) for the validity of the $p$-values in $p_{\text{AdapCRT}}$. As mentioned in Section~\ref{subsection:discussionnatural}, it may be possible to also have a feasible resampling procedure that does not require Assumption~\ref{assumption:main} but enjoys the same benefits of the ART.

% \textbf{\\ \centering  This is a preliminary draft. Please do not cite or distribute without permission of the authors.}

%%%%%%%%%%%%%%%%%%%%%%%%%%%%%%%%%%%%%%%%%%%%%%%%%%%%%%%%%%%%%%%%%%
%%%%%%%%%%%%%%%%%%%%%%%%%%%%%%%%%%%%%%%%%%%%%%%%%%%%%%%%%%%%%%%%%%

\newpage

\bibliography{bibtex_file,my} 
\bibliographystyle{pa}

%%%%%%%%%%%%%%%%%%%%%%%%%%%%%%%%%%%%%%%%%%%%%%%%%%%%%%%%%%%%%%%%%%
%%%%%%%%%%%%%%%%%%%%%%%%%%%%%%%%%%%%%%%%%%%%%%%%%%%%%%%%%%%%%%%%%%
\newpage

\appendix

\section{Proof of Main Results Presented in Section \ref{section:method} }
\label{subsection:proofART}
\begin{proof}[Proof of Theorem \ref{thm:main}]

    By definition of our resampling procedure, under $\text{H}_0$,
    \begin{equation*}
        \tilde{X}_1 \mid (Y_1, Z_1) \overset{\text{d}}{=} \tilde{X}_1 \mid Z_1 \overset{\text{d}}{=} X_1 \mid Z_1 \overset{\text{d}}{=} X_1 \mid (Y_1, Z_1)
    \end{equation*}
    where the last \say{$\overset{\text{d}}{=}$} is by the null hypothesis of conditional independence, namely $X_1 \indep Y_1 \mid Z_1$. Moreover, it also suggests
    \begin{equation*}
        (\tilde{X}_1,Y_1, Z_1) \overset{\text{d}}{=} (X_1,Y_1, Z_1) \text{.}
    \end{equation*}
    Then we will prove the following statement holds for any $k \in \{ 1,2,\dots,n\}$ by induction,
    \begin{equation}
    \label{eq:induction_claim}
        (\tilde{X}_{1:k},Y_{1:k}, Z_{1:k}) \overset{\text{d}}{=} (X_{1:k},Y_{1:k}, Z_{1:k}) \text{.}
    \end{equation}
    Assuming Equation \ref{eq:induction_claim} holds for $k-1$, we now prove it also holds for $k$. For simplicity, in the rest of this proof, we will use $P(\cdot)$ as a generic notation for \textit{pdf} or \textit{pmf}, though the proof holds for more general distributions without a \textit{pdf} or \textit{pmf}. First,
    \begin{equation}
    \label{eq:proof_of_validity_long_equation_1}
        \begin{aligned}
            & P \left [ ( \tilde{X}_{1:(k-1)}, Y_{1:(k-1)}, Z_{1:k}) =  (x_{1:(k-1)}, y_{1:(k-1)}, z_{1:k}) \right ] \\
            \overset{\text{(i)}}{=} & P \left [ Z_k \mid ( \tilde{X}_{1:(k-1)}, Y_{1:(k-1)}, Z_{1:(k-1)}) =  (x_{1:(k-1)}, y_{1:(k-1)}, z_{1:(k-1)}) \right ] \\
            & \cdot P \left [ ( \tilde{X}_{1:(k-1)}, Y_{1:(k-1)}, Z_{1:(k-1)}) =  (x_{1:(k-1)}, y_{1:(k-1)}, z_{1:(k-1)})\right ] \\
            \overset{\text{(ii)}}{=} & P \left [ Z_k \mid (  Y_{1:(k-1)}, Z_{1:(k-1)}) =  ( y_{1:(k-1)}, z_{1:(k-1)}) \right ] \\
            & \cdot P \left [ ( \tilde{X}_{1:(k-1)}, Y_{1:(k-1)}, Z_{1:(k-1)}) =  (x_{1:(k-1)}, y_{1:(k-1)}, z_{1:(k-1)})\right ] \\
            \overset{\text{(iii)}}{=} & P \left [ Z_k \mid (  Y_{1:(k-1)}, Z_{1:(k-1)}) =  ( y_{1:(k-1)}, z_{1:(k-1)}) \right ] \\
            & \cdot P \left [ ( X_{1:(k-1)}, Y_{1:(k-1)}, Z_{1:(k-1)}) =  (x_{1:(k-1)}, y_{1:(k-1)}, z_{1:(k-1)})\right ] \\
            \overset{\text{(iv)}}{=} & P \left [ Z_k \mid ( X_{1:(k-1)} Y_{1:(k-1)}, Z_{1:(k-1)}) =  (x_{1:(k-1)},  y_{1:(k-1)}, z_{1:(k-1)}) \right ] \\
            & \cdot P \left [ ( X_{1:(k-1)}, Y_{1:(k-1)}, Z_{1:(k-1)}) =  (x_{1:(k-1)}, y_{1:(k-1)}, z_{1:(k-1)})\right ] \\
            = & P \left [ ( X_{1:(k-1)}, Y_{1:(k-1)}, Z_{1:k}) =  (x_{1:(k-1)}, y_{1:(k-1)}, z_{1:k}) \right ] \text{,}
        \end{aligned}
    \end{equation}
    where (i) is simply by Bayes rule; (ii) is because $Z_k \indep \tilde{X}_{1:{k-1}} \mid (Y_{1:(k-1)}, Z_{1:(k-1)})$ since $\tilde{X}_{1:{k-1}}$ is a random function of only $Y_{1:(k-1)}$ and $Z_{1:(k-1)}$; and lastly, (iii) is by induction assumption; (iv) is by Assumption \ref{assumption:main}. Moreover,
    \begin{equation*}
        \begin{aligned}
            & P \left [ ( \tilde{X}_{1:k}, Y_{1:k}, Z_{1:k}) = (x_{1:k}, y_{1:k}, z_{1:k}) \right ] \\
            \overset{\text{(i)}}{=} & P \left [ Y_k = y_k \mid ( \tilde{X}_{1:k}, Y_{1:(k-1)}, Z_{1:k}) =  (x_{1:k}, y_{1:(k-1)}, z_{1:k}) \right ] \cdot  P \left [( \tilde{X}_{1:k}, Y_{1:(k-1)}, Z_{1:k}) =  (x_{1:k}, y_{1:(k-1)}, z_{1:k}) \right ]   \\
            \overset{\text{(ii)}}{=} & P \left [ Y_k = y_k \mid Z_k = z_k \right ] \cdot  P \left [( \tilde{X}_{1:k}, Y_{1:(k-1)}, Z_{1:k}) =  (x_{1:k}, y_{1:(k-1)}, z_{1:k}) \right ] \\
            \overset{\text{(iii)}}{=} & P \left [ Y_k = y_k \mid Z_k = z_k \right ] \cdot  P \left [ \tilde{X}_k = x_k \mid ( \tilde{X}_{1:(k-1)}, Y_{1:(k-1)}, Z_{1:k}) =  (x_{1:(k-1)}, y_{1:(k-1)}, z_{1:k}) \right ]  \\
            & \cdot P \left [ ( \tilde{X}_{1:(k-1)}, Y_{1:(k-1)}, Z_{1:k}) =  (x_{1:(k-1)}, y_{1:(k-1)}, z_{1:k}) \right ] \\
            \overset{\text{(iv)}}{=} & P \left [ Y_k = y_k \mid Z_k = z_k \right ] \cdot  P \left [ X_k = x_k \mid ( X_{1:(k-1)}, Y_{1:(k-1)}, Z_{1:k}) =  (x_{1:(k-1)}, y_{1:(k-1)}, z_{1:k}) \right ]\\
            & \cdot P \left [ ( \tilde{X}_{1:(k-1)}, Y_{1:(k-1)}, Z_{1:k}) =  (x_{1:(k-1)}, y_{1:(k-1)}, z_{1:k}) \right ] \\
            \overset{\text{(v)}}{=} & P \left [ Y_k = y_k \mid Z_k = z_k \right ] \cdot  P \left [ X_k = x_k \mid ( X_{1:(k-1)}, Y_{1:(k-1)}, Z_{1:k}) =  (x_{1:(k-1)}, y_{1:(k-1)}, z_{1:k}) \right ]\\
            & \cdot P \left [ ( X_{1:(k-1)}, Y_{1:(k-1)}, Z_{1:k}) =  (x_{1:(k-1)}, y_{1:(k-1)}, z_{1:k}) \right ] \\
            = & P \left [ ( X_{1:k}, Y_{1:k}, Z_{1:k}) = (x_{1:k}, y_{1:k}, z_{1:k}) \right ] \text{,}
        \end{aligned}
    \end{equation*}
    where (i) is again simply by Bayes rule; (ii) is because $Y_k$ is a random function of only $Z_k$ (up to time $k$) under the null $H_0$ and thus is independent of anything with index smaller or equal to $k$ conditioning on $Z_k$; (iii) is again by Bayes rule; (iv) is by Definition \ref{def:naturalresampling}; and finally (v) is by the previous equation above. Equation \ref{eq:induction_claim} is thus established by induction, as a corollary of which, we also get for any $k \le n$,
    \begin{equation*}
        \tilde{X}_{1:n} | (Y_{1:n}, Z_{1:n}) \overset{\text{d}}{=} X_{1:n} | (Y_{1:n}, Z_{1:n})
    \end{equation*}
    Finally, note that $\tilde{X} \indep X \mid (Y,Z)$. So, conditioning on $(Y,Z)$, $\tilde{X}$ and $ X $ are exchangeable, which means the $p$-value defined in  Equation \ref{eq:p-valueART} is conditionally valid, conditioning on $(Y,Z)$. Since $\mathbb{P}\left ( p < \alpha \mid Y, Z \right ) \le \alpha$ holds conditionally, it also holds marginally.
\end{proof}

\begin{proof}[Proof of Theorem \ref{thm:necessity}]
Note that Assumption \ref{assumption:main} was only utilized once in the proof of Theorem \ref{thm:main}, namely (iv) of Equation \ref{eq:proof_of_validity_long_equation_1}. So upon assuming $(\tilde{X}_{1:k},Y_{1:k}, Z_{1:k}) \overset{\text{d}}{=} (X_{1:k},Y_{1:k}, Z_{1:k})$, we know immediately from Equation \ref{eq:proof_of_validity_long_equation_1} that
\begin{equation*}
\begin{aligned}
    & P \left [ Z_k = z_k \mid (  Y_{1:(k-1)}, Z_{1:(k-1)}) =  ( y_{1:(k-1)}, z_{1:(k-1)}) \right ] \\
    =& P \left [ Z_k =z_k \mid ( X_{1:(k-1)} Y_{1:(k-1)}, Z_{1:(k-1)}) =  (x_{1:(k-1)},  y_{1:(k-1)}, z_{1:(k-1)}) \right ]
\end{aligned}
\end{equation*}
which is exactly Assumption \ref{assumption:main}. 
\end{proof}

\section{Proof of Results Presented in Section \ref{sec:normal_means_model}}
\label{subsection:proof_NMM}
Before proving the main power results, we first state a self-explanatory lemma concerning the effect of taking $B$ to go to infinity, which justifies assuming $B$ to be large enough and ignoring the effect of discrete $p$-values like the one defined in Equation \ref{eq:p-valueART}. Similar proof arguments are made in \citep{peng_sharpfisher}, thus we omit the proof of this lemma. The lemma states that as $B \to \infty$, conditioning on any given values of $(X, \mathbf{Y},\mathbf{Z})$,
\begin{equation*}
    p\text{-value} \coloneq \frac{1}{B+1} \left [1 + \sum_{b=1}^B \mathbbm{1}_{\{T(\mathbf{\tilde{X}^b}, \mathbf{Z}, \mathbf{Y}) \geq T(\mathbf{X}, \mathbf{Z}, \mathbf{Y})\}} \right ] \overset{\text{a.s.}}{\rightarrow} \mathbb{P} \left (  T(\mathbf{\tilde{X}^b}, \mathbf{Z}, \mathbf{Y}) \geq T(\mathbf{X}, \mathbf{Z}, \mathbf{Y})  \mid \mathbf{Y}, \mathbf{Z} \right ) \text{.}
\end{equation*}

\begin{lemma}[Power of ART under $B \to \infty$]
\label{lemma:taking_B_to_go_to_infty}
    For any adaptive sapling procedure $A$ satisfies Definition \ref{def:adaptive} and any test statistic $T$, as we take $B \to \infty$, the asymptotic conditional power of ART (with CRT being an degenerate special case) condition on $(Y,Z)$ is equal to
    \begin{equation*}
        \mathbb{P} \left ( T(\mathbf{X},\mathbf{Y}, \mathbf{Z}) \ge z_{1-\alpha}(T(\tilde{\mathbf{X}},\mathbf{Y}, \mathbf{Z})) \mid \mathbf{Y}, \mathbf{Z} \right ) \text{,}
    \end{equation*}
    while the unconditional (marginal) power is equal to
    \begin{equation*}
        \mathbb{P}_{\mathbf{X}, \tilde{\mathbf{X}}, \mathbf{Y}, \mathbf{Z}} \left ( \mathbb{P} \left ( T(\mathbf{X},\mathbf{Y}, \mathbf{Z}) \ge z_{1-\alpha}(T(\tilde{\mathbf{X}},\mathbf{Y}, \mathbf{Z})) \mid \mathbf{Y}, \mathbf{Z} \right ) \right ) \text{.}
    \end{equation*}
    Note that the joint distribution of $(\mathbf{X}, \tilde{\mathbf{X}}, \mathbf{Y}, \mathbf{Z})$ is implicitly specified by the sampling procedure $A$.
\end{lemma}

\begin{lemma}[Normal Means Model with \iid sampling procedures: Joint Asymptotic Distributions of $\bar{Y}_j$'s, $\tilde{\bar{Y}}_j$'s and $\bar{Y}$ Under the Alternative $\text{H}_1$]
\label{lemma:joint_asymptotic_distribution_iid}
    Define
    \begin{equation*}
        T_{\text{all}} = \left (\tilde{\bar{Y}}_1, \tilde{\bar{Y}}_2, \dots, \tilde{\bar{Y}}_{p-1}, \bar{Y}_1, \bar{Y}_2, \dots, \bar{Y}_{p-1}, \bar{Y} \right )^T \in \mathbb{R}^{2p-1} \text{.}
    \end{equation*}
    Upon assuming the normal means model introduced in Section \ref{sec:normal_means_model}, under the alternative $\text{H}_1$ with $h = h_0 / \sqrt{n}$, as $n \to \infty$,
    \begin{equation*}
        \sqrt{n} \cdot T_{\text{all}}  \overset{\text{d}}{\rightarrow} T_{\text{all}}^{\infty},
    \end{equation*}
    with
    \begin{equation*}
        T_{\text{all}}^{\infty} = 
        \begin{pmatrix}
            G_1 + R + h_0 q_1 \\
            G_2 + R + h_0 q_1 \\
            \cdots \\
            G_{p-1} + R + h_0 q_1 \\
            H_1 + R + h_0 \\
            H_2 + R  \\
            \cdots \\
            H_{p-1} + R  \\
            R
        \end{pmatrix}
        \in \mathbb{R}^{2p-1}\text{,}
    \end{equation*}
    % \begin{equation}
    %     \sqrt{n} \cdot T_{\text{all}}  \overset{\text{d}}{\rightarrow} \mathcal{N} \left ( \mu_{\text{all}},    \Sigma_{\text{all}}     \right ) \text{,}
    % \end{equation}
    % where
    % \begin{equation}
    %     \begin{bmatrix}
    %         \Sigma & 0 & 0 \\
    %         0 & \Sigma & 0 \\
    %         0 & 0 & 1
    %     \end{bmatrix}
    % \end{equation}
    % in which $\Sigma$ was defined in Equation \ref{eq:def_Sigma}.
    where $G \coloneq (G_1, G_2, \dots, G_{p-1})$ and $H \coloneq (H_1, H_2, \dots, H_{p-1})$ both follow the same $(p-1)$ dimensional multivariate Gaussian distribution $\mathcal{N} \left  ( 0, \Sigma \right ) $ and $R$ is a standard normal random variable. Note that $\Sigma$ was defined in the statement of Theorem \ref{thm:power_iid}. Moreover, $G$, $H$ and $R$ are independent.
    % \begin{equation}
    %     \begin{aligned}
    %         T_{\text{A}} = & T_{\text{A}} \left ( G, H, R\right ) \\
    %         \coloneq & \max \Bigg ( \left \{ H_1 +  R + h_0 q_1 (2-q_1) \right \} \cap \{ H_j + R + h_0 q_1 (1-q_1),j = 2,\dots,p-1\} \\
    %         &\cap \left \{ R + h_0 q_1 (1- q_1) - \frac{1}{q_p}  \sum_{i=1}^{p-1} q_j H_j   \right \} \Bigg )
    %     \end{aligned}
    % \end{equation}
    % and
    % \begin{equation}
    %     \begin{aligned}
    %         \tilde{T}_{\text{A}} & =\tilde{T}_{\text{A}} \left ( G, H, R\right ) \\
    %         &\coloneq \max \left ( \left \{ G_j + R + h_0 q_1,j = 1,\dots,p-1 \right \}  \cap \left \{ R + h_0 q_1 - \frac{1}{q_p} \sum_{j=1}^{p-1}q_j G_j \right \} \right ) \text{.}
    %     \end{aligned}
    % \end{equation}
    % Matrices $\Sigma_0$ and $D$ are defined as
    % \begin{equation}
    %     \Sigma_0 \coloneq
    %     \begin{bmatrix}
    %         v(q_1) & -q_1 q_2 & -q_1 q_3 & \cdots & -q_1 q_{p-1}\\
    %         -q_1 q_2 & v(q_2) &  - q_2 q_3 & \cdots & -q_2 q_{p-1} \\
    %         -q_1 q_3 & - q_2 q_3  & v(q_3) & \cdots & -q_3 q_{p-1} \\
    %         \cdots & \cdots & \cdots & \cdots & \cdots \\
    %         -q_1 q_{p-1} & -q_2 q_{p-1} & -q_3 q_{p-1} & \cdots & v(q_{p-1})
    %     \end{bmatrix}
    %     \in \mathbb{R}^{(p-1)\times (p-1)} \text{,}
    % \end{equation}
    % and
    % \begin{equation}
    %     D \coloneq \operatorname{diag}(q_1,q_2,\dots,q_{p-1}) \in \mathbb{R}^{(p-1)\times (p-1)} \text{.}
    % \end{equation}
\end{lemma}

\begin{remark}
Roughly speaking, after removing means, $R$ captures the randomness of $\mathbf{Y}$ being sampled from its marginal distribution; $H$ captures the randomness of sampling $\mathbf{X}$ conditioning on $\mathbf{Y}$; lastly, $G$ captures the randomness of resampling $\tilde{\textbf{X}}$ given $\mathbf{Y}$.
\end{remark}
\begin{remark}
We also note that we do not include characterizing the distribution of $\tilde{\bar{Y}}_{p}$ or $\bar{Y}_{p}$ to avoid stating the convergence in terms of a degenerate multivariate Gaussian distribution since $\bar{Y}_p$ is a deterministic function given $\bar{Y}$ and the remaining $p-1$ means of the other arms.
\end{remark}

\begin{proof}[Proof of Lemma \ref{lemma:joint_asymptotic_distribution_iid}]
% \textbf{Jiaze: all the proofs in this section will be put into the appendix eventually - let's keep them here for now, since it would help in terms of writing.\\}
We first characterize the conditional distribution of $\tilde{\bar{Y}}_j$.
% First of all, we try to characterize the conditional distribution of $\tilde{\bar{Y}}_j$ conditioning on $\mathbf{Y}$ under $\text{H}_1$. For simplicity of notations, we will drop \say{$ \cdot | \mathbf{Y}$} when it is clear we are not taking $\mathbb{E}$ or $\mathbb{P}$ with respect to the randomness of $\mathbf{Y}$. 
For any $j \in \{1,2,\dots, p \}$,
\begin{equation*}
    \begin{aligned}
        \tilde{\bar{Y}}_j& \coloneq \frac{\sum_{i=1}^n Y_i \mathbb{1}_{\tilde{X_i} = j}}{\sum_{i=1}^n\mathbb{1}_{\tilde{X_i} = j}} \\
        % & =  \frac{s(q_j)}{q_j \sqrt{n}} \left [ \frac{ \sum_{i=1}^n Y_i \left ( \mathbb{1}_{\tilde{X_i} = j} - q_j\right ) }{  s(q_j) \sqrt{\sum_{i=1}^n Y_i^2} } + \frac{q_j \sum_{i=1}^{n} Y_i}{ s(q_j) \sqrt{\sum_{i=1}^n Y_i^2} } \right ] \sqrt{\frac{ \sum_{i=1}^n Y_i^2} {n } } \frac{q_j n }{\sum_{i=1}^n\mathbb{1}_{\tilde{X_i} = 1}} \\
        % & = \frac{1}{\sqrt{n}} \left [ \frac{s(q_j)}{q_j}  \frac{ \sum_{i=1}^n Y_i \left ( \mathbb{1}_{\tilde{X_i} = j} - q_j\right ) }{  s(q_j) \sqrt{\sum_{i=1}^n Y_i^2} }   + \sqrt{\frac{n}{ \sum_{i=1}^n Y_i^2}} \frac{ \sum_{i=1}^{n} Y_i}{  \sqrt{n} } \right ] \sqrt{\frac{ \sum_{i=1}^n Y_i^2} {n } } \frac{q_j n }{\sum_{i=1}^n\mathbb{1}_{\tilde{X_i} = 1}}
        % & = \frac{1}{\sqrt{n}} \left [ \frac{s(q_j)}{q_j}  \frac{ \sum_{i=1}^n Y_i \left ( \mathbb{1}_{\tilde{X_i} = j} - q_j\right ) }{  s(q_j) \sqrt{n} }   +  \frac{ \sum_{i=1}^{n} Y_i}{  \sqrt{n} } \right ] \frac{q_j n }{\sum_{i=1}^n\mathbb{1}_{\tilde{X_i} = j}}
        & = \frac{1}{\sqrt{n}} \left [ \frac{1}{q_j}  \frac{ \sum_{i=1}^n Y_i \left ( \mathbb{1}_{\tilde{X_i} = j} - q_j\right ) }{   \sqrt{n} }   +  \frac{ \sum_{i=1}^{n} Y_i}{  \sqrt{n} } \right ] \frac{q_j n }{\sum_{i=1}^n\mathbb{1}_{\tilde{X_i} = j}}
        \text{.}
    \end{aligned}
\end{equation*}
%where $s(q_j) = \operatorname{sd}(\operatorname{Bern}(q_j)) = \operatorname{sd}(\mathbb{1}_{\tilde{X}_i = j}) = \sqrt{q_j(1-q_j)}$.
% To analyse the asymptotic distribution of $ \tilde{\bar{Y}}_1 $, we define the following \say{surrogate},
% \begin{equation}
%     \tilde{J}_1 \coloneq \frac{ \sum_{i=1}^n Y_i \mathbb{1}_{\tilde{X_i} = 1}}{  s(q_1) \sqrt{\sum_{i=1}^n Y_i^2} } \text{.}
% \end{equation}
By Central Limit Theorem, % with probability one (with respect to $\mathbf{Y}$),
since $\operatorname{Var} \left (  Y_i (\mathbb{1}_{\tilde{X}_i = j} - q_j  )  \right ) \to q_j (1 - q_j)$ as $n \to \infty$,
\begin{equation*}
    % \frac{ \sum_{i=1}^n Y_i \left ( \mathbb{1}_{\tilde{X_i} = j} - q_j \right ) }{  s(q_j) \sqrt{\sum_{i=1}^n Y_i^2} } \overset{\text{d}}{\rightarrow} \mathcal{N}(0,1) \text{,}
    \frac{ \sum_{i=1}^n Y_i \left ( \mathbb{1}_{\tilde{X_i} = j} - q_j \right ) }{  \sqrt{q_j (1 - q_j) n} } \overset{\text{d}}{\rightarrow} \mathcal{N}(0,1) \text{,}
\end{equation*}
which together with Slutsky's Theorem and the fact that $q_j n / \sum_{i=1}^n \mathbb{1}_{\tilde{X}_i = j} \to 1$ almost surely gives, % with probability one (with respect to $\mathbf{Y}$),
% \\ Jiaze: this is \say{rigorous} mathematically but not notation-wise - gonna work on it later.
% \begin{equation}
%     J_{1,n} \coloneq   \sqrt{n}  \tilde{\bar{Y}}_1 - \frac{\sum_{i=1}^{n} Y_i}{\sqrt{n}}   \overset{\text{d}}{\rightarrow} \mathcal{N} \left ( 0, \frac{v(q_1)}{q_1^2}\right )\text{,}
% \end{equation}
% Moreover, it is similar and also easier to see that for any $j \ne 1$,
% Similarly, for any $j \ne 1$,
\begin{equation*}
    J_{j,n} \coloneq  \sqrt{n} \tilde{\bar{Y}}_j - \frac{\sum_{i=1}^{n} Y_i}{\sqrt{n}}   \overset{\text{d}}{\rightarrow} \mathcal{N} \left ( 0, \frac{v(q_j)}{q_j^2}\right ) \text{,}
\end{equation*}
where $v(q_j) = \operatorname{Var}(\operatorname{Bern(q_j)}) = \operatorname{Var}(\mathbb{1}_{\tilde{X}_j = 1}) = q_j (1 - q_j)$. 
% In fact, by Multivariate Lindeberg-Feller CLT (reference needed here)
% \begin{equation}
%     J_{-p,n} \overset{\text{d}}{\rightarrow} \mathcal{N} \left ( 0, ?\right ) 
% \end{equation}
% in which $\{ J_{-p,n} = (J_{1,n}, J_{2,n}, \dots, J_{p-1,n})\in \mathbb{R}^{p-1}: n \in \mathbb{N}^{\star} $ are a sequence of $(p-1)$ dimensional random vectors. 
% Building upon these one dimensional asymptotics, now we try to derive their asymptotic joint distribution. Before moving forward, we define a few useful notations.
Additional to these one dimensional asymptotic results, we can also derive their joint asymptotic distribution. Before moving forward, we define a few useful notations,
\begin{equation*}
    \mathbf{J}_{-p,n} \coloneq (J_{1,n}, J_{2,n}, \dots, J_{p-1,n})\in \mathbb{R}^{p-1}
    \text{,}
\end{equation*}
\begin{equation*}
    V_i \coloneq \left ( Y_i (\mathbb{1}_{\tilde{X}_i = 1} - q_1), Y_i (\mathbb{1}_{\tilde{X}_i = 2} - q_2), \cdots, Y_{i} (\mathbb{1}_{\tilde{X}_i = p-1}-q_{p-1})      \right ) \in \mathbb{R}^{p-1} \text{,}
\end{equation*}
\begin{equation*}
    \bar{\Sigma}_n \coloneq \frac{1}{n} \sum_{i=1}^{n}  \operatorname{Var}(V_i) % = \left ( \frac{1}{n} \sum_{i=1}^n Y_i^2 \right ) \Sigma \text{,}
    \text{,}
\end{equation*}
and
\begin{equation}
\label{eq:Sigma_0}
    \Sigma_0 \coloneq \operatorname{Var}\left ( \left (\mathbb{1}_{\tilde{X}_i = 1}, \mathbb{1}_{\tilde{X}_i = 2},\cdots, \mathbb{1}_{\tilde{X}_i = p-1} \right )\right )  =    
        \begin{bmatrix}
            v(q_1) & -q_1 q_2 & -q_1 q_3 & \cdots & -q_1 q_{p-1}\\
            -q_1 q_2 & v(q_2) &  - q_2 q_3 & \cdots & -q_2 q_{p-1} \\
            -q_1 q_3 & - q_2 q_3  & v(q_3) & \cdots & -q_3 q_{p-1} \\
            \cdots & \cdots & \cdots & \cdots & \cdots \\
            -q_1 q_{p-1} & -q_2 q_{p-1} & -q_3 q_{p-1} & \cdots & v(q_{p-1})
        \end{bmatrix}\text{.}
    % \Sigma_0 \coloneq \operatorname{diag}(v(q_1), v(q_2), \dots, v(q_{p-1})) \in \mathbb{R}^{(p-1)\times (p-1)}
    % \text{.}
\end{equation}
By Multivariate Lindeberg-Feller CLT (see for instance \cite{ash2000probability}),
% (reference needed here), % with probability one (with respect to $\mathbf{Y}$),
\begin{equation}
\label{eq:multi_variate_CLT_tilde}
    \sqrt{n} \bar{\Sigma}_n^{-1/2} \left ( \bar{V} - \mathbb{E} \bar{V}  \right ) \overset{\text{d}}{\rightarrow} \mathcal{N} \left ( 0, I_{p-1} \right )  \text{.}
\end{equation}
which further gives
\begin{equation*}
    \sqrt{n} \left ( \bar{V} - \mathbb{E} \bar{V}  \right ) \overset{\text{d}}{\rightarrow} \mathcal{N} \left ( 0, \Sigma_0 \right )
\end{equation*}
because of % with probability one (with respect to $\mathbf{Y}$),
\begin{equation*}
    % \bar{\Sigma}_n \overset{\text{a.s.}}{\rightarrow} \Sigma \text{.}
    \lim_{n \to \infty}\bar{\Sigma}_n =  \Sigma_0 \text{.}
\end{equation*}
Therefore we have
\begin{equation}
\label{eq:J_bf_resampling}
    \mathbf{J}_{-p,n} \overset{\text{d}}{\rightarrow} \mathcal{N} \left ( 0, \Sigma  \right )\text{,}
\end{equation}
where
\begin{equation*}
    \Sigma = D^{-1} \Sigma_0 D^{-1} 
    %= \operatorname{diag}((1- q_1) /q_1, (1- q_2)/q_2,\dots, (1 - q_{p-1}) /q_{p-1})  )
\end{equation*}
with
\begin{equation}
\label{eq:D}
    D = \operatorname{diag}(q_1,q_2,\cdots,q_{p-1}) \in \mathbb{R}^{(p-1)\times (p-1)} \text{.}
\end{equation}
Roughly speaking, this suggests that after removing the shared randomness induced by $\frac{\sum_{i=1}^n Y_i}{\sqrt{n}}$, all the $\sqrt{n} \tilde{\bar{Y}}_j$'s are asymptotically independent and Gaussian distributed.

% Next, we turn to $\bar{Y}_j | Y$. 
Next, we turn to $\bar{Y}_j$. 
% In particular we start with $\bar{Y}_1$. 
Note that in this part we will view $X_i$ as generated from $F_{X|Y}$ after the generation of $Y_i$ according to its marginal distribution. 
% Again, we will drop \say{$\cdot|Y$}. 
The only difference in the observed test statistic and the above is that we have 
\begin{equation*}
    X_i | Y_i \sim \mathcal{M}(q^{\star}_i)
\end{equation*}
with $q^{\star}_i = (q_{i,1}^{\star},q_{i,2}^{\star},\cdots, q_{i,p}^{\star} )$ and
\begin{equation*}
    q_{i,j}^{\star} = \frac{q_j  \mathcal{N} \left (Y_i;\frac{h_0}{\sqrt{n}} \mathbb{1}_{j = 1} ,1 \right )  }{\sum_{k = 1}^p q_k  \mathcal{N} \left (Y_i;\frac{h_0}{\sqrt{n}} \mathbb{1}_{k = 1} ,1 \right )} = \frac{\exp \left [ -\frac{1}{2} \left ( Y_i - \frac{h_0}{\sqrt{n}}\mathbb{1}_{j=1} \right )^2  \right ]}{ \sum_{k = 1}^{p} q_k \exp \left [ -\frac{1}{2} \left ( Y_i - \frac{h_0}{\sqrt{n}}\mathbb{1}_{k=1} \right )^2  \right ]} 
\end{equation*}
instead. Again, Multivariate Lindeberg-Feller CLT gives, % with probability one with respect to $\mathbf{Y}$,
\begin{equation}
\label{eq:multi_variate_CLT_observe}
    \sqrt{n} (\bar{\Sigma}_n^{\star})^{-1/2} \left ( \bar{V}^{\star} - \mathbb{E} \bar{V}^{\star}  \right ) \overset{\text{d}}{\rightarrow} \mathcal{N} \left ( 0, I_{p-1} \right )  \text{,}
\end{equation}
with
\begin{equation*}
    V_i^{\star} \coloneq \left ( Y_i (\mathbb{1}_{X_i = 1} - q_{i,1}^{\star}), Y_i ( \mathbb{1}_{X_i = 2} - q_{i,2}^{\star}), \cdots, Y_{i} (\mathbb{1}_{X_i = p-1} -q_{i,p-1}^{\star})     \right ) \in \mathbb{R}^{p-1} \text{,}
\end{equation*}
\begin{equation*}
    % \bar{\Sigma}_n^{\star} = \frac{1}{n} \sum_{i=1}^n \operatorname{Var}\left ( V_i^{\star} \mid Y \right ) \text{.}
    \bar{\Sigma}_n^{\star} = \frac{1}{n} \sum_{i=1}^n \operatorname{Var}\left ( V_i^{\star} \right ) \text{.}
\end{equation*}
Note that, since $\lim_{n \to \infty} \operatorname{Var}  \left ( Y_i( \mathbb{1}_{X_i = j} - q_{i,j}^{\star}) \right ) = q_j (1 -q_j)$ and $\lim_{n \to \infty} \operatorname{Cov} \left( Y_i( \mathbb{1}_{X_i = j_1} - q_{i,j_1}^{\star}),Y_i( \mathbb{1}_{X_i = j_2} - q_{i,j_2}^{\star}) \right) = -q_{j_1} q_{j_2}$, 
\begin{equation*}
    % \bar{\Sigma}_n^{\star} \overset{\text{a.s.}}{\rightarrow} \Sigma \text{,}
    \lim_{n \to \infty} \bar{\Sigma}_n^{\star} = \Sigma_0 \text{,}
\end{equation*}
which further gives
\begin{equation}
\label{eq:CLT_V_star}
    \sqrt{n} \left ( \bar{V}^{\star} - \mathbb{E} \bar{V}^{\star}  \right ) \overset{\text{d}}{\rightarrow} \mathcal{N} \left ( 0, \Sigma_0 \right )\text{.}
\end{equation}
Similar to $\mathbf{J}$'s, we define $\mathbf{J}^{\star}$'s as well,
\begin{equation*}
    J_{j,n}^{\star} \coloneq  \sqrt{n} \bar{Y}_j - \frac{\sum_{i=1}^{n} q_{i,j}^{\star} Y_i}{ q_j \sqrt{n}} = \frac{\sum_{i=1}^n Y_i \mathbb{1}_{X_i = j}}{ q_j \sqrt{n}}  - \frac{\sum_{i=1}^{n} q_{i,j}^{\star} Y_i}{ q_j \sqrt{n}} + o_p(1)  = \frac{\sqrt{n} \left ( \bar{V}^{\star} \right )_j}{q_j} + o_p(1)\text{.}
\end{equation*}
and
\begin{equation*}
    \mathbf{J}_{-p,n}^{\star} \coloneq (J_{1,n}^{\star}, J_{2,n}^{\star}, \dots, J_{p-1,n}^{\star})\in \mathbb{R}^{p-1} \text{,}
\end{equation*}
% Finally,
% \begin{equation}
%     \left [ \sqrt{n} \mathbb{E} \left ( \bar{V}^{\star} \mid Y \right ) \right ]_j = \frac{\sum_{i=1}^n q_{i,j}^{\star} Y_i}{\sqrt{n}}\text{,}
% \end{equation}
which together with Equation \ref{eq:CLT_V_star} gives
\begin{equation}
\label{eq:J_bf_star}
    \mathbf{J}_{-p,n}^{\star} \overset{\text{d}}{\rightarrow} \mathcal{N} \left ( 0, \Sigma  \right ) \text{.}
\end{equation}
Note that though Equation \ref{eq:J_bf_resampling} and Equation \ref{eq:J_bf_star} are almost exactly the same, it does not suggest $\bar{Y}_j$'s and $\tilde{\bar{Y}}_j$'s have the same asymptotic distribution, since the \say{mean} parts that have been removed actually behave differently, namely $\frac{\sum_{i=1}^n Y_i}{\sqrt{n}}$ and $\frac{\sum_{i=1}^{n} q_{i,j}^{\star} Y_i}{ q_j \sqrt{n}}$, as demonstrated in Lemma \ref{lemma:resamling_lemma}, Lemma \ref{lemma:star_lemma}, Lemma \ref{lemma:star_lemma_j} and Lemma \ref{lemma:star_lemma_joint}. Roughly speaking, under this $\sqrt{n}$ scaling, the randomness that leads to the Gaussian noise part in CLT is the same across them as demonstrated in Equation \ref{eq:J_bf_resampling} and Equation \ref{eq:J_bf_star}, but the Gaussian distribution they are converging to have different means. 
% $\mathbf{J}_{-p,n}$ and $\mathbf{J}_{-p,n}^{\star}$ are actually independent conditioning on $Y$. Theorem \ref{thm:power_iid} is thus established by combining Equation \ref{eq:power_normal_means_B_to_infty}, Equation \ref{eq:J_bf_resampling}, Equation \ref{eq:J_bf_star}, Lemma \ref{lemma:resamling_lemma}, Lemma \ref{lemma:star_lemma} and Lemma \ref{lemma:star_lemma_j}.

Finally, following exactly the same logic, we can further derive the following joint asymptotic distribution of $\mathbf{J}_{-p,n}$, $\mathbf{J}_{-p,n}^{\star}$ and $\frac{\sum_{i=1}^n Y_i}{\sqrt{n}}$. Letting
\begin{equation*}
    \mathbf{J}_{\text{ALL}}  \coloneq \left ( \frac{\sum_{i=1}^n Y_i}{\sqrt{n}}, \mathbf{J}_{-p,n}, \mathbf{J}_{-p,n}^{\star} \right ) \in \mathbb{R}^{2p-1},
\end{equation*}
we have
\begin{equation*}
    \mathbf{J}_{\text{ALL}} \overset{\text{d}}{\rightarrow} \mathcal{N} \left ( 0, \Sigma_{\text{ALL}}  \right ) \coloneq \mathcal{N} \left (0, 
    \begin{bmatrix}
        1 & 0 & 0 \\
        0 & \Sigma & 0 \\
        0 & 0 & \Sigma
    \end{bmatrix}
    \right ) \text{.}
\end{equation*}

\end{proof}

\begin{lemma}
\label{lemma:resamling_lemma}
    As $n \to \infty$,
    \begin{equation*}
        \frac{1}{n} \sum_{i=1}^{n} Y_i^2 \overset{\text{a.s.}}{\rightarrow} 1 \quad \text{and} \quad \frac{\sum_{i=1}^{n} Y_i}{\sqrt{n}} \overset{\text{d}}{\rightarrow} \mathcal{N}(h_0 q_1, 1)\text{.}
    \end{equation*}
\end{lemma}
\begin{proof}
By defining $E_i \coloneq S_i W_i + (1 - S_i) G_i \sim \mathcal{N}(0,1)$, we have
\begin{equation*}
    Y_i = E_i + \frac{S_i h_0}{\sqrt{n}}
\end{equation*}
Note that $E_i$ and $S_i$ are not independent. Thus,
    \begin{equation*}
        \begin{aligned}
            \frac{1}{n} \sum_{i=1}^{n} Y_i^2 &= \frac{1}{n} \sum_{i=1}^{n} \left ( E_i + \frac{S_i h_0}{\sqrt{n}} \right )^2\\
            &= \frac{1}{n} \sum_{i = 1}^{n} E_i^2 + \frac{1}{n^2}  \sum_{i=1}^{n} S_i h_0 + \frac{1}{n^{3/2}} \sum_{i=1}^{n} 2 h_0 E_i S_i\\
            &\overset{\text{a.s.}}{\rightarrow} 1 \text{,}
        \end{aligned}
    \end{equation*}
    since by Law of Large Numbers the last two terms will vanish asymptotically and the first term will converge to $\mathbb{E}(E_i^2) = 1$. Moreover,
    \begin{equation*}
        \begin{aligned}
             \frac{\sum_{i=1}^n Y_i}{ \sqrt{n} } &= \frac{\sum_{i=1}^n E_i}{\sqrt{n}} + h_0 \frac{\sum_{i=1}^n S_i}{n} \\ 
             & \overset{\text{d}}{\rightarrow} \mathcal{N}\left (h_0 q_1, 1\right )\text{,}
        \end{aligned}
    \end{equation*}
    where the last line is obtained by applying CLT to the first term and LLN to the second term.
\end{proof}

\begin{lemma}
\label{lemma:star_lemma}
    As $n \to \infty$,
    % \textbf{Jiaze: This is wrong!}
    % \begin{equation}
    %     \frac{ \sum_{i=1}^{n} q_{i,1}^{\star} Y_i }{q_1\sqrt{n}} \overset{\text{d}}{\rightarrow} \mathcal{N} \left ( h_0 q_1(2 - q_1), 1 \right ) \text{.}
    % \end{equation}
    % \textbf{Jiaze: Intead, it should be}
    \begin{equation*}
        \frac{ \sum_{i=1}^{n} q_{i,1}^{\star} Y_i }{q_1\sqrt{n}} \overset{\text{d}}{\rightarrow} \mathcal{N} \left ( q_1 h_0, 1 \right ) \text{.}
    \end{equation*}
\end{lemma}
\begin{proof}
    We first show 
    % \begin{equation}
    % \label{eq:star_lemma_convergence_of_mean_scaled_by_root_n}
    %     \lim_{n \to \infty} \mathbb{E} \left ( \sqrt{n} q_{i,1}^{\star} Y_i  \right ) = h_0 q_1^2 (2 - q_1)\text{.} 
    % \end{equation}
    \begin{equation}
    \label{eq:star_lemma_convergence_of_mean_scaled_by_root_n}
        \lim_{n \to \infty} \mathbb{E} \left ( \sqrt{n} q_{i,1}^{\star} Y_i  \right ) = h_0 \text{.} 
    \end{equation}
    Recall that $Y_i$ can be seen as a mixture of two normal distributions $\mathcal{N}(0,1)$ and $\mathcal{N} \left ( \frac{h_0}{\sqrt{n}},1 \right )$ with weights $1-q_1$ and $q_1$. Thus $\mathbb{E} \left ( \sqrt{n} q_{i,1}^{\star} Y_i  \right )$ is equal to
    \begin{equation*}
        \begin{aligned}
            \sqrt{n} \int_{\mathbb{R}} \frac{y q_1 e^{-(y - h_0/\sqrt{n})^2/2}}{ q_1 e^{-(y - h_0/\sqrt{n})^2 /2} + (1 - q_1) e^{-y^2/2}} \left [ (1 - q_1) \frac{1}{\sqrt{2 \pi}} e^{-y^2/2} + q_1 \frac{1}{\sqrt{2 \pi}} e^{-(y - h_0/\sqrt{n})^2/2} \right ] \mathrm{d}y 
            \coloneq A_0 + A_1 \text{.}
        \end{aligned}
    \end{equation*}
    Note that with a change of variable $h = h_0/\sqrt{n}$,
    \begin{equation*}
        \begin{aligned}
            \lim_{n \to \infty} A_1 &= \frac{q_1^2 \sqrt{n}}{\sqrt{2 \pi}} \int_{\mathbb{R}} \frac{y e^{-(y - h_0/\sqrt{n})^2/2}}{ q_1 e^{-(y - h_0/\sqrt{n})^2/2} + (1 - q_1) e^{-y^2/2}}    e^{-(y - h_0/\sqrt{n})^2/2} \mathrm{d}y \\
            &= \lim_{h \to 0}  \frac{q_1^2 h_0}{\sqrt{2 \pi}} \left [ \frac{1}{h}  \int_{\mathbb{R}} \frac{y  e^{-(y - h)^2/2}}{ q_1 e^{-(y - h)^2/2} + (1 - q_1) e^{-y^2/2}}    e^{-(y - h)^2/2} \mathrm{d}y  \right ] \\
            &= \frac{q_1^2 h_0}{\sqrt{2 \pi}} \left . \frac{\mathrm{d} \left [ \int_{\mathbb{R}} \frac{y  e^{-(y - h)^2/2}}{ q_1 e^{-(y - h)^2/2} + (1 - q_1) e^{-y^2/2}}    e^{-(y - h)^2/2} \mathrm{d}y  \right ] }{\mathrm{d} h} \right |_{h = 0} \\
            &= \frac{q_1^2 h_0}{\sqrt{2 \pi}}  \int_{\mathbb{R}} \left . \frac{\mathrm{d} \left [  \frac{y  e^{-(y - h)^2/2}}{ q_1 e^{-(y - h)^2/2} + (1 - q_1) e^{-y^2/2}}    e^{-(y - h)^2/2}  \right ] }{\mathrm{d} h} \right |_{h = 0} \mathrm{d}y \\
            &= \frac{q_1^2 h_0}{\sqrt{2 \pi}}   \int_{\mathbb{R}} (2 - q_1) y^2 e^{-y^2/2}\mathrm{d}y \\
            &= h_0 q_1^2 (2 - q_1) \text{.}
        \end{aligned}
    \end{equation*}
    Similarly,
    % \begin{equation}
    %     \lim_{n \to \infty} A_0 = 0\text{.}
    % \end{equation}
    % \textbf{Jiaze: This is wrong! It should be the following instead!}
    \begin{equation*}
        \lim_{n \to \infty} A_0 = h_0 q_1 (1 - q_1)^2\text{.}
    \end{equation*}
    Equation \ref{eq:star_lemma_convergence_of_mean_scaled_by_root_n} is thereby established. Then we compute $ \lim_{n \to \infty}\operatorname{Var}(q_{i,1}^{\star}Y_i)$ using the same strategy.
    \begin{equation}
    \label{eq:star_lemma_convergence_of_variance}
        \begin{aligned}
            \lim_{n \to \infty}\operatorname{Var}(q_{i,1}^{\star}Y_i) 
            &= \lim_{n \to \infty} \left \{ \mathbb{E} \left [ (q_{i,1}^{\star}Y_i)^2 \right ] - \left [ \mathbb{E} (q_{i,1}^{\star}Y_i)\right ]^2 \right \} \\
            &= \lim_{n \to \infty}  \mathbb{E} \left [ (q_{i,1}^{\star}Y_i)^2 \right ] \\
            &= \lim_{n \to \infty} \int_{\mathbb{R}} y^2 \left [ \frac{ q_1 e^{-(y - h_0/\sqrt{n})^2/2}}{ q_1 e^{-(y - h_0/\sqrt{n})^2 /2} + (1 - q_1) e^{-y^2/2}} \right ]^2 \left [ (1 - q_1) \frac{1}{\sqrt{2 \pi}} e^{-y^2/2} + q_1 \frac{1}{\sqrt{2 \pi}} e^{-(y - h_0/\sqrt{n})^2/2} \right ] \mathrm{d}y \\
            &= \int_{\mathbb{R}} \lim_{h \to 0} \left \{ y^2 \left [ \frac{ q_1 e^{-(y - h)^2/2}}{ q_1 e^{-(y - h)^2 /2} + (1 - q_1) e^{-y^2/2}} \right ]^2 \left [ (1 - q_1) \frac{1}{\sqrt{2 \pi}} e^{-y^2/2} + q_1 \frac{1}{\sqrt{2 \pi}} e^{-(y -h)^2/2} \right ]  \right \}\mathrm{d}y\\
            % &\coloneq \lim_{n \to \infty} ( B_0 + B_1 ) \\
            % &= (1 - q_1) q_1^2 + q_1^3 \\
            % &= q_1^2 \text{.}
            &= \int_{\mathbb{R}} q_1^2 \frac{e^{-y^2/2}}{\sqrt{2 \pi}} \mathrm{d} y \\
            &= q_1^2 \text{.}
        \end{aligned}
    \end{equation}
    Combining Equation \ref{eq:star_lemma_convergence_of_mean_scaled_by_root_n} and Equation \ref{eq:star_lemma_convergence_of_variance}, the lemma is thus established by Central Limit Theorem.
\end{proof}
Following exactly the same logic, we have the following parallel lemma for $j \ne 1$. 
\begin{lemma}
\label{lemma:star_lemma_j}
    For $j \ne 1$, as $n \to \infty$,
    % \begin{equation}
    %     \frac{ \sum_{i=1}^{n} q_{i,j}^{\star} Y_i }{q_j\sqrt{n}} \overset{\text{d}}{\rightarrow} \mathcal{N} \left (  h_0 q_1 (1 - q_1) , 1 \right )
    % \end{equation}
    % \textbf{Jiaze: This is wrong! It should be the following instead!}
    \begin{equation*}
        \frac{ \sum_{i=1}^{n} q_{i,j}^{\star} Y_i }{q_j\sqrt{n}} \overset{\text{d}}{\rightarrow} \mathcal{N} \left ( 0 , 1 \right )
    \end{equation*}
\end{lemma}
\begin{proof}
        We first show 
    % \begin{equation}
    %     \lim_{n \to \infty} \mathbb{E} \left ( \sqrt{n} q_{i,j}^{\star} Y_i  \right ) =  h_0 q_j q_1 (1 - q_1)
    %     %h_0 q_1^2 (2 - q_1)
    %     \text{.} 
    % \end{equation}
    % \textbf{Jiaze: This is wrong! It should be the following instead!}
        \begin{equation*}
        \lim_{n \to \infty} \mathbb{E} \left ( \sqrt{n} q_{i,j}^{\star} Y_i  \right ) =  0
        %h_0 q_1^2 (2 - q_1)
        \text{.} 
    \end{equation*}
    Again, recall that $Y_i$ can be seen as a mixture of two normal distributions $\mathcal{N}(0,1)$ and $\mathcal{N} \left ( \frac{h_0}{\sqrt{n}},1 \right )$ with weights $1-q_1$ and $q_1$. Thus $\mathbb{E} \left ( \sqrt{n} q_{i,j}^{\star} Y_i  \right )$ is equal to
    \begin{equation*}
        \begin{aligned}
            \sqrt{n} \int_{\mathbb{R}} \frac{y q_j e^{-y^2/2}}{ q_1 e^{-(y - h_0/\sqrt{n})^2 /2} + (1 - q_1) e^{-y^2/2}} \left [ (1 - q_1) \frac{1}{\sqrt{2 \pi}} e^{-y^2/2} + q_1 \frac{1}{\sqrt{2 \pi}} e^{-(y - h_0/\sqrt{n})^2/2} \right ] \mathrm{d}y 
            \coloneq B_0 + B_1 \text{.}
        \end{aligned}
    \end{equation*}
    With a change of variable $h = h_0/\sqrt{n}$, we have
    \begin{equation*}
        \begin{aligned}
            \lim_{n \to \infty} B_1 &= \frac{q_1 q_j \sqrt{n}}{\sqrt{2 \pi}} \int_{\mathbb{R}} \frac{y e^{-y^2/2}}{ q_1 e^{-(y - h_0/\sqrt{n})^2/2} + (1 - q_1) e^{-y^2/2}}    e^{-(y - h_0/\sqrt{n})^2/2} \mathrm{d}y \\
            &= \lim_{h \to 0}  \frac{q_1 q_j h_0}{\sqrt{2 \pi}} \left [ \frac{1}{h}  \int_{\mathbb{R}} \frac{y  e^{-y^2/2}}{ q_1 e^{-(y - h)^2/2} + (1 - q_1) e^{-y^2/2}}    e^{-(y - h)^2/2} \mathrm{d}y  \right ] \\
            &= \frac{q_1 q_j h_0}{\sqrt{2 \pi}} \left . \frac{\mathrm{d} \left [ \int_{\mathbb{R}} \frac{y  e^{-y^2/2}}{ q_1 e^{-(y - h)^2/2} + (1 - q_1) e^{-y^2/2}}    e^{-(y - h)^2/2} \mathrm{d}y  \right ] }{\mathrm{d} h} \right |_{h = 0} \\
            &= \frac{q_1 q_j h_0}{\sqrt{2 \pi}}  \int_{\mathbb{R}} \left . \frac{\mathrm{d} \left [  \frac{y  e^{-y^2/2}}{ q_1 e^{-(y - h)^2/2} + (1 - q_1) e^{-y^2/2}}    e^{-(y - h)^2/2}  \right ] }{\mathrm{d} h} \right |_{h = 0} \mathrm{d}y \\
            &= \frac{q_1 q_j h_0}{\sqrt{2 \pi}}   \int_{\mathbb{R}} (1 - q_1) y^2 e^{-y^2/2}\mathrm{d}y \\
            &= h_0 q_j q_1 (1 - q_1) \text{.}
        \end{aligned}
    \end{equation*}
    Similarly,
    % \begin{equation}
    %     \lim_{n \to \infty} B_0 = 0\text{.}
    % \end{equation}
    % \textbf{Jiaze: This is wrong! It should be the following instead!}
    \begin{equation*}
        \lim_{n \to \infty} B_0 = -h_0 q_j q_1 ( 1 - q_1)\text{.}
    \end{equation*}
    Finally, we have $\lim_{n \to \infty}\operatorname{Var}(q_{i,1}^{\star}Y_i)  = q_j^2$ as well, which by CLT finishes the proof.
\end{proof}
We can further write down their asymptotic joint distribution. We note that $q_{i,j}^{\star} = \frac{q_j}{q_2} q_{i,2}^{\star}$ deterministically for $j>2$, thus it suffices to only include $j = 1,2$ in the joint asymptotic distribution. 
\begin{lemma}
\label{lemma:star_lemma_joint}
    As $n \to \infty$,
    \begin{equation*}
        \left (\frac{\sum_{i=1}^n Y_i}{ \sqrt{n}}, \frac{ \sum_{i=1}^{n} q_{i,1}^{\star} Y_i }{q_1\sqrt{n} },  \frac{ \sum_{i=1}^{n} q_{i,2}^{\star} Y_i }{q_2\sqrt{n} } \right  )  \overset{\text{d}}{\rightarrow} \mathcal{N} \left ( \mu_3,\Sigma_3 \right ) \text{,}
    \end{equation*}
    where
    \begin{equation*}
        \mu_3 = \left (h_0 q_1,   h_0 q_1(2 - q_1),  h_0 q_1(1 - q_1)\right )^T \in \mathbb{R}^3,
    \end{equation*}
    and $\Sigma_3 \in \mathbb{R}^{3\times 3}$ is equal to 
    \begin{equation*}
        \begin{bmatrix}
            1 & 1 & 1\\
            1 & 1 & 1\\
            1 & 1 & 1 
        \end{bmatrix}
        \text{.}
    \end{equation*}
    In other words, asymptotically these three random variables are completely linearly correlated.
\end{lemma}
\begin{proof}
    By Lemma \ref{lemma:star_lemma}, it suffices to show 
    \begin{equation*}
        \lim_{n \to \infty} \operatorname{Cor} \left ( \frac{\sum_{i=1}^n Y_i}{ \sqrt{n}} , \frac{ \sum_{i=1}^{n} q_{i,1}^{\star} Y_i }{q_1\sqrt{n} }\right ) = \lim_{n \to \infty} \operatorname{Cor} \left ( \frac{\sum_{i=1}^n Y_i}{ \sqrt{n}} , \frac{ \sum_{i=1}^{n} q_{i,2}^{\star} Y_i }{q_2\sqrt{n} }\right ) = \lim_{n \to \infty} \operatorname{Cor} \left ( \frac{ \sum_{i=1}^{n} q_{i,1}^{\star} Y_i }{q_1\sqrt{n} },  \frac{ \sum_{i=1}^{n} q_{i,2}^{\star} Y_i }{q_2\sqrt{n} } \right )  = 1 \text{,}
    \end{equation*}
    which can be established by the following three computations,
    \begin{equation*}
        \begin{aligned}
            \lim_{n \to \infty} \operatorname{Cov} \left ( Y_i, q_{i,1}^{\star} Y_i \right ) 
            &= \lim_{n \to \infty} \mathbb{E} \left ( Y_i \cdot q_{i,1}^{\star} Y_i\right ) \\
            &= \lim_{n \to \infty} \int_{\mathbb{R}} \frac{y^2 q_1 e^{-(y -h_0/\sqrt{n})^2/2}}{ q_1 e^{-(y - h_0/\sqrt{n})^2 /2} + (1 - q_1) e^{-y^2/2}} \left [ (1 - q_1) \frac{1}{\sqrt{2 \pi}} e^{-y^2/2} + q_1 \frac{1}{\sqrt{2 \pi}} e^{-(y - h_0/\sqrt{n})^2/2} \right ] \mathrm{d}y \\
            % &\coloneq \lim_{n \to \infty} ( C_0 + C_1  )
            &= \lim_{h \to 0} \int_{\mathbb{R}} \frac{y^2 q_1 e^{-(y - h)^2/2}}{ q_1 e^{-(y -h)^2 /2} + (1 - q_1) e^{-y^2/2}} \left [ (1 - q_1) \frac{1}{\sqrt{2 \pi}} e^{-y^2/2} + q_1 \frac{1}{\sqrt{2 \pi}} e^{-(y - h)^2/2} \right ] \mathrm{d}y \\
            &= \int_{\mathbb{R}} \lim_{h \to 0} \left \{ \frac{y^2 q_1 e^{-(y - h)^2/2}}{ q_1 e^{-(y -h)^2 /2} + (1 - q_1) e^{-y^2/2}} \left [ (1 - q_1) \frac{1}{\sqrt{2 \pi}} e^{-y^2/2} + q_1 \frac{1}{\sqrt{2 \pi}} e^{-(y - h)^2/2} \right ] \right \} \mathrm{d}y \\
            &= q_1 \int_{\mathbb{R}} y^2 \frac{e^{-y^2/2}}{\sqrt{2 \pi}}\mathrm{d}y \\
            &= q_1 \text{;}
        \end{aligned}
    \end{equation*}
    \begin{equation*}
        \begin{aligned}
            \lim_{n \to \infty} \operatorname{Cov} \left ( Y_i, q_{i,2}^{\star} Y_i \right ) 
            &= \lim_{n \to \infty} \mathbb{E} \left ( Y_i \cdot q_{i,2}^{\star} Y_i\right ) \\
            &= \lim_{n \to \infty} \int_{\mathbb{R}} \frac{y^2 q_2 e^{-y^2/2}}{ q_1 e^{-(y - h_0/\sqrt{n})^2 /2} + (1 - q_1) e^{-y^2/2}} \left [ (1 - q_1) \frac{1}{\sqrt{2 \pi}} e^{-y^2/2} + q_1 \frac{1}{\sqrt{2 \pi}} e^{-(y - h_0/\sqrt{n})^2/2} \right ] \mathrm{d}y \\
            % &\coloneq \lim_{n \to \infty} ( C_0 + C_1  )
            &= \lim_{h \to 0} \int_{\mathbb{R}} \frac{y^2 q_2 e^{-y^2/2}}{ q_1 e^{-(y -h)^2 /2} + (1 - q_1) e^{-y^2/2}} \left [ (1 - q_1) \frac{1}{\sqrt{2 \pi}} e^{-y^2/2} + q_1 \frac{1}{\sqrt{2 \pi}} e^{-(y - h)^2/2} \right ] \mathrm{d}y \\
            &= \int_{\mathbb{R}} \lim_{h \to 0} \left \{ \frac{y^2 q_2 e^{-y^2/2}}{ q_1 e^{-(y -h)^2 /2} + (1 - q_1) e^{-y^2/2}} \left [ (1 - q_1) \frac{1}{\sqrt{2 \pi}} e^{-y^2/2} + q_1 \frac{1}{\sqrt{2 \pi}} e^{-(y - h)^2/2} \right ] \right \} \mathrm{d}y \\
            &= q_2 \int_{\mathbb{R}} y^2 \frac{e^{-y^2/2}}{\sqrt{2 \pi}}\mathrm{d}y \\
            &= q_2 \text{;}
        \end{aligned}
    \end{equation*}
    \begin{equation*}
        \begin{aligned}
            \lim_{n \to \infty} \operatorname{Cov} \left (q_{i,1}^{\star} Y_i, q_{i,2}^{\star} Y_i \right ) 
            &= \lim_{n \to \infty} \mathbb{E} \left ( q_{i,1}^{\star} q_{i,2}^{\star} Y_i^2\right ) \\
            &= \lim_{n \to \infty} \int_{\mathbb{R}} \frac{y^2 q_1 q_2 e^{-(y -h_0/\sqrt{n})^2/2} e^{-y^2/2}}{ \left[ q_1 e^{-(y - h_0/\sqrt{n})^2 /2} + (1 - q_1) e^{-y^2/2} \right ]^2} \left [ (1 - q_1) \frac{1}{\sqrt{2 \pi}} e^{-y^2/2} + q_1 \frac{1}{\sqrt{2 \pi}} e^{-(y - h_0/\sqrt{n})^2/2} \right ] \mathrm{d}y \\
            % &\coloneq \lim_{n \to \infty} ( C_0 + C_1  )
            &= \lim_{h \to 0} \int_{\mathbb{R}} \frac{y^2 q_1 q_2 e^{-(y - h)^2/2} e^{-y^2/2}}{ \left [ q_1 e^{-(y -h)^2 /2} + (1 - q_1) e^{-y^2/2} \right ]^2 } \left [ (1 - q_1) \frac{1}{\sqrt{2 \pi} } e^{-y^2/2} + q_1 \frac{1}{\sqrt{2 \pi}} e^{-(y - h)^2/2} \right ] \mathrm{d}y \\
            &= \int_{\mathbb{R}} \lim_{h \to 0} \left \{ \frac{y^2 q_1q_2 e^{-(y - h)^2/2} e^{-y^2/2}}{ \left [ q_1 e^{-(y -h)^2 /2} + (1 - q_1) e^{-y^2/2} \right ]^2} \left [ (1 - q_1) \frac{1}{\sqrt{2 \pi}} e^{-y^2/2} + q_1 \frac{1}{\sqrt{2 \pi}} e^{-(y - h)^2/2} \right ] \right \} \mathrm{d}y \\
            &= q_1 q_2 \int_{\mathbb{R}} y^2 \frac{e^{-y^2/2}}{\sqrt{2 \pi}}\mathrm{d}y \\
            &= q_1 q_2 \text{.}
        \end{aligned}
    \end{equation*}
\end{proof}

\section{Additional Simulations in Normal-means Model}
\label{Appendix:sim_NMM}
To show that our results presented in Section~\ref{subsec:normal_means_model_takeways} are not sensitive to the initially chosen adaptive parameters and to also further optimize for multiple adaptive procedures $A$ as shown in Algorithm~\ref{algo:ART_procedure}, we create Figure \ref{fig:normal_means_pipeline}. Figure~\ref{fig:normal_means_pipeline} shows the power of the ART using different combinations of the adaptive parameters, $\epsilon$ and reweighting value $t_0$, in three different scenarios of $p$ and $h_0$.

Figure~\ref{fig:normal_means_pipeline} shows that an adaptive procedure with exploration parameter $\epsilon = 0.7$ seems to be a favorable choice across different signal strengths. Additionally, we find that the optimal reweighting parameter $t$ can be different across different scenarios but does not seem to matter largely across the different scenarios. We find that our initially chosen parameter of $\epsilon = 0.5$ in Section~\ref{subsec:normal_means_model_takeways} was not necessarily even the most optimal choice, demonstrating the robustness of the results presented in Section~\ref{subsec:normal_means_model_takeways}

\begin{figure}[H]
\begin{center}
\includegraphics[width=\textwidth]{"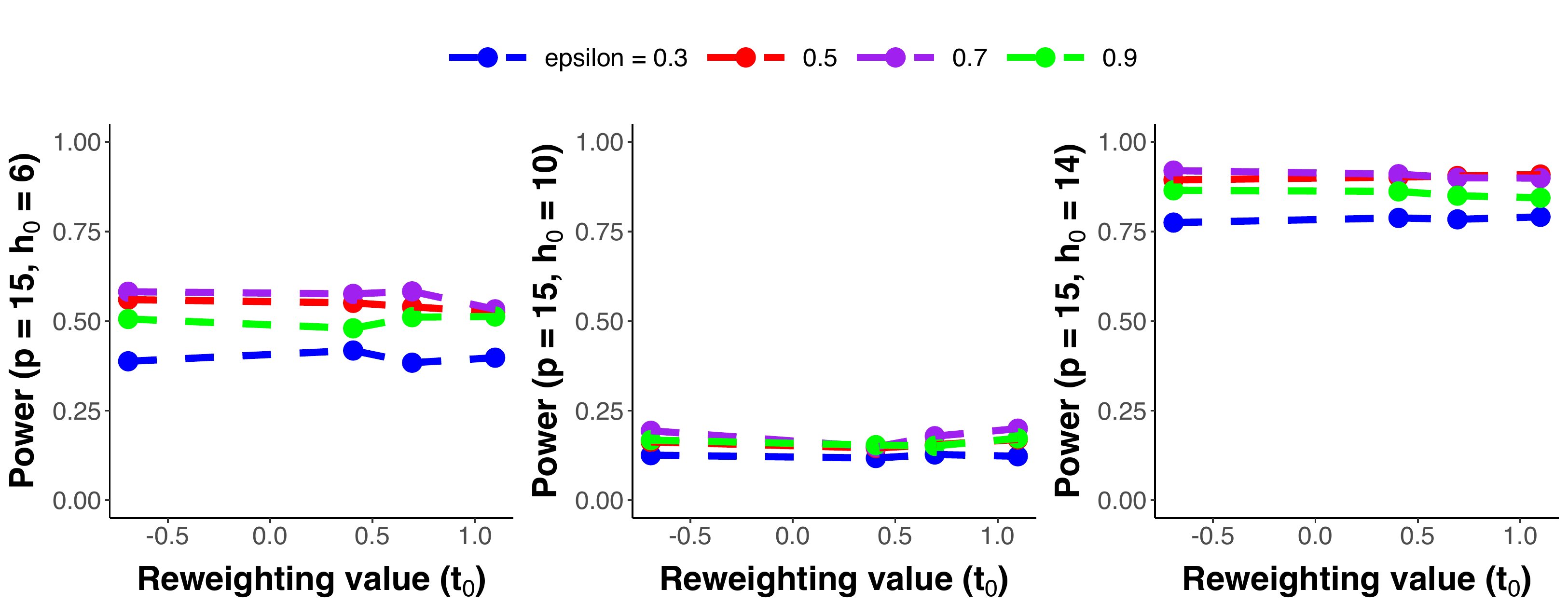"}
\caption{Each panel showcases the power for different exploration parameter $\epsilon$ across different reweighting parameter $t_0$, where $t = t_0/h_0$. The panels differ by different signal strengths $h_0 = 6, 10, 14$ while the number of arms are fixed at $p = 15$.}
\label{fig:normal_means_pipeline}
\end{center}
\end{figure}

\section{Details of Simulations in Conjoint Analysis}
In this section we first give further details of our simulation setup used in Figure~\ref{fig:main_simulation} then give additional simulations demonstrating the robustness of our results presented in Section~\ref{subsection:simulations_conjoint} to a different test statistic. 

\subsection{Setup}
\label{appendix:conjoint_details}
For our simulation setup, $(X,Z)$ each contain one factor with four levels, i.e., $X_t^L, X_t^R,Z_t^L, Z_t^R$ take values $1, 2, 3, 4$. The response model follows a logistic regression model with main effects and interactions on only one specific combination, 
\begin{align*}
\Pr(Y_{t} &= 1 \mid X_t, Z_t) =  \text{logit}^{-1}\bigg[\beta_{X}\mathbbm{1}\{X_t^L = 1, X_t^R \neq 1\} - \beta_{X}\mathbbm{1}\{X_t^L \neq 1, X_t^R = 1\} \\ 
& + \beta_{Z}\mathbbm{1}\{Z_t^L = 1, Z_t^R \neq 1\}  - \beta_{Z}\mathbbm{1}\{Z_t^L \neq 1, Z_t^R = 1\} \\
&+ \beta_{XZ}\mathbbm{1}\{X_t^L = 1, Z_t^L = 2, X_t^R \neq 1, Z_t^R \neq 2 \} - \beta_{XZ}\mathbbm{1}\{X_t^L \neq 1, Z_t^L \neq 2, X_t^R = 1, Z_t^R = 2 \}  \bigg],
\end{align*}
where the first four indicators force main effects $\beta_X, \beta_Z$ of $X$ and $Z$, respectively, on the first levels of each factor and the last two indicators force an interaction effect $\beta_{XZ}$ between the first and second level of factors $X$ and $Z$. For example, $\mathbbm{1}\{X_t^L = 1, Z_t^L = 2, X_t^R \neq 1, Z_t^R \neq 2 \}$ is one if the left profile values of $(X, Z)$ are $(1, 2)$, respectively, but the right profile values of $(X,Z)$ are not $(1, 2)$ simultaneously. We note that the interaction indicator is still one if $(X_t^L, Z_t^L) = (1,2)$ and $(X_t^R, Z_t^R) = (1, 3)$ as long as both $(X_t^L, Z_t^L)$ and $(X_t^R, Z_t^R)$ are not $(1,2)$ simultaneously. For the left plot of Figure~\ref{fig:main_simulation}, $\beta_X = \beta_Z = 0.6$ while $\beta_{XZ} = 0.9$ while we vary the sample size in the $x$-axis. For the right plot of Figure~\ref{fig:main_simulation}, the interaction $\beta_{XZ} = 0$ while we vary $\beta_X = \beta_Z = (0, 0.3, 0.6, 0.9, 1.2)$ in the $x$-axis with a fixed sample size of $n = 1,000$. Lastly, our response model assumes ``no profile order effect'' since all main and interaction effects are repeated symmetrically for the right and left profile (except we shift the sign because $Y = 1$ refers to the left profile being selected). 

To increase power we also incorporate the common ``no profile order effect'' in our test statistic in Equation~\eqref{eq:teststat_conjoint}. When fitting a Lasso logistical regression of $\mathbf{Y}$ with main effects and interaction of $(\mathbf{X}, \mathbf{Z})$, we obtain a separate effect for both the left and right effects. Since the ``no profile order effect'' constraints the left and right effects to be similar, we formally impose the following constraints
\begin{equation}
\begin{aligned}
    \hat\beta_{k} \ & = \ \hat\beta_{k}^{L} \ = \ - \hat\beta_{k}^{R}, \quad
    \hat \gamma_{kl} \  = \ \hat \gamma_{kl}^{L} \ = \ -\hat \gamma_{kl}^{R},
    \end{aligned}\label{eq:constraints}
\end{equation}
where the superscripts $L$ and $R$ denote the left and right profile effects, respectively. To incorporate this symmetry constraint, we split our original $\mathbb{R}^{n \times (4 + 1)}$ data matrix $(\mathbf{X},\mathbf{Z}, \mathbf{Y})$ into a new data matrix with dimension $\mathbb{R}^{2n \times (2 + 1)}$, where the first $n$ rows contain the values for the left profile (and the corresponding $Y$) and the next $n$ rows contain the values for the right profile with new response $1-Y$, \citep{mypaper} shows that this formally imposes the constraints in Equation~\eqref{eq:constraints} by destroying any profile order information in the new data matrix. 

\subsection{Additional Simulations for Conjoint Analysis}
\label{appendix:conjoint_differentTS}
Readers may wonder if the conclusions presented in Section~\ref{subsection:sim_results} are sensitive to the choice of the test statistic. For this reason, we present additional simulation results with a different test statistic under the same simulation setting as that in Figure~\ref{fig:main_simulation}. 

Although the test statistic defined in Equation~\ref{eq:teststat_conjoint} is both natural and based off a similar test statistic used in a recent CRT application of conjoint studies in \citep{mypaper}, another widely used test statistic is the average marginal component effect (AMCE) pioneered by \citet{AMCE}. The AMCE is a non-parametric approach that relies on simple difference-in-means to infer the average marginal component effect of each factor by averaging over the distribution of other factors. \citet{AMCE} shows that the AMCE can be estimated directly through a linear regression of $Y$ on $X$ (when $X$ is independently randomized) to estimate the AMCEs related to the levels of $X$. Consequently, in order to test if a factor matters at all, practitioners test if all the AMCEs related to $X$ are statistically indistinguishable from each other. Because the AMCEs can be directly estimated from a linear regression of $Y$ on $X$, the aforementioned test is equivalent to the $F$-test from a linear regression of $Y$ on $X$. Motivated by this popular practice, we choose the $F$-statistic from a linear regression of $Y$ on $X$ as the new test statistic. We also enforce the no profile order effect constraint by using the same appended data matrix introduced above to enforce the constraints in Equation~\ref{eq:constraints}.

The simulations presented in Figure~\ref{fig:additional_simulation} has an identical setup as that in Figure~\ref{fig:main_simulation} except we change the strength of the signal because the power of the ART and the CRT based on the $F$-test statistic is significantly higher than that of the tests based on the test statistic in Equation~\ref{eq:teststat_conjoint}. The left plot of Figure~\ref{fig:main_simulation} has main effects $\beta_X = \beta_Z = 0.2$ while $\beta_{XZ} = 0.4$ while we vary the sample size similarly $x$-axis by $n = (450, 600, 750, 1,000, 1,300)$. The right plot of Figure~\ref{fig:main_simulation} similarly has no interaction effect ($\beta_{XZ} = 0$) and a fixed sample size of $n = 1,000$, but we vary the main effects $\beta_X = \beta_Z = (0, 0.1, 0.2, 0.3, 0.4)$ on the $x$-axis. 

Although the power difference presented in Figure~\ref{fig:additional_simulation} is not as stark as that shown in Figure~\ref{fig:main_simulation}, Figure~\ref{fig:additional_simulation} still shows that the power of the ART is uniformly higher than that of the CRT. For example when $n = 1,000$ in the left panel, the there is a difference in 8 percentage points (64\% versus 72\%) between the \iid sampling procedure and the adaptive sampling procedure with $\epsilon = 0.25$ (blue). When the main effect is as strong as 0.3 in the right panel, there is a difference in 8 percentage points (65\% versus 73\%) between the \iid sampling procedure and the adaptive sampling procedure with $\epsilon = 0.5$ (red).

\begin{figure}[t]
\begin{center}
\includegraphics[width=\textwidth]{"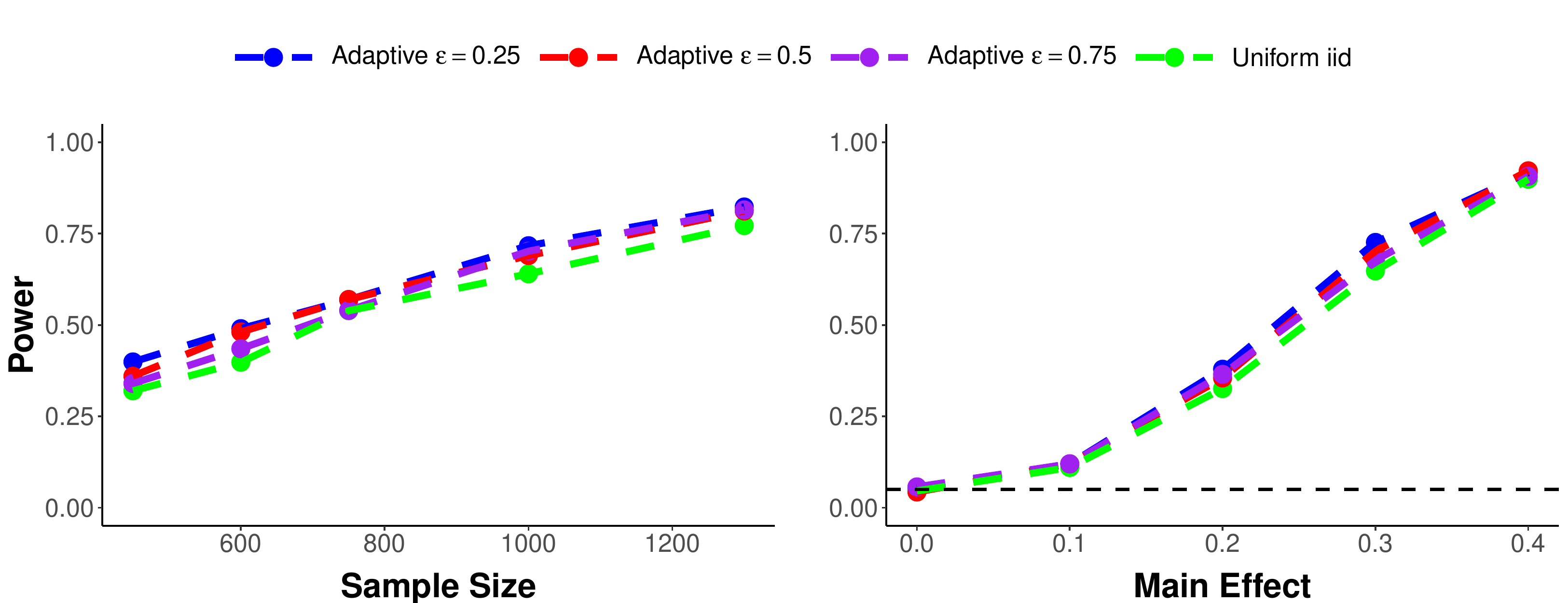"}
\caption{This figure shows additional simulation in the same setting as that for Figure~\ref{fig:main_simulation} except the ART and the CRT uses the $F$-test statistic from a linear regression of $Y$ on $X$. We further change the effect sizes compared to Figure~\ref{fig:main_simulation}. The left plot fixes main effect $\beta_X = \beta_Z = 0.2$ with an interaction effect of $\beta_{XZ} = 0.4$. The right plot varies the main effect of $X$ and $Z$ by $(0, 0.1, 0.2, 0.3, 0.4)$. All other simulation details remain the same as that in Figure~\ref{fig:main_simulation}.} 
\label{fig:additional_simulation}
\end{center}
\end{figure}

\end{document}